\documentclass[letterpaper,12pt]{article}
\usepackage{diagbox}
\usepackage{latexsym,amssymb,amsmath, bm}
\usepackage{mathtools}
\usepackage{amsthm}
\usepackage{lipsum,graphicx,multicol}
\usepackage{setspace,color}
\usepackage[sort&compress]{natbib}
\usepackage{tcolorbox}
\usepackage{subcaption}
\usepackage{booktabs}
\usepackage{float}
\usepackage[normalem]{ulem}
\usepackage{authblk}
\usepackage{leftidx}
\usepackage{dsfont}
\usepackage{geometry}
\usepackage{bbm}
\usepackage[title,titletoc]{appendix}
\usepackage{actuarialangle}
\usepackage{lscape}
\usepackage{afterpage}
\usepackage{relsize}
\usepackage[colorlinks=true,allcolors=blue, linkcolor=red]{hyperref}
\allowdisplaybreaks

\usepackage{lineno}
\usepackage{bbm}

\usepackage{accents}

\def \eqd {\buildrel d \over =}

\def \E {\mathbb{E}}

\def \d {\textrm{d}}

\def \RR {\mathbb{R}}

\def\SD{{\mathrm{SD}}}

\def\Cov{{\mathrm{Cov}}}
\def\Var{{\mathrm{Var}}}
\def\GMD{\mathrm{GMD}}

\def\CV{{\mathrm{CV}}}
\def\Mx{{\mathrm{Mx}}}

\def\Expo{{\mathrm{Exp}}}

\def\IFR{{\mathrm{IFR}}}
\def\DFR{{\mathrm{DFR}}}
\def\IRHR{{\mathrm{IRHR}}}
\def\DRHR{{\mathrm{DRHR}}}
\def\NBUE{{\mathrm{NBUE}}}
\def\NWUE{{\mathrm{NWUE}}}

\def\cx{{\mathrm{cx}}}

\numberwithin{equation}{section} 
\newtheorem{theorem}{Theorem}[section]
\newtheorem{definition}[theorem]{Definition}

\newtheorem{example}[theorem]{Example}
\newtheorem{remark}[theorem]{Remark}
\newtheorem{proposition}[theorem]{Proposition}
\newtheorem{lemma}[theorem]{Lemma}
\newtheorem{corollary}[theorem]{Corollary}


\definecolor{darkread}{rgb}{0.7, 0, 0}

\begin{document}
\onehalfspacing
\title{\Large\textbf{{When Is the Gini Loading More Prudent? Tail Structure and the Ordering of the Standard Deviation and the Gini Mean Difference}}}
%% Alternative, more explicitly actuarial title, should you prefer to signal
%% the insurance content already in the header:
%\title{\Large\textbf{{Tail Structure and the Ordering of the Standard Deviation
%and the Gini Mean Difference, with Applications to Premium Loadings and Capital}}}

\author[]{\small Nawaf Mohammed \thanks{\url{nawaf.mohammed.ac@gmail.com}}}

%\affil[]{\footnotesize Department of Mathematics and Statistics, York University, Toronto, Ontario, Canada}
\date{}
\maketitle
\vspace{-1cm}
\begin{abstract}
The standard deviation (SD) and the Gini mean difference (GMD) are the two canonical measures of variability used to load premiums, set risk margins and allocate capital, yet no universal ordering between them exists. We show that the comparison is \emph{equivalent} to asking whether the coefficient of variation of the spacing $|X-X'|$ generated by two independent copies of the risk exceeds unity, so that the exponential law -- whose spacing is again exponential -- is the universal knife-edge separating the two regimes. Reading the GMD as twice the maxiance, that is, as a second-order \emph{dual} moment in the sense of Yaari's dual theory, the problem becomes an explicit comparison of primal and dual second-order variability. We derive a closed-form representation of the mean excess function of the spacing in terms of the hazard and reverse hazard rates of $X$, and use it to prove that heavy-tailed behavior -- a decreasing hazard rate or an increasing reverse hazard rate -- yields SD dominance, whereas two-sided light tails yield GMD dominance; when both the hazard and the reverse hazard rate are monotone, equality characterizes the exponential law. The sufficient conditions for SD dominance are stable under mixing, those for GMD dominance under convolution, and both under tail truncation, which makes them operational in frailty and collective risk models. We classify the severity, lifetime and frequency distributions of actuarial practice accordingly, quantify the consequences for SD- and Gini-loaded premium principles and for Gini-type tail risk measures, and show that the sign of $\SD-\GMD$ across thresholds furnishes a simple diagnostic for tail aging.

\end{abstract}

\newpage
\section{Introduction}
\label{sec:intro}

An actuary rarely reports a measure of dispersion for its own sake. Variability enters the premium as a safety loading, the technical provisions as a risk margin, and the balance sheet as economic capital, so the choice of dispersion functional is a pricing and solvency decision \citep{Boonen2024,Furman2017,Rockafellar2003}. This paper compares the two canonical candidates and shows that the comparison is settled by the tail structure of the risk.

The \emph{standard deviation} (SD) remains the most widely used. For a random variable $X$ with an independent copy $X'$ it admits the pairwise representation
\begin{equation}
\label{eq:SD}
\SD[X]=\sqrt{\E\big[(X-\E[X])^2\big]} = \sqrt{\frac{1}{2}\,\E\big[(X - X')^2\big]},
\end{equation}
the square root of the average squared distance between two independent realizations of the same risk. Its shortcomings in risk-sensitive settings are well known. Squared deviations disproportionately penalize large observations, so for heavy-tailed losses the SD overstates effective dispersion; more seriously, it presupposes a finite second moment, an assumption that fails outright for the Pareto-type severities routinely fitted to liability, catastrophe and cyber portfolios \citep{Embrechts1997,Klugman2012a}. Its safety loading $\E[X]+\lambda\SD[X]$ is moreover not monotone with respect to the usual stochastic order, and may assign a larger premium to a pointwise smaller risk \citep{Denuit2005}.

The \emph{Gini mean difference} (GMD), introduced by Corrado Gini \citep{Ceriani2011}, replaces squared by absolute deviations,
\begin{equation}
\label{eq:GMD}
\GMD[X] = \E\big[|X - X'|\big],
\end{equation}
requires only a finite first moment, and is therefore well defined on precisely those severities on which the SD is silent \citep{Yitzhaki2002}. It is also structurally distinguished. From $\max(a,b)=\tfrac12(a+b)+\tfrac12|a-b|$,
\begin{equation}
\label{eq:maxiance}
\Mx[X]\;:=\;\E\big[\max(X,X')\big]-\E[X]\;=\;\tfrac{1}{2}\,\GMD[X],
\end{equation}
so the GMD is twice the \emph{maxiance}, the second central \emph{dual} moment of \cite{Eeckhoudt2022}; see \cite{Eeckhoudt2020} and \cite{Castano2025} for the associated calculus and its role in ordering higher-order risks. For $X\ge0$ the dual reading is explicit,
\begin{equation}
\label{eq:dualrep}
\E\big[\max(X,X')\big]=\int_0^{\infty}\big(1-F_X(x)^2\big)\,\d x=\int_0^{\infty} g\big(S_X(x)\big)\,\d x,
\qquad g(p)=2p-p^2 ,
\end{equation}
a Yaari--Wang distortion functional with concave distortion $g$ \citep{Yaari1987,Wang1996}. Consequently the Gini loading $\pi_{\GMD}(X)=\E[X]+\lambda\GMD[X]$ is itself a distortion premium, with distortion $g_\lambda(p)=(1+2\lambda)p-2\lambda p^2$, increasing and concave precisely for $\lambda\in[0,\tfrac12]$; over that range it is coherent, comonotonic additive and consistent with second-order stochastic dominance, unlike its SD-based counterpart \citep{Denuit2005}. This places the present analysis in the mean--Gini tradition of \cite{Yitzhaki1982} and \cite{Muliere1989} and alongside the Gini-type capital rules of \cite{Furman2017} and \cite{Chen2025a}.

Reading \eqref{eq:SD} and \eqref{eq:maxiance} side by side isolates the conceptual content of our problem: the SD is the square root of the second central moment, the GMD is twice the second central \emph{dual} moment, and comparing them asks which of the two calculi -- the $L^2$ calculus of expected utility or the distortion calculus of Yaari's dual theory -- penalizes a given risk more heavily. The stakes are immediate, since for the two mean--deviation principles
\begin{equation}
\label{eq:premia}
\pi_{\SD}(X)=\E[X]+\lambda\,\SD[X]
\qquad\text{and}\qquad
\pi_{\GMD}(X)=\E[X]+\lambda\,\GMD[X],
\end{equation}
both of which arise in contemporary insurance design \citep{Boonen2024}, the ranking is determined for every $\lambda$ simultaneously by the ordering of $\SD[X]$ and $\GMD[X]$; the same applies to risk margins, reinsurance loadings and, after conditioning on a Value-at-Risk threshold, to Gini shortfall and mean--Gini capital allocation. Distributional structure is decisive here, and it separates the two worlds of actuarial modeling. Claim severities are typically heavy tailed with \emph{decreasing} failure rates -- Pareto, Lomax and mixed-exponential models, often as a consequence of unobserved heterogeneity, with lognormal and loggamma severities decreasing in the tail \citep{Klugman2012a,Willmot2011a} -- whereas lifetime and duration models are governed by \emph{increasing} failure rates \citep{Dickson2019}. We show that these two worlds sit on opposite sides of the SD--GMD divide.

The comparison admits an exact reduction. Writing $Y:=|X-X'|$ for the \emph{spacing}, equations \eqref{eq:SD} and \eqref{eq:GMD} give $\SD[X]=\sqrt{\tfrac12\E[Y^2]}$ and $\GMD[X]=\E[Y]$, whence
\begin{equation}
\label{eq:CVcriterion}
\SD[X]\ \ge\ \GMD[X]
\quad\Longleftrightarrow\quad
\Var[Y]\ \ge\ \big(\E[Y]\big)^2
\quad\Longleftrightarrow\quad
\CV[Y]\ \ge\ 1 .
\end{equation}
The ordering is therefore \emph{equivalent} to the spacing being at least as variable, in the coefficient-of-variation sense, as an exponential law with the same mean. This identifies a canonical benchmark: if $X\sim\Expo(\theta)$ then memorylessness gives $Y\sim\Expo(\theta)$, so $\CV[Y]=1$ and $\SD[X]=\GMD[X]=\theta^{-1}$. The exponential law is thus the universal knife-edge, and the two regimes correspond to spacings that are respectively more and less variable than exponential; a sufficient route to either is stop-loss dominance relative to the exponential benchmark, in the sense of Theorem~3.A.55 of \cite{Shaked2007}. Criterion \eqref{eq:CVcriterion} also delimits the scale-free \emph{dispersion ratio} $\mathcal{R}[X]:=\SD[X]/\GMD[X]$, which is invariant under affine transformations. Applying the Cauchy--Schwarz inequality to the covariance representation $\GMD[X]=4\,\Cov\big(X,F_X(X)\big)$ gives $\mathcal{R}[X]\ge\sqrt{3}/2$ for every absolutely continuous risk, with equality exactly for the uniform law \citep[see also][]{LaHaye2019}; in the opposite direction no bound exists, since for a generalized Pareto risk with shape $\xi\uparrow1/2$ the SD diverges while the GMD remains bounded. The ratio therefore sweeps $[\sqrt3/2,\infty)$, anchored at its minimum by the uniform law and pinned at unity by the exponential law, and the present paper asks which structural features of $X$ place it on either side of that anchor.

Our answer is driven by a single analytical object, the mean excess function of the spacing $Y$, for which we obtain a closed-form representation in terms of the hazard and reverse hazard rates of $X$. Two regimes emerge. When either tail of $X$ is heavy -- a decreasing hazard rate on the right, or an increasing reverse hazard rate on the left -- the spacing is new worse than used in expectation and the SD dominates. When both tails are simultaneously light, the spacing is new better than used in expectation and the GMD dominates. One-sided lightness does not suffice, as we show by explicit counterexample. Whenever both the hazard and the reverse hazard rate are monotone the inequalities are strict except at the exponential law and its reflections, which yields a characterization of the equality case. The sufficient conditions for the two regimes are stable under different operations -- those for SD dominance under mixing over a common support, those for GMD dominance under convolution and, in both cases, under tail truncation -- and this duality carries the actuarial message of the paper: heterogeneity in the scale of claim severities drives the comparison toward the SD, whereas the pooling of homogeneous log-concave risks preserves GMD dominance. We classify the severity, lifetime and frequency models of actuarial practice accordingly, document a standard mortality law that satisfies neither set of hypotheses yet obeys the ordering, and show that the sign of $u\mapsto\SD[X_u^{+}]-\GMD[X_u^{+}]$ furnishes an estimable diagnostic for aging in the tail of a loss distribution.

The technical inputs we draw from reliability theory are classical and are used as such: the implications between monotone hazard and reverse hazard rates go back to \cite{Barlow1963}, closure of the DFR class under mixing to \cite{Proschan1963} and \cite{Barlow1975}, and the log-concavity calculus to \cite{Ibragimov1956}, \cite{Prekopa1971}, \cite{An1998} and \cite{Bagnoli2005}. What is new is the object to which this machinery is applied. The spacing representation of the mean excess function, the resulting exact criterion for the primal--dual dispersion comparison, the exponential characterization of the equality case, the actuarial classification and the tail-aging diagnostic have, to the best of our knowledge, no counterpart in the existing literature on either dispersion measure.

The remainder of the paper is organized as follows. Section~\ref{sec:preliminaries} fixes notation, develops the dual-moment interpretation of the GMD, and establishes the exact criterion \eqref{eq:CVcriterion} together with the representation of the mean excess function of the spacing. Sections~\ref{sec:SDgeGMD} and~\ref{sec:SDleGMD} treat the two dominance regimes in turn, each containing the corresponding characterization of the equality case and a catalogue of the actuarial models it covers. Section~\ref{sec:extensions} studies truncated risks $(X\mid X>u)$ and $(X\mid X\le u)$ and derives the tail-aging diagnostic. Section~\ref{sec:discrete} develops the discrete theory and the frequency models of collective risk theory, where the lattice structure tilts the comparison systematically toward the SD. Section~\ref{sec:conclusions} concludes. All proofs are collected in Appendix~A.
%%%%%%%%%%%%%%%%%%%%%%%%%%%%%%%%%%%%%%%%%%%%%%%%%%%%
\section{Preliminaries and Key Propositions}
\label{sec:preliminaries}

\subsection{Standing assumptions and notation}

Throughout the paper, and unless explicitly stated otherwise, $X$ denotes a non-degenerate random variable that is absolutely continuous with respect to Lebesgue measure and possesses a finite second moment, $\E[X^2]<\infty$. We write $f_X(x)$, $F_X(x)$ and $S_X(x)=1-F_X(x)$ for its density, cumulative distribution function (CDF) and decumulative (survival) distribution function (DDF), respectively, and we assume that the support of $X$ is an interval, possibly unbounded, on whose interior $f_X$ is strictly positive. This mild regularity is the minimal requirement under which the hazard-type functions introduced below are everywhere well defined, and it is satisfied by every parametric family considered in this paper. The discrete counterpart of the theory is developed separately in Section~\ref{sec:discrete}, where the probability mass function is denoted by $p_X(x)$ in order to avoid any confusion with the continuous case.

The absolute continuity assumption is used repeatedly and non-trivially, and we record its most consequential implication at the outset. If $X'$ denotes an independent copy of $X$, then $\mathbb{P}(X=X')=0$, so that the \emph{spacing}
\[
Y:=|X-X'|
\]
carries no atom at the origin. As we shall see, this is precisely what guarantees the identity $m_Y(0)=\E[Y]$ for the mean excess function of Definition~\ref{def:meanexcessfunction}, and the failure of that identity in the presence of atoms is the sole source of the correction term that appears in the discrete theory of Section~\ref{sec:discrete}.

With the continuity convention in place, we fix the remaining terminological conventions. The terms {\it increasing} and {\it decreasing} are understood in their non-strict sense. Moreover, when referring to functional properties such as monotonicity or log-convexity/log-concavity, these are always meant to hold on the relevant support of the function, that is, almost surely rather than necessarily point-wise.

We now introduce the following associated functions, which will be used repeatedly thereafter.
\begin{definition}
\label{def:importantfunctions}
~
\begin{itemize}
\item[•] Hazard rate function:
\[
h_X(x)=\dfrac{f_X(x)}{S_X(x)}.
\]
\item[•] Reverse hazard rate function:
\[
r_X(x)=\dfrac{f_X(x)}{F_X(x)}.
\]
\item[•] Residual survival function (for $t\ge0$):
\[
D_X(x,t)=\dfrac{S_X(x+t)}{S_X(x)}.
\]
\item[•] Reverse residual survival function (for $t\ge0$):
\[
C_X(x,t)=\dfrac{F_X(x-t)}{F_X(x)}.
\]
\end{itemize}
\end{definition}
These quantities play a central role in insurance and risk theory. The hazard rate and reverse hazard rate functions are fundamental tools in life insurance and reliability theory for modeling mortality and failure mechanisms, and they also arise naturally in casualty and property insurance when describing claim arrival and loss occurrence processes \citep{Dickson2019,Klugman2012a}. The residual and reverse residual survival functions characterize, respectively, future-life and past-life behavior and are standard instruments in life contingencies and survival analysis. In the language of reliability, monotonicity of $h_X$ delineates the $\IFR$ and $\DFR$ classes, while monotonicity of $r_X$ delineates their reverse analogues, denoted $\IRHR$ and $\DRHR$; these are broadly the classes into which actuarial lifetime and severity models respectively fall, as we document in Sections~\ref{sec:SDgeGMD} and~\ref{sec:SDleGMD}.

Moreover, several well-known equivalences link the monotonicity properties of these functions to structural characteristics of the underlying distribution functions $S_X(x)$ and $F_X(x)$. These equivalences are classical in survival and reliability theory; for comprehensive treatments, see \cite{Shaked2007} and \cite{Barlow1975}, for the log-concavity calculus that underlies them see \cite{An1998} and \cite{Bagnoli2005}, and for introductory actuarial discussions of hazard and survival functions, see \cite{Dickson2019} or \cite{Klugman2012a}. For completeness, and because both directions of each equivalence are invoked in the proofs of our main theorems, the following proposition formalizes them.
\begin{proposition}
\label{prop:equivalences}
The following statements are equivalent:
\begin{itemize}
\item[(A1)] $h_X(x)$ is increasing (decreasing);
\item[(A2)] for each $t\ge0$, $D_X(x,t)$ is decreasing (increasing) in $x$;
\item[(A3)] $S_X(x)$ is log-concave (log-convex).
\end{itemize}
Likewise, the following are equivalent:
\begin{itemize}
\item[(B1)] $r_X(x)$ is increasing (decreasing);
\item[(B2)] for each $t\ge0$, $C_X(x,t)$ is decreasing (increasing) in $x$;
\item[(B3)] $F_X(x)$ is log-convex (log-concave).
\end{itemize}
\end{proposition}
\noindent A proof is provided in Appendix~A.\ref{app:equivalences}.
\newline

\subsection{The Gini mean difference as a dual moment}

Before turning to the comparison itself, we collect the representations of the GMD that organize the analysis. They are individually classical, but read together they identify the precise sense in which the GMD is the dual counterpart of the SD, and they will be used at several points in the sequel. Throughout, we write
\[
X_{(1)}:=\min(X,X')
\qquad\text{and}\qquad
X_{(2)}:=\max(X,X')
\]
for the order statistics of a sample of size two drawn from $F_X$, and we define the \emph{maxiance} of $X$ by
\[
\Mx[X]:=\E\big[X_{(2)}\big]-\E[X].
\]

\begin{proposition}
\label{prop:GMDdual}
Let $X$ be absolutely continuous and non-degenerate with $\E|X|<\infty$, and let $X'$ be an independent copy of $X$. Then the following hold; only part \emph{(iv)} uses absolute continuity.
\begin{itemize}
\item[(i)] The spacing coincides with the sample range, $Y=X_{(2)}-X_{(1)}$, and consequently
\[
\GMD[X]=\E\big[X_{(2)}\big]-\E\big[X_{(1)}\big].
\]
\item[(ii)] The order statistics are balanced about the mean, $\E[X_{(1)}]+\E[X_{(2)}]=2\,\E[X]$, whence
\[
\GMD[X]=2\Big(\E\big[X_{(2)}\big]-\E[X]\Big)=2\Big(\E[X]-\E\big[X_{(1)}\big]\Big)=2\,\Mx[X].
\]
\item[(iii)] The GMD admits the tail representation
\[
\GMD[X]=2\int_{-\infty}^{\infty}F_X(x)\,S_X(x)\,\d x .
\]
\item[(iv)] The GMD admits the covariance representation
\[
\GMD[X]=4\,\Cov\big(X,F_X(X)\big).
\]
\item[(v)] If in addition $X\ge0$, then
\[
\E\big[X_{(2)}\big]=\int_0^{\infty}g\big(S_X(x)\big)\,\d x,
\qquad g(p)=2p-p^2=1-(1-p)^2 ,
\]
so that $\E[X_{(2)}]$ is a distortion functional generated by the concave distortion $g$.
\end{itemize}
\end{proposition}
The proof is given in Appendix~A.\ref{app:GMDdual}.
\newline

Proposition~\ref{prop:GMDdual} repays a moment's reflection. Part (i) grants the GMD an operational reading that is immediately meaningful to an actuary: it is the \emph{expected range} of two independent claims drawn from the same risk class, that is, the average spread one anticipates between two policyholders of identical profile. Part (ii) identifies $\GMD[X]$ as twice the second central dual moment of $X$ in the sense of \cite{Eeckhoudt2022}, placing the present comparison alongside the dual risk apportionment theory of \cite{Eeckhoudt2020} and the higher-order dual orderings of \cite{Castano2025}. Part (iv) is the Lerman--Yitzhaki covariance form that underpins the mean--Gini analysis of \cite{Yitzhaki1982,Yitzhaki2002}, while part (v) exhibits $\E[X_{(2)}]$ as a Yaari--Wang distortion risk measure \citep{Yaari1987,Wang1996}. The SD, by contrast, is generated by the $L^2$ norm of the centered risk. The object of this paper is therefore the comparison of a primal and a dual assessment of second-order variability.

\begin{remark}
\label{rmk:deviationmeasures}
Both measures are coherent deviation measures in the sense of \cite{Rockafellar2003,Rockafellar2006}. For the GMD this follows in one line: if $(X,Z)$ and $(X',Z')$ are independent copies of the same random pair, the triangle inequality gives
\[
\GMD[X+Z]=\E\big|(X+Z)-(X'+Z')\big|\le \E|X-X'|+\E|Z-Z'|=\GMD[X]+\GMD[Z],
\]
while positive homogeneity, translation invariance and strict positivity on non-degenerate risks are immediate from \eqref{eq:GMD}. We note also that the GMD is well defined under the sole assumption $\E|X|<\infty$, since $\E|X-X'|<\infty$ if and only if $\E|X|<\infty$, whereas the SD requires a finite second moment. The standing assumption $\E[X^2]<\infty$ of this paper is thus imposed only so that both measures may be compared at all; it is not innocuous for the heavy-tailed severity models of actuarial practice.
\end{remark}

\subsection{Reduction to the spacing and the exponential benchmark}

We now reduce the comparison of the two measures to a single scalar criterion. Setting $Y=|X-X'|$ in equations \eqref{eq:SD} and \eqref{eq:GMD} yields
\[
\SD[X]=\sqrt{\dfrac{1}{2}\,\E[Y^2]}
\qquad\text{and}\qquad
\GMD[X]=\E[Y],
\]
so that the entire problem is one of comparing the first two moments of the spacing. The following proposition records the resulting criterion, which is an exact equivalence rather than a sufficient condition, together with the universal bound it entails. We shall refer to
\[
\mathcal{R}[X]:=\dfrac{\SD[X]}{\GMD[X]}
\]
as the \emph{dispersion ratio} of $X$; by the affine invariance noted in Section~\ref{sec:intro}, $\mathcal{R}$ is a scale-free and location-free functional of the law of $X$.

\begin{proposition}
\label{prop:spacingcriterion}
Let $X$ be absolutely continuous and non-degenerate with $\E[X^2]<\infty$, and let $Y=|X-X'|$. Then
\begin{equation}
\label{eq:CVreductionprelim}
\SD[X]\ \ge\ \GMD[X]
\quad\Longleftrightarrow\quad
\Var[Y]\ \ge\ \big(\E[Y]\big)^2
\quad\Longleftrightarrow\quad
\CV[Y]\ \ge\ 1,
\end{equation}
and the corresponding equivalences hold with every inequality reversed, and with every inequality replaced by an equality. Moreover,
\[
\mathcal{R}[X]\ \ge\ \dfrac{\sqrt{3}}{2},
\]
with equality if and only if $X$ is uniformly distributed on an interval.
\end{proposition}
The proof is given in Appendix~A.\ref{app:spacingcriterion}.
\newline

Criterion \eqref{eq:CVreductionprelim} states that the SD dominates the GMD precisely when the spacing is at least as variable, in the coefficient-of-variation sense, as an exponential law with the same mean. This identifies a canonical reference point for the entire theory.

\begin{proposition}
\label{prop:expbenchmark}
If $X\sim\Expo(\theta)$ with $\theta>0$, then $Y=|X-X'|\sim\Expo(\theta)$ and
\[
\SD[X]=\GMD[X]=\theta^{-1},
\qquad\text{equivalently}\qquad
\mathcal{R}[X]=1 .
\]
The same conclusion holds for every affine image $aX+b$ with $a\neq0$.
\end{proposition}
The proof is given in Appendix~A.\ref{app:expbenchmark}.
\newline

The exponential law is thus the tipping point of the problem: its memorylessness reproduces the spacing distribution exactly, and the two dominance regimes correspond to spacings that are, respectively, more and less dispersed than exponential. It remains to identify verifiable conditions placing a given risk on one side of this benchmark. The natural vehicle is the mean excess function.

\begin{definition}
\label{def:meanexcessfunction}
For a non-negative random variable $X$ with $\E[X]<\infty$ and for $t\ge0$ such that $S_X(t)>0$, the mean excess function of $X$ is defined by
\[
m_X(t) = \E[X - t \mid X > t]=\dfrac{\int_t^{\infty}S_X(w)\,\d w}{S_X(t)} .
\]
In particular $m_X(0)=\E[X]/S_X(0)$, so that $m_X(0)=\E[X]$ whenever $X$ has no atom at the origin.
\end{definition}

The mean excess function captures the expected residual lifetime, or excess loss, beyond a given threshold and is used extensively in tail analysis and risk modeling. Applied to the spacing $Y=|X-X'|$, absolute continuity of $X$ gives $\mathbb{P}(Y=0)=\mathbb{P}(X=X')=0$, hence $S_Y(0)=1$ and
\begin{equation}
\label{eq:mY0}
m_Y(0)=\E[Y]=\GMD[X].
\end{equation}
Identity \eqref{eq:mY0} is the precise point at which the continuity assumption enters the comparison, and its discrete failure is quantified in Section~\ref{sec:discrete}.

The monotonicity of $m_Y$ relative to its initial value is exactly the classical aging dichotomy, which we recall in the form in which it will be used.

\begin{definition}
\label{def:NBUENWUE}
A non-negative random variable $Y$ with finite mean is said to be \emph{new better than used in expectation} ($\NBUE$) if $m_Y(t)\le m_Y(0)$ for all $t\ge0$ with $S_Y(t)>0$, and \emph{new worse than used in expectation} ($\NWUE$) if $m_Y(t)\ge m_Y(0)$ for all such $t$.
\end{definition}

\begin{proposition}
\label{prop:OrderNonNegative}
Let $Y$ be a non-negative random variable with $\E[Y^2]<\infty$ and no atom at the origin. If $m_Y(t)\ge(\le)\,m_Y(0)$ for all $t\ge0$ with $S_Y(t)>0$, that is, if $Y$ is $\NWUE$ ($\NBUE$), then
\[
\sqrt{\dfrac{1}{2}\,\E[Y^2]}\ge(\le)\,\E[Y].
\]
\end{proposition}
The proof of Proposition \ref{prop:OrderNonNegative} is relegated to Appendix~A.\ref{app:OrderNonNegative}.
\newline

Combining Propositions~\ref{prop:spacingcriterion} and~\ref{prop:OrderNonNegative} yields the criterion in the form we shall use throughout: if the spacing $Y=|X-X'|$ is $\NWUE$ then $\SD[X]\ge\GMD[X]$, while if it is $\NBUE$ then $\SD[X]\le\GMD[X]$. Intuitively, when $m_Y(t)\ge m_Y(0)$ -- for example, when $m_Y(t)$ is increasing -- large excesses become more likely, leading the SD to dominate the GMD; in the opposite case the reverse ordering holds.

This sufficient condition is not an isolated device but an instance of a stochastic dominance relation with respect to the exponential benchmark of Proposition~\ref{prop:expbenchmark}. The following proposition makes that statement precise and, as a by-product, delivers the equality case. It is well known in the theory of stochastic orders; see Theorem~3.A.55 of \cite{Shaked2007} and \cite{Barlow1975}. We include a short self-contained proof because the argument also yields the characterization of the equality case, which we shall exploit in Theorems~\ref{thm:SDgeGMDstrict} and~\ref{thm:SDleGMDstrict}.

\begin{proposition}
\label{prop:cxExponential}
Let $Y$ be a non-negative, non-degenerate random variable with mean $\mu\in(0,\infty)$ and $\E[Y^2]<\infty$, and let $E_{\mu}\sim\Expo(\mu^{-1})$ denote an exponential random variable with the same mean. Then
\[
Y \text{ is } \NBUE \ \Longrightarrow\ Y\le_{\cx}E_{\mu},
\qquad
Y \text{ is } \NWUE \ \Longrightarrow\ Y\ge_{\cx}E_{\mu}.
\]
In particular $\CV[Y]\le1$ in the first case and $\CV[Y]\ge1$ in the second. Moreover, if $Y$ is either $\NBUE$ or $\NWUE$ and $\CV[Y]=1$, then $Y\eqd E_{\mu}$.
\end{proposition}
The proof is given in Appendix~A.\ref{app:cxExponential}.
\newline

\subsection{The mean excess function of the spacing}

To derive results of practical relevance, it is necessary to complement Proposition~\ref{prop:OrderNonNegative} by linking the behavior of $m_Y(t)$ to distributional properties of the underlying variable $X$. Accordingly, we conclude this section by completing this implication chain, expressing the mean excess function of $Y$ in terms of the hazard rate, reverse hazard rate, residual survival, and reverse residual survival functions of $X$. The next proposition establishes this connection and is the analytical engine of the paper.

\begin{proposition}
\label{prop:m_Y_Representation}
Let $Y=|X-X'|$ for two independent and identically distributed random variables $X$ and $X'$. Then the mean excess function $m_Y(t)$ admits the representation
\begin{equation}
\label{eq:m_YRepresentation}
m_Y(t)=\dfrac{\E^{F}\!\left[C_X(X,t)\,h_X(X)^{-1}\right]}{\E^{F}[C_X(X,t)]}
=\dfrac{\E^{S}\!\left[D_X(X,t)\,r_X(X)^{-1}\right]}{\E^{S}[D_X(X,t)]},
\end{equation}
where the expectations $\E^{F}[\cdot]$ and $\E^{S}[\cdot]$ are taken with respect to the measures
\[
\mathrm{d} Q^F(x)=\dfrac{F_X(x)}{\E[F_X(X)]}\mathrm{d} F_X(x)
\quad\text{and}\quad
\mathrm{d} Q^{S}(x)=\dfrac{S_X(x)}{\E[S_X(X)]}\mathrm{d} F_X(x),
\]
respectively.
\end{proposition}
The proof is deferred to Appendix~A.\ref{app:m_Y_Representation}.
\newline

The tilted measures appearing in \eqref{eq:m_YRepresentation} are not merely technical devices; they are the laws of the extreme order statistics of a sample of size two, and hence the dual-moment measures identified in Proposition~\ref{prop:GMDdual}. The following remark makes this explicit.

\begin{remark}
\label{rmk:orderstatistics}
Since $X$ is absolutely continuous, $F_X(X)$ is uniformly distributed on $(0,1)$ and therefore $\E[F_X(X)]=\E[S_X(X)]=\tfrac12$. Consequently
\[
\d Q^F(x)=2F_X(x)f_X(x)\,\d x
\qquad\text{and}\qquad
\d Q^{S}(x)=2S_X(x)f_X(x)\,\d x
\]
are precisely the densities of $X_{(2)}=\max(X,X')$ and $X_{(1)}=\min(X,X')$. Representation \eqref{eq:m_YRepresentation} may therefore be written in the transparent form
\[
m_Y(t)
=\dfrac{\E\!\left[C_X\big(X_{(2)},t\big)\,h_X\big(X_{(2)}\big)^{-1}\right]}{\E\!\left[C_X\big(X_{(2)},t\big)\right]}
=\dfrac{\E\!\left[D_X\big(X_{(1)},t\big)\,r_X\big(X_{(1)}\big)^{-1}\right]}{\E\!\left[D_X\big(X_{(1)},t\big)\right]},
\]
the expectations now being taken under the unmodified law of $X$. The analysis of the spacing is thus conducted entirely under the law of the sample maximum, or dually under that of the sample minimum, which is consistent with the reading of the GMD as a second-order dual moment established in Proposition~\ref{prop:GMDdual}. The reciprocal hazard rate $h_X^{-1}$ and the reverse residual survival function $C_X$ enter as the two competing monotone factors, and the entire theory of Sections~\ref{sec:SDgeGMD} and~\ref{sec:SDleGMD} consists in determining whether these factors are similarly or oppositely ordered.
\end{remark}

With Propositions~\ref{prop:spacingcriterion}, \ref{prop:OrderNonNegative} and~\ref{prop:m_Y_Representation} in hand, the implication chain is complete: monotonicity properties of $h_X$ and $r_X$ control the comonotonicity of the two factors in \eqref{eq:m_YRepresentation}, which controls the position of $m_Y(t)$ relative to $m_Y(0)$, which in turn controls the position of $\CV[Y]$ relative to unity and hence the ordering of the two dispersion measures. Section~\ref{sec:SDgeGMD} exploits this chain in the heavy-tailed direction, and Section~\ref{sec:SDleGMD} in the light-tailed one.
%%%%%%%%%%%%%%%%%%%%%%%%%%%%%%%%%%%%%%%%%%%%%%%%%%%%
\section{SD dominance}
\label{sec:SDgeGMD}

We now traverse the implication chain of Section~\ref{sec:preliminaries} in the heavy-tailed direction, identifying the distributional features that push a risk strictly beyond the exponential boundary of Proposition~\ref{prop:expbenchmark}, and show that the SD dominates precisely in the regime that actuarial severity modeling inhabits.

Before stating the main theorem, we establish a structural proposition that delimits the class of distributions to which such dominance results can apply. Its content is a rigidity statement: monotone aging at one end of a distribution constrains not merely the shape but the very support of the risk.
\begin{proposition}
\label{prop:r_Xh_Ximplications}
~
\begin{itemize}
\item[(i)] If $h_X(x)$ is decreasing then $r_X(x)$ is decreasing as well. Additionally, $X$ must be bounded below and unbounded above.
\item[(ii)] If $r_X(x)$ is increasing then $h_X(x)$ is increasing as well. Additionally, $X$ must be bounded above and unbounded below.
\end{itemize}
\end{proposition}
The result is due to \cite{Barlow1963}; a proof is provided in Appendix~A.\ref{app:r_Xh_Ximplications} for completeness, since the support statements are used in an essential way in Theorem~\ref{thm:SDgeGMDstrict} below.
\newline

Proposition~\ref{prop:r_Xh_Ximplications} establishes that a distribution cannot exhibit heavy tails at both extremes. Specifically, a heavy right tail -- characterized by a decreasing hazard rate \(h_X(x)\), or equivalently a log-convex survival function \(S_X(x)\) -- necessarily implies a light left tail, as reflected by a decreasing reverse hazard rate \(r_X(x)\) or, equivalently, a log-concave distribution function \(F_X(x)\). Conversely, an increasing \(r_X(x)\) renders the left tail heavier, which in turn forces the right tail to be light. The proposition further elucidates this trade-off by linking tail behavior to constraints on the support of the random variable \(X\), conveying that the left tail (respectively, the right tail) is not merely light but non-existent when $h_X(x)$ is decreasing (respectively, when $r_X(x)$ is increasing).

\begin{theorem}
\label{thm:SDgeGMD}
If $h_X(x)$ is decreasing or $r_X(x)$ is increasing, then
\[
\SD[X]\ge \GMD[X].
\]
\end{theorem}
The proof is given in Appendix~A.\ref{app:SDgeGMD}. It rests on the correlation inequality recorded as Lemma~\ref{lem:chebyshev} in Appendix~A.\ref{app:chebyshev}, applied to the two factors $C_X(\cdot,t)$ and $h_X^{-1}$ appearing in the representation \eqref{eq:m_YRepresentation}; see also \cite{Hardy1988}.
\newline

In the language of Section~\ref{sec:preliminaries}, the hypothesis makes the spacing $Y=|X-X'|$ new worse than used in expectation, hence more dispersed than its exponential benchmark in the convex order by Proposition~\ref{prop:cxExponential}: a heavy tail inflates the quadratic penalty of the SD relative to the linear penalty of the GMD.

Setting $t=0$ in the representation \eqref{eq:m_YRepresentation} yields, as a by-product of the proof, an identity of independent interest which closes the loop with the dual-moment reading of Proposition~\ref{prop:GMDdual}.

\begin{corollary}
\label{cor:GMDhazard}
For every absolutely continuous $X$ with $\E|X|<\infty$,
\[
\GMD[X]
=\E^{F}\!\left[h_X(X)^{-1}\right]
=\E^{S}\!\left[r_X(X)^{-1}\right],
\]
or equivalently, in terms of the order statistics of a sample of size two,
\[
\GMD[X]
=\E\!\left[\dfrac{1}{h_X\big(X_{(2)}\big)}\right]
=\E\!\left[\dfrac{1}{r_X\big(X_{(1)}\big)}\right].
\]
\end{corollary}
The short proof is given in Appendix~A.\ref{app:GMDhazard}.
\newline

The dominance of Theorem~\ref{thm:SDgeGMD} is in fact strict, except at a single distributional boundary. The following theorem identifies that boundary and constitutes the first of our characterization results.

\begin{theorem}
\label{thm:SDgeGMDstrict}
Suppose that $h_X(x)$ is decreasing. Then
\[
\SD[X]\ >\ \GMD[X]
\]
unless $h_X(x)$ is constant, in which case $X$ is an exponential random variable up to a location shift and $\SD[X]=\GMD[X]$. Symmetrically, if $r_X(x)$ is increasing, then $\SD[X]>\GMD[X]$ unless $r_X(x)$ is constant, in which case $X$ is a reflected and shifted exponential random variable and $\SD[X]=\GMD[X]$.

Equivalently, within the $\DFR$ and $\IRHR$ classes the exponential law and its affine images are the only risks for which the SD and the GMD coincide. The restriction cannot be relaxed to an arbitrary monotone hazard rate: the law with $S_X(x)=(1-x)^{1/4}$ on $[0,1]$, the reflection of the $\mathrm{Beta}(\tfrac14,1)$ law of Example~\ref{ex:SDleGMD}\emph{(4)}, has the increasing hazard rate $h_X(x)=1/\big(4(1-x)\big)$ and satisfies $\SD[X]=\GMD[X]=4/15$.
\end{theorem}
The proof is given in Appendix~A.\ref{app:SDgeGMDstrict}.
\newline

Importantly, for the purpose of constructing examples it suffices to treat only one of the two monotonicity classes. Indeed, distributions with decreasing $h_X(x)$ and those with increasing $r_X(x)$ are related through a simple reflection: each class is the mirror image of the other. The following remark formalizes this relationship.
\begin{remark}
\label{rmk:h_Xr_X_Mirror}
Let $X_1$ be absolutely continuous and let $c\in\RR$. Define the reflected variable $X_2$ by
\[
X_2 \stackrel{d}{=} c - X_1 .
\]
Then $h_{X_1}(x)$ is increasing (decreasing) if and only if $r_{X_2}(x)$ is decreasing (increasing). This equivalence follows from a reflection argument applied to the density, CDF and DDF, namely,
\[
F_{X_2}(x)=S_{X_1}(c-x),
\quad
f_{X_2}(x)=f_{X_1}(c-x),
\]
and hence
\[
r_{X_2}(x)
= \frac{f_{X_2}(x)}{F_{X_2}(x)}
= \frac{f_{X_1}(c - x)}{S_{X_1}(c - x)}
= h_{X_1}(c - x),
\]
the map $x\mapsto c-x$ reversing the sense of monotonicity. Since $\SD$ and $\GMD$ are both invariant under reflection, the ordering $\SD[X_1]\ge\GMD[X_1]$ holds if and only if $\SD[X_2]\ge\GMD[X_2]$.
\end{remark}

Since the severity and lifetime models of actuarial practice are overwhelmingly right-sided, we restrict the examples to that class; left-sided counterparts follow by reflection.

\begin{example}
\label{ex:SDgeGMD}
In this example, we collect several well-known severity distributions whose hazard rate functions $h_X(x)$ are decreasing, and illustrate the implications of Theorems~\ref{thm:SDgeGMD} and~\ref{thm:SDgeGMDstrict} by comparing the SD and the GMD measures. Throughout we report the scale-free dispersion ratio $\mathcal{R}[X]=\SD[X]/\GMD[X]$ introduced in Section~\ref{sec:preliminaries}, which is the appropriate diagnostic since the difference $\SD[X]-\GMD[X]$ is itself measured in monetary units and degenerates whenever both measures become small.

\begin{itemize}

\item[(1)] Let $X\sim \mathrm{Gamma}(\alpha,1)$ with shape parameter $0<\alpha\le1$.
For this range of $\alpha$, the hazard rate function $h_X(x)$ is decreasing; see, for instance, Example~3.11 of \cite{Klugman2012a}. The SD of \(X\) is given by \(\SD[X]=\sqrt{\alpha}\), and, as established in Proposition~\ref{prop:gammaratio} below, the GMD admits the closed form
\[
\GMD[X]=\dfrac{2\,\Gamma\!\left(\alpha+\tfrac{1}{2}\right)}{\sqrt{\pi}\,\Gamma(\alpha)} .
\]
For $\alpha=1$ this returns $\GMD[X]=1=\SD[X]$, consistently with Proposition~\ref{prop:expbenchmark}, while for $\alpha=\tfrac12$ it returns $\GMD[X]=2/\pi\approx0.637<\SD[X]=1/\sqrt{2}\approx0.707$.
\newline
Figure~\ref{fig:SDdominancepanels} plots the difference of the two measures against \(\alpha\) for this family and for the Weibull family of item~(2).

\begin{figure}[H]
\centering
\begin{subfigure}{0.48\textwidth}
\centering
\includegraphics[width=\textwidth]{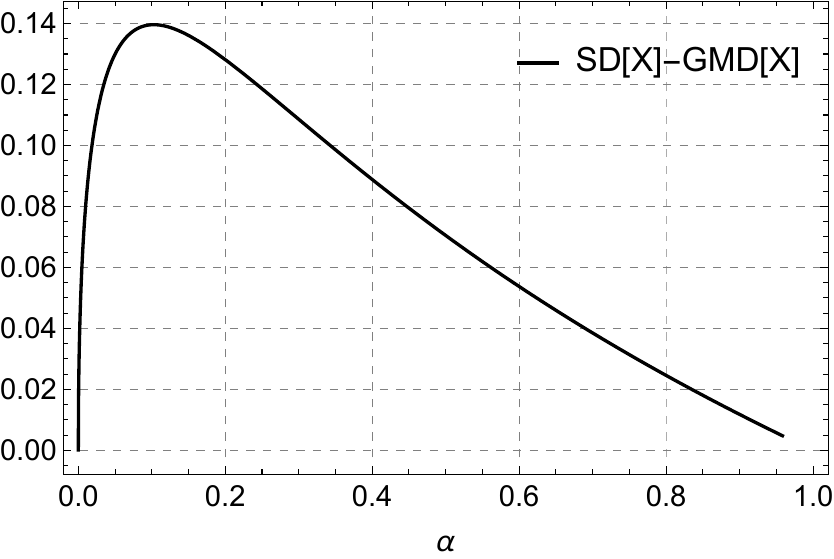}
\caption{Gamma}
\label{fig:GammaSD}
\end{subfigure}
\hfill
\begin{subfigure}{0.48\textwidth}
\centering
\includegraphics[width=\textwidth]{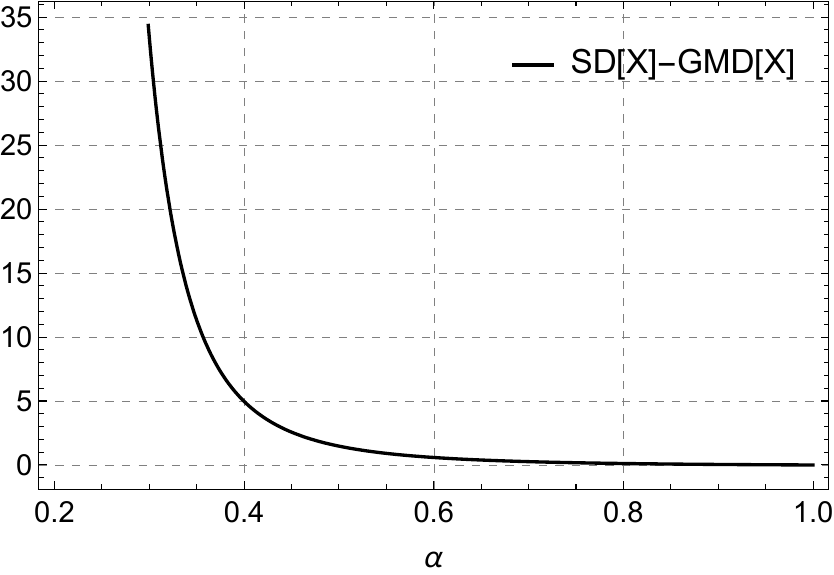}
\caption{Weibull}
\label{fig:WeibullSD}
\end{subfigure}
\caption{Plot of $\SD[X]-\GMD[X]$ against the shape parameter $\alpha$ on the decreasing-hazard range $0<\alpha\le1$, for (a) the Gamma and (b) the Weibull family. In both families the difference vanishes at $\alpha=1$, where the law is exponential, in accordance with Theorem~\ref{thm:SDgeGMDstrict}.}
\label{fig:SDdominancepanels}
\end{figure}

The figure confirms that $\SD[X]\ge \GMD[X]$ for all admissible values of $\alpha$, in accordance with Theorem~\ref{thm:SDgeGMD}, and that the difference vanishes at $\alpha=1$, in accordance with Theorem~\ref{thm:SDgeGMDstrict}. The difference also vanishes as $\alpha\downarrow0$, but this reflects only the collapse of the scale of the risk: the dispersion ratio $\mathcal{R}[X]$ diverges in that limit, as Proposition~\ref{prop:gammaratio} makes precise.

\item[(2)] Suppose that $X\sim\mathrm{Weibull}(\alpha,1)$. Then the hazard rate function is given by
\[
h_X(x)=\dfrac{\alpha x^{\alpha-1}\exp\left(-x^{\alpha}\right)}{\exp\left(-x^{\alpha}\right)}
=\alpha x^{\alpha-1}.
\]
It follows that $h_X(x)$ is decreasing if and only if $0<\alpha\le1$.
For this range of the shape parameter, the SD and the GMD measures are given by
\[
\SD[X]=\sqrt{\Gamma\left(1+\dfrac{2}{\alpha}\right)-\Gamma\left(1+\dfrac{1}{\alpha}\right)^2},
\]
and
\[
\GMD[X]=2 \left(1-2^{-\frac{1}{\alpha}}\right)\Gamma\left(1+\frac{1}{\alpha}\right),
\]
the latter following from $\GMD[X]=2\int_0^\infty F_X S_X\,\d x$ together with $\int_0^\infty\exp(-\lambda x^{\alpha})\,\d x=\lambda^{-1/\alpha}\Gamma(1+1/\alpha)$. At $\alpha=1$ both reduce to unity, the exponential boundary of Theorem~\ref{thm:SDgeGMDstrict}. Plotting the difference of the two quantities as a function of $\alpha$ again reveals the dominance of the SD measure over the GMD.

\item[(3)] Let $X\sim\mathrm{GPD}(\xi,1)$ follow a Generalized Pareto distribution with shape parameter $0\le\xi<1/2$ and unit scale, the canonical model for excess losses in extreme value analysis \citep{Embrechts1997}. The corresponding hazard rate function is
\[
h_X(x)=\dfrac{(1+\xi x)^{-\frac{1}{\xi}-1}}{(1+\xi x)^{-\frac{1}{\xi}}}
=(1+\xi x)^{-1},
\]
which is decreasing for all $0\le\xi<1/2$; the restriction $\xi<1/2$ is required for the SD to exist at all. The two measures are given by
\[
\SD[X]=\dfrac{1}{(1-\xi)\sqrt{1-2\xi}}
\quad \text{and} \quad
\GMD[X]=\dfrac{2}{(1-\xi)(2-\xi)},
\]
so that
\[
\mathcal{R}[X]=\dfrac{2-\xi}{2\sqrt{1-2\xi}} .
\]
Since $(2-\xi)^2-4(1-2\xi)=\xi^2+4\xi\ge0$, with equality if and only if $\xi=0$, we obtain $\mathcal{R}[X]\ge1$ with equality exactly at $\xi=0$, where the GPD degenerates to the exponential law. This is Theorem~\ref{thm:SDgeGMDstrict} in closed form. Note also that $\mathcal{R}[X]\to\infty$ as $\xi\uparrow1/2$: the SD diverges while the GMD converges to $8/3$, a quantitative illustration of the failure of the SD as a variability measure for heavy-tailed severities.

\item[(4)]
Suppose $X$ is a random variable with DDF
\[
S_X(x)=\exp\left(-\frac{\sqrt{\pi }}{2}\, \mathrm{erf}(x)-x\right),\quad x\ge0,
\]
where $\mathrm{erf}(\cdot)$ denotes the error function, defined by
\[
\mathrm{erf}(x)=\dfrac{2}{\sqrt{\pi}}\int_0^x\exp(-t^2)\,\d t .
\]
The hazard rate can be retrieved as
\[
h_X(x)=\exp\left(-x^2\right)+1,
\]
which is decreasing on $x\ge0$ and non-constant. Thus the SD strictly dominates the GMD, as numerical evaluation confirms:
\[
\SD[X]\approx0.757>\GMD[X]\approx0.675,
\qquad \mathcal{R}[X]\approx 1.121 .
\]

Furthermore, since $h_X(x)$ is decreasing, Proposition~\ref{prop:equivalences} shows this to be equivalent to $S_X(x)$ being log-convex.
However, unlike the previous examples, inspection of $\log f_X(x)$ reveals that it is not convex, its second derivative being negative for $x\in[0,0.43)$; in particular $\big(\log f_X\big)''(0)=-1$.
This shows that log-convexity of the DDF (or of the CDF) does not necessarily translate into log-convexity of the density, the log-concave counterpart of this asymmetry being Example~\ref{ex:counterexample_f_X_notlogconcave}.
\end{itemize}
\end{example}

The Gamma family of Example~\ref{ex:SDgeGMD}(1) admits a complete resolution, and we record it separately because it yields an exact characterization rather than a one-sided implication, and because it spans both dominance regimes.

\begin{proposition}
\label{prop:gammaratio}
Let $X\sim\mathrm{Gamma}(\alpha,\theta)$ with shape $\alpha>0$ and scale $\theta>0$. Then
\[
\GMD[X]=\dfrac{2\theta\,\Gamma\!\left(\alpha+\tfrac{1}{2}\right)}{\sqrt{\pi}\,\Gamma(\alpha)},
\qquad
\SD[X]=\theta\sqrt{\alpha},
\qquad
\mathcal{R}[X]=\dfrac{\sqrt{\pi\alpha}\,\Gamma(\alpha)}{2\,\Gamma\!\left(\alpha+\tfrac{1}{2}\right)} .
\]
The map $\alpha\mapsto\mathcal{R}[X]$ is strictly decreasing on $(0,\infty)$, with
\[
\lim_{\alpha\downarrow0}\mathcal{R}[X]=\infty,
\qquad
\mathcal{R}[X]\big|_{\alpha=1}=1,
\qquad
\lim_{\alpha\to\infty}\mathcal{R}[X]=\dfrac{\sqrt{\pi}}{2} .
\]
Consequently,
\[
\SD[X]\ge\GMD[X]\iff \alpha\le1,
\qquad
\SD[X]\le\GMD[X]\iff \alpha\ge1,
\]
with equality if and only if $\alpha=1$.
\end{proposition}
The proof is given in Appendix~A.\ref{app:gammaratio}.
\newline

Table~\ref{tab:SDdominance} collects the closed-form expressions obtained above for the severity models of actuarial practice that fall within the $\DFR$ class; the complementary catalogue for lifetime models is given in Section~\ref{sec:SDleGMD}, while models whose hazard rate is not globally monotone are treated in Section~\ref{sec:extensions}.

\begin{table}[H]
\centering
\small
\setlength{\tabcolsep}{4pt}
\caption{Severity models with decreasing hazard rate, for which $\SD[X]\ge\GMD[X]$ by Theorem~\ref{thm:SDgeGMD}. Scale parameters are normalized to unity; the stated parameter ranges include the moment conditions needed for the SD to exist. For the Burr law $S_X(x)=(1+x^c)^{-k}$, whose hazard rate $ckx^{c-1}/(1+x^c)$ is decreasing for $c\le1$, $\mu_m=\Gamma(k-\tfrac mc)\Gamma(1+\tfrac mc)/\Gamma(k)$ denotes the $m$-th moment, and the GMD entry follows from $\GMD[X]=2\int_0^\infty(S_X-S_X^2)\,\d x$ together with $\int_0^\infty(1+x^c)^{-k}\,\d x=\Gamma(1+\tfrac1c)\Gamma(k-\tfrac1c)/\Gamma(k)$.}
\label{tab:SDdominance}
\begin{tabular}{@{}llll@{}}
\toprule
Family & $\DFR$ range & $\SD[X]$ & $\GMD[X]$ \\
\midrule
Exponential
& --
& $1$
& $1$
\\[4pt]
$\mathrm{Gamma}(\alpha,1)$
& $0<\alpha\le1$
& $\sqrt{\alpha}$
& $\dfrac{2\,\Gamma\!\left(\alpha+\frac{1}{2}\right)}{\sqrt{\pi}\,\Gamma(\alpha)}$
\\[6pt]
$\mathrm{Weibull}(\alpha,1)$
& $0<\alpha\le1$
& $\sqrt{\Gamma\!\left(1+\frac{2}{\alpha}\right)-\Gamma\!\left(1+\frac{1}{\alpha}\right)^{2}}$
& $2\left(1-2^{-\frac{1}{\alpha}}\right)\Gamma\!\left(1+\frac{1}{\alpha}\right)$
\\[6pt]
$\mathrm{GPD}(\xi,1)$
& $0\le\xi<\frac{1}{2}$
& $\dfrac{1}{(1-\xi)\sqrt{1-2\xi}}$
& $\dfrac{2}{(1-\xi)(2-\xi)}$
\\[6pt]
$\mathrm{Burr}(c,k)$
& $0<c\le1,\ ck>2$
& $\sqrt{\mu_2-\mu_1^{2}}$
& $2\mu_1-2\Gamma\!\left(1+\tfrac1c\right)\dfrac{\Gamma\left(2k-\frac1c\right)}{\Gamma(2k)}$
\\
\bottomrule
\end{tabular}
\end{table}

As noted in the introduction, affine transformations preserve the ordering between the SD and the GMD and therefore do not affect SD dominance over the GMD. We conclude this section with a second and considerably more consequential stability property: the $\DFR$ class, and with it the sufficient condition for SD dominance, is preserved under mixing over a common support. This matters because mixing is not a mathematical convenience in insurance but the standard device for representing unobserved heterogeneity across policyholders, and because it generates, from the single most elementary $\DFR$ model, much of the heavy-tailed severity toolkit.

\begin{proposition}
\label{prop:h_xr_X_mixtures}
Let $\{X_{\theta}\}$ be a family of random variables indexed by the values $\theta$ of an independent random variable $\Theta$, and let $X$ denote their mixture. If all $X_\theta$ share the same lower endpoint and all $h_{X_{\theta}}(x)$ are decreasing (respectively, all $X_\theta$ share the same upper endpoint and all $r_{X_{\theta}}(x)$ are increasing), then the hazard rate $h_X(x)$ (respectively, the reverse hazard rate $r_X(x)$) of the mixture $X$ is also decreasing (respectively, increasing).
\end{proposition}
This closure property is classical, being due to \cite{Proschan1963} for life distributions on $[0,\infty)$; see also \cite{Barlow1975}. The common endpoint cannot be dispensed with: an equal mixture of $\Expo(1)$ and $1+\Expo(1)$, both $\DFR$, has a hazard rate that jumps upward at $x=1$. A proof is included in Appendix~A.\ref{app:h_xr_X_mixtures}, the mechanism being the stability of log-convexity under mixtures.
\newline

The consequence for the present comparison is immediate. If each component $X_\theta$, $\theta\in\Theta$, satisfies $\SD[X_{\theta}]\ge\GMD[X_{\theta}]$ by virtue of a decreasing hazard rate with a common lower endpoint, or of an increasing reverse hazard rate with a common upper endpoint, then Proposition~\ref{prop:h_xr_X_mixtures} transmits the relevant monotonicity to the mixture, and Theorem~\ref{thm:SDgeGMD} yields
\[
\SD[X] \ge \GMD[X].
\]

The actuarial reach of this closure property is best appreciated in the canonical frailty model, which we now develop explicitly.

\begin{example}
\label{ex:frailty}
Consider a portfolio of policyholders whose individual claim severities are exponential but whose severity levels differ across the portfolio in a manner not captured by the observable rating factors. Formally, let
\[
X \mid \Lambda=\lambda \ \sim\ \Expo(\lambda),
\qquad
\Lambda\sim\mathrm{Gamma}(\alpha,\beta),
\]
where $\Lambda$ is the frailty and $\beta$ denotes the rate parameter of its Gamma law. Each conditional law has a constant, hence decreasing, hazard rate, so Proposition~\ref{prop:h_xr_X_mixtures} guarantees that the marginal severity $X$ is $\DFR$ and Theorem~\ref{thm:SDgeGMD} yields SD dominance. The marginal law is available in closed form, since
\[
S_X(x)=\E\!\left[\exp(-\Lambda x)\right]=\left(1+\dfrac{x}{\beta}\right)^{-\alpha},
\qquad x\ge0,
\]
which is the Pareto distribution of the second kind, equivalently the $\mathrm{GPD}$ with shape $\xi=1/\alpha$ and scale $\beta/\alpha$. Thus the most widely used heavy-tailed severity model in general insurance arises as a mixture of the most elementary one, and inherits SD dominance from it.

Assuming $\alpha>2$ so that the SD exists, direct computation gives
\[
\SD[X]=\dfrac{\beta\sqrt{\alpha}}{(\alpha-1)\sqrt{\alpha-2}},
\qquad
\GMD[X]=\dfrac{2\alpha\beta}{(\alpha-1)(2\alpha-1)},
\qquad
\mathcal{R}[X]=\dfrac{2\alpha-1}{2\sqrt{\alpha(\alpha-2)}} .
\]
Since $(2\alpha-1)^2-4\alpha(\alpha-2)=4\alpha+1>0$, we obtain $\mathcal{R}[X]>1$ for every admissible $\alpha$, with $\mathcal{R}[X]\to\infty$ as $\alpha\downarrow2$ and $\mathcal{R}[X]\to1$ as $\alpha\to\infty$.
\end{example}

\begin{remark}
\label{rmk:heterogeneity}
Example~\ref{ex:frailty} admits a quantitative reading that we regard as the principal actuarial content of this section. The squared coefficient of variation of the Gamma frailty is $\CV^2[\Lambda]=1/\alpha$, and an expansion of the dispersion ratio gives
\[
\mathcal{R}[X]^2-1=\dfrac{4\alpha+1}{4\alpha(\alpha-2)}=\dfrac{1}{\alpha}+O\!\left(\alpha^{-2}\right)
=\CV^2[\Lambda]+O\!\left(\CV^{4}[\Lambda]\right)
\qquad\text{as }\alpha\to\infty .
\]
To first order, therefore, the excess of the squared dispersion ratio over unity equals the squared coefficient of variation of the unobserved frailty. The gap between the SD and the GMD is, in this precise sense, a measure of unobserved heterogeneity in the portfolio: when the rating structure explains the risk completely, the frailty degenerates, the severity becomes exponential and the two dispersion measures coincide; the more residual heterogeneity remains, the further the SD exceeds the GMD. Since the SD-loaded and Gini-loaded premium principles of \eqref{eq:premia} are ranked by exactly this quantity, the choice of dispersion measure in pricing is, for such portfolios, a choice about how heavily to charge for heterogeneity that the tariff has failed to capture.
\end{remark}
%%%%%%%%%%%%%%%%%%%%%%%%%%%%%%%%%%%%%%%%%%%%%%%%%%%%
\section{GMD dominance}
\label{sec:SDleGMD}
Section~\ref{sec:SDgeGMD} showed that a heavy tail on either side inflates the SD relative to the GMD. It is natural to conjecture the reverse ordering when both tails are light; lightness at one end alone, however, does not suffice, since a single non-light tail may generate enough extreme variability to preserve SD dominance. The following example makes this precise.

\begin{example}
\label{ex:counterexampleGMDdominance}
Suppose that $X$ follows a distribution with DDF
\[
S_X(x)=1-\frac{\mathrm{erfi}(1+x)}{\mathrm{erfi}(2)},\quad x\in[-1,1],
\]
where $\mathrm{erfi}(\cdot)$ denotes the imaginary error function, defined by
\[
\mathrm{erfi}(x)=\dfrac{2}{\sqrt{\pi}}\int_0^x\exp(t^2)\,\mathrm{d} t.
\]
The corresponding hazard rate and reverse hazard rate functions are given by
\[
h_X(x)=\dfrac{2}{\sqrt{\pi }}\,\dfrac{ \exp((1+x)^2)}{ \mathrm{erfi}(2)-\mathrm{erfi}(1+x)},
\quad\text{and}\quad
r_X(x)=\dfrac{2}{\sqrt{\pi}}\,\dfrac{ \exp((1+x)^2)}{\mathrm{erfi}(1+x)}.
\]

Since $\log f_X(x)=(1+x)^2+\mathrm{const}$, we have $f_X'(x)=2(1+x)f_X(x)$, and differentiating the two displays above yields
\[
h^{'}_X(x)=\dfrac{h_X(x)}{S_X(x)}\left(f_X(x)+2(1+x)S_X(x)\right),
\quad\text{and}\quad
r_X^{'}(x)=\dfrac{r_X(x)}{F_X(x)}\left(2(1+x)F_X(x)-f_X(x)\right).
\]
From the first of these we immediately see that $h^{'}_X(x)\ge0$ on $[-1,1]$, since every factor is non-negative there; hence $h_X(x)$ is increasing. In contrast, writing $w=1+x\in[0,2]$, the sign of $r_X'$ is that of
\[
2w\int_0^{w}\exp(t^2)\,\d t-\exp\!\left(w^2\right),
\]
which is negative at $w=0$ and positive at $w=1$, and changes sign once at $w^{*}\approx0.924$. Consequently $r_X(x)$ is not monotone: it decreases on $[-1,-0.076]$ and subsequently increases on $[-0.076,1]$. This behavior is illustrated in Figure~\ref{fig:counterhazards}.

\begin{figure}[H]
\centering
\begin{subfigure}{0.48\textwidth}
\centering
\includegraphics[width=\textwidth]{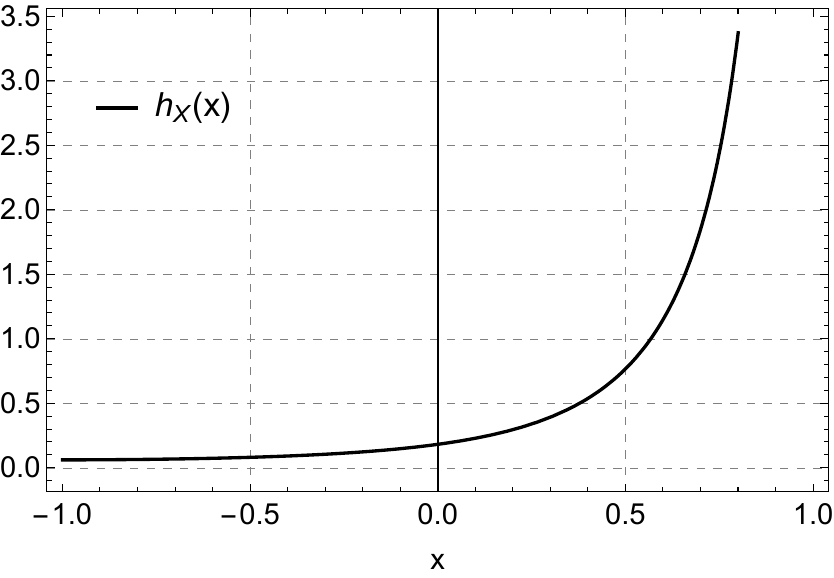}
\caption{Hazard rate $h_X(x)$}
\label{fig:hXcounter}
\end{subfigure}
\hfill
\begin{subfigure}{0.48\textwidth}
\centering
\includegraphics[width=\textwidth]{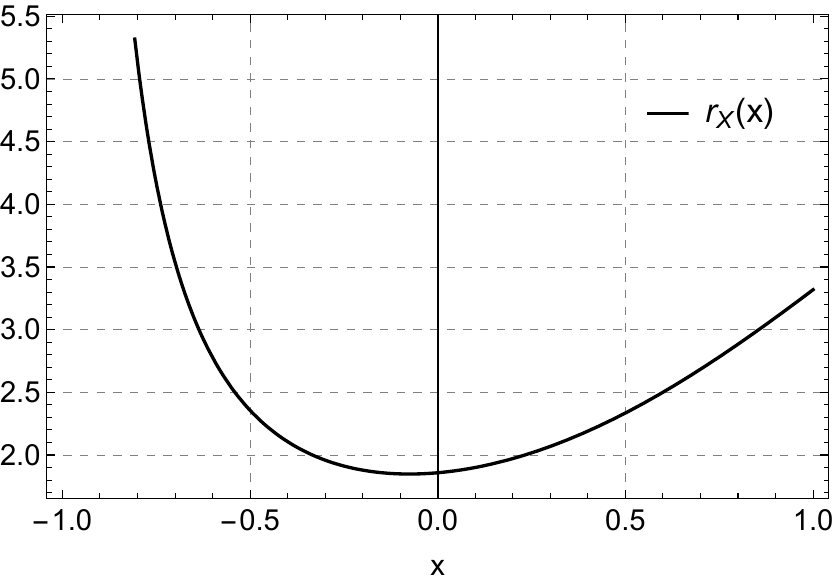}
\caption{Reverse hazard rate $r_X(x)$}
\label{fig:rXcounter}
\end{subfigure}
\caption{Hazard rate and reverse hazard rate of the risk of Example~\ref{ex:counterexampleGMDdominance}. Panel (a) shows that $h_X$ is increasing on $[-1,1]$; panel (b) shows that $r_X$ is not monotone, decreasing on $[-1,-0.076]$ and increasing thereafter. One-sided lightness therefore fails to deliver GMD dominance.}
\label{fig:counterhazards}
\end{figure}

Numerical evaluation of the dispersion measures yields
\[
\SD[X]\approx0.407>\GMD[X]\approx0.402,
\qquad
\mathcal{R}[X]\approx1.013 .
\]
This calculation shows that, despite the increasing nature of the hazard rate, the GMD does not dominate the SD. Hence, an increasing $h_X(x)$ alone is insufficient to ensure the desired ordering. A parallel conclusion can be drawn for the reverse hazard rate. Indeed, by considering the reflected random variable $-X$, we obtain a distribution for which $r_{-X}(x)$ is decreasing, as a consequence of the increasing $h_X(x)$, while $h_{-X}(x)$ inherits the non-monotone behavior of $r_X(x)$. In this case as well, the SD continues to dominate the GMD.

It is worth observing that this example lies entirely outside the reach of Theorem~\ref{thm:SDgeGMD}. Since the support $[-1,1]$ is compact, Proposition~\ref{prop:r_Xh_Ximplications} precludes both a decreasing $h_X(x)$, which would require unboundedness above, and an increasing $r_X(x)$, which would require unboundedness below. SD dominance may therefore occur for bounded risks even though no hypothesis of Theorem~\ref{thm:SDgeGMD} can ever be met by such a risk, a first indication that the conditions developed in this paper are sufficient rather than necessary.
\end{example}

Example~\ref{ex:counterexampleGMDdominance} makes clear that ensuring GMD dominance over the SD requires simultaneous control of both tails of the distribution. This insight naturally motivates our second dominance theorem, which provides a sufficient condition for $\SD[X]\le \GMD[X]$ and serves as the exact counterpart to Theorem~\ref{thm:SDgeGMD}.
\begin{theorem}
\label{thm:SDleGMD}
If $h_X(x)$ is increasing and $r_X(x)$ is decreasing, then
\[
\SD[X]\le \GMD[X].
\]
\end{theorem}
The proof, given in Appendix~A.\ref{app:SDleGMD}, is the mirror image of that of Theorem~\ref{thm:SDgeGMD}: the two factors $C_X(\cdot,t)$ and $h_X^{-1}$ appearing in the representation \eqref{eq:m_YRepresentation} are now oppositely ordered rather than similarly ordered, so that Lemma~\ref{lem:chebyshev}\emph{(ii)} reverses the inequality and renders the spacing $\NBUE$ rather than $\NWUE$.
\newline

As on the heavy-tailed side, the inequality is strict except at a single distributional boundary, and the boundary is again exponential. The following theorem is the counterpart of Theorem~\ref{thm:SDgeGMDstrict}, and its proof isolates a dichotomy that governs the equality case in both regimes simultaneously.

\begin{theorem}
\label{thm:SDleGMDstrict}
Suppose that $h_X(x)$ is increasing and $r_X(x)$ is decreasing. Then
\[
\SD[X]\ <\ \GMD[X]
\]
unless $h_X(x)$ or $r_X(x)$ is constant. In the former case $X$ is an exponential random variable up to a location shift, in the latter a reflected exponential random variable up to a location shift, and in both cases $\SD[X]=\GMD[X]$.

Equivalently, within the class of risks whose hazard and reverse hazard rates are both monotone, the exponential law and its affine images -- reflections included -- are the only ones for which the two dispersion measures coincide.
\end{theorem}
The proof is given in Appendix~A.\ref{app:SDleGMDstrict}, and rests on the equality dichotomy recorded as Lemma~\ref{lem:equalitydichotomy} in Appendix~A.\ref{app:equalitydichotomy}.
\newline

Theorems~\ref{thm:SDgeGMDstrict} and~\ref{thm:SDleGMDstrict} together establish that the exponential law is not merely a convenient benchmark but the unique fixed point of the comparison among risks whose hazard and reverse hazard rates are both monotone. Within the classes of Theorem~\ref{thm:SDgeGMD} the dispersion ratio $\mathcal{R}[X]$ exceeds unity, within the class of Theorem~\ref{thm:SDleGMD} it falls below unity, and in both cases unity is attained only by a memoryless mechanism. Corollary~\ref{cor:masterchar} assembles these statements together with the attainable range of $\mathcal{R}[X]$.

\begin{corollary}
\label{cor:masterchar}
Let $X$ be absolutely continuous and non-degenerate with $\E[X^2]<\infty$.
\begin{itemize}
\item[(i)] If $h_X$ and $r_X$ are both monotone, then
\[
\SD[X]=\GMD[X]
\quad\Longleftrightarrow\quad
X \text{ is exponential up to reflection, scaling and location shift.}
\]
\item[(ii)] Its dispersion ratio satisfies
\[
\dfrac{\sqrt{3}}{2}\ \le\ \mathcal{R}[X]\ <\ \infty ,
\]
the lower bound being attained exactly by the uniform law, and $\mathcal{R}[X]=1$ occurring, when $h_X$ and $r_X$ are both monotone, only at the exponential law.
\item[(iii)] Both dominance regimes and the critical value are realized inside a single one-parameter family: for $X\sim\mathrm{Gamma}(\alpha,\theta)$ the map $\alpha\mapsto\mathcal{R}[X]$ decreases strictly from $\infty$ to $\sqrt{\pi}/2$, the dispersion ratio of the normal law, and crosses unity exactly at $\alpha=1$.
\end{itemize}
\end{corollary}
Part \emph{(i)} combines Theorems~\ref{thm:SDgeGMDstrict} and~\ref{thm:SDleGMDstrict}, whose common equality analysis is Lemma~\ref{lem:equalitydichotomy}; part \emph{(ii)} combines Proposition~\ref{prop:spacingcriterion} with part \emph{(i)}; part \emph{(iii)} is Proposition~\ref{prop:gammaratio}. The three statements together justify regarding the exponential law not as a convenient benchmark but as the unique fixed point of the primal--dual comparison.

By Proposition~\ref{prop:equivalences}, the hypotheses of Theorem~\ref{thm:SDleGMD} amount to log-concavity of both $F_X$ and $S_X$; a single and more tractable sufficient condition is available in terms of the density alone.

\begin{proposition}
\label{prop:logConcave_f_X}
If the density function \(f_X(x)\) of $X$ is log-concave, then both the CDF \(F_X(x)\) and the DDF \(S_X(x)\) are log-concave. Equivalently, the hazard rate \(h_X(x)\) is increasing and the reverse hazard rate \(r_X(x)\) is decreasing. Consequently,
\[
\SD[X] \le \GMD[X].
\]
\end{proposition}
The implication from a log-concave density to log-concave distribution and survival functions is classical, going back to \cite{Ibragimov1956} and \cite{Prekopa1971}; systematic treatments are given by \cite{An1998} and \cite{Bagnoli2005}. A self-contained proof is provided in Appendix~A.\ref{app:logConcave_f_X}.
\newline

\begin{example}
\label{ex:SDleGMD}
This example presents several prominent families of distributions for which the ordering
\(\SD[X] \le \GMD[X]\) follows directly from Proposition~\ref{prop:logConcave_f_X}. In each case,
the conclusion is driven by the log-concavity of the density function, which guarantees the
required monotonicity of the hazard and reverse hazard rates.

\begin{itemize}
\item[(1)] Let \(X \sim \mathrm{Gamma}(\alpha,1)\) with density
\[
f_X(x) = \frac{1}{\Gamma(\alpha)} x^{\alpha-1} \exp(- x), \qquad x \ge 0.
\]
Taking logarithms and computing the second derivative yields
\[
\left(\log f_X(x)\right)^{''}=\dfrac{1-\alpha}{x^2}.
\]
Thus, for \(\alpha \ge 1\), the density is log-concave. This implies an increasing hazard rate
\(h_X(x)\) (see Example~3.11 in \cite{Klugman2012a}) and a decreasing reverse hazard rate
\(r_X(x)\). By Proposition~\ref{prop:gammaratio},
\[
\SD[X]=\sqrt{\alpha}
\qquad\text{and}\qquad
\GMD[X]=\dfrac{2\,\Gamma\!\left(\alpha+\tfrac12\right)}{\sqrt{\pi}\,\Gamma(\alpha)},
\]
so that the ordering may be verified analytically for every $\alpha\ge1$, the dispersion ratio decreasing strictly from $1$ at $\alpha=1$ to the Gaussian value $\sqrt{\pi}/2$ as $\alpha\to\infty$. Panel (a) of Figure~\ref{fig:GMDdominancepanels} plots the difference of the two measures against \(\alpha\) and confirms GMD dominance throughout, the difference vanishing at $\alpha=1$ in accordance with Theorem~\ref{thm:SDleGMDstrict}.

\begin{figure}[H]
\centering
\begin{subfigure}{0.48\textwidth}
\centering
\includegraphics[width=\textwidth]{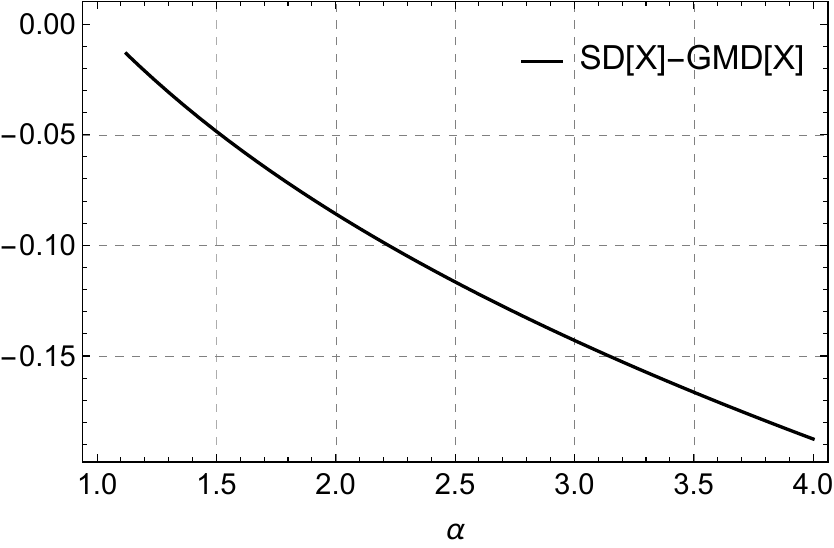}
\caption{Gamma}
\label{fig:GammaGMD}
\end{subfigure}
\hfill
\begin{subfigure}{0.48\textwidth}
\centering
\includegraphics[width=\textwidth]{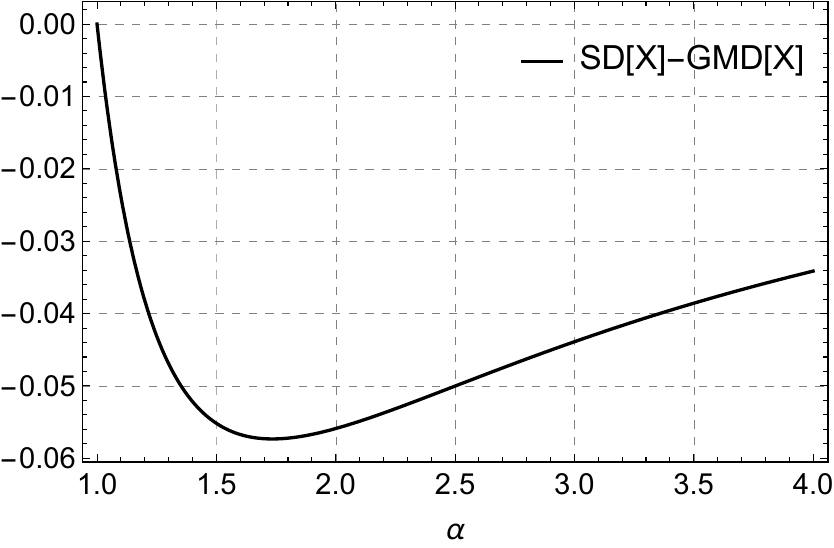}
\caption{Weibull}
\label{fig:WeibullGMD}
\end{subfigure}
\caption{Plot of $\SD[X]-\GMD[X]$ against the shape parameter $\alpha$ on the log-concave range $\alpha\ge1$, for (a) the Gamma and (b) the Weibull family. In both families the difference vanishes at $\alpha=1$, in accordance with Theorem~\ref{thm:SDleGMDstrict}.}
\label{fig:GMDdominancepanels}
\end{figure}

\item[(2)] Suppose \(X \sim \mathrm{Weibull}(\alpha,1)\) with density 
\[
f_X(x) = \alpha x^{\alpha-1} \exp(- x^{\alpha}), \qquad x \ge 0.
\]
The second derivative of the logarithm of the density is given by
\[
\left(\log f_X(x)\right)^{''}=\frac{(1-\alpha)\left(\alpha x^{\alpha}+1\right)}{x^2}.
\]
Hence, the density is log-concave whenever \(\alpha \ge 1\). Explicit expressions for
\(\SD[X]\) and \(\GMD[X]\) are provided in Example~\ref{ex:SDgeGMD}. Their difference
comparison in panel (b) of Figure~\ref{fig:GMDdominancepanels} illustrates the dominance of the GMD over the SD, as
predicted by Proposition~\ref{prop:logConcave_f_X}, with equality at the exponential value $\alpha=1$.

\item[(3)] If \(X\) is standard normal, then its density
\[
f_X(x) = \frac{1}{\sqrt{2\pi}}
\exp\!\left(-\frac{x^2}{2}\right),\quad x\in\RR,
\]
is log-concave on \(\mathbb{R}\). Consequently, both \(F_X(x)\) and \(S_X(x)\) are
log-concave, and Proposition~\ref{prop:logConcave_f_X} implies
\(\SD[X] \le \GMD[X]\). Direct computation yields
\(\SD[X]=1\) and \(\GMD[X]=\dfrac{2}{\sqrt{\pi}}\), whence $\mathcal{R}[X]=\sqrt{\pi}/2\approx0.886$.

\item[(4)] Let \(X \sim \mathrm{Beta}[\alpha,\beta]\) with density
\[
f_X(x) = \dfrac{1}{\mathrm{B}(\alpha,\beta)}\,x^{\alpha-1}(1-x)^{\beta-1},
\quad x\in[0,1].
\]
Then
\[
\left(\log f_X(x)\right)^{''}
=\dfrac{1-\alpha}{x^2}+\dfrac{1-\beta}{(1-x)^2}.
\]
If \(\alpha,\beta \ge 1\), the density is log-concave, implying
\(\SD[X] \le \GMD[X]\). For simplicity, let \(\beta=1\). In this case,
\[
\SD[X]=\sqrt{\frac{\alpha}{(\alpha+1)^2(\alpha+2)}}
\quad \mathrm{and} \quad
\GMD[X]=\frac{2\alpha}{2\alpha^2+3\alpha+1},
\]
and a straightforward algebraic calculation gives
\[
\GMD[X]^2-\SD[X]^2
=\frac{\alpha(4\alpha-1)}{(\alpha+1)^2(\alpha+2)(2\alpha+1)^2},
\]
which is non-negative for all \(\alpha \ge \dfrac{1}{4}\), and in particular for
\(\alpha \ge 1\). The gap between the two thresholds is instructive: on the range $\alpha\in[\tfrac14,1)$ the density is log-convex and the hazard rate is not monotone, so that neither Proposition~\ref{prop:logConcave_f_X} nor Theorem~\ref{thm:SDleGMD} applies, and yet GMD dominance persists. Our conditions are therefore genuinely sufficient rather than necessary, a point we return to below.

\item[(5)] A logistic random variable \(X\) has density
\[
f_X(x) = \frac{\exp(-x)}{(1+\exp(-x))^2}, \qquad x \in \mathbb{R}.
\]
The second derivative of its log-density is
\[
\left(\log f_X(x)\right)^{''}
=-\frac{2 \exp(x)}{\left(\exp(x)+1\right)^2},
\]
which is strictly negative, establishing log-concavity. Here the identity $F_X(x)S_X(x)=f_X(x)$ holds pointwise, so that Proposition~\ref{prop:GMDdual}\emph{(iii)} gives $\GMD[X]=2\int_{-\infty}^{\infty}f_X(x)\,\d x=2$ without further computation. As a result,
\[
\SD[X]=\dfrac{\pi}{\sqrt{3}}\approx1.814 < \GMD[X]=2,
\qquad \mathcal{R}[X]=\dfrac{\pi}{2\sqrt{3}}\approx0.907 .
\]

\item[(6)] Let \(X\) follow a Laplace distribution with scale $b>0$, with density
\[
f_X(x)=\dfrac{1}{2b}\exp\!\left(-\dfrac{|x|}{b}\right),\qquad x\in\RR .
\]
Since $\log f_X(x)=-|x|/b+\mathrm{const}$ is concave, the density is log-concave and Proposition~\ref{prop:logConcave_f_X} applies. Exploiting the symmetry of $F_XS_X$ about the origin,
\[
\GMD[X]=4\int_0^{\infty}\left(1-\tfrac12e^{-x/b}\right)\tfrac12e^{-x/b}\,\d x=\dfrac{3b}{2},
\qquad
\SD[X]=b\sqrt{2},
\]
so that $\mathcal{R}[X]=2\sqrt{2}/3\approx0.943$. This example is of independent interest, since the Laplace law has a constant hazard rate on $[0,\infty)$ and a constant reverse hazard rate on $(-\infty,0]$; it shows that constancy of $h_X$ on one half-line together with constancy of $r_X$ on the complementary half-line is not enough to produce equality, which by Theorem~\ref{thm:SDleGMDstrict} requires global constancy.

\item[(7)] Let \(X\) follow the Gompertz law, the canonical model of adult mortality, with hazard rate $h_X(x)=B\exp(cx)$ for $B,c>0$ and
\[
S_X(x)=\exp\!\left(-\dfrac{B}{c}\big(\exp(cx)-1\big)\right),
\qquad
f_X(x)=B\exp(cx)\,S_X(x),\qquad x\ge0 .
\]
Then $\log f_X(x)=\log B+cx-\tfrac{B}{c}(\exp(cx)-1)$, whence
\[
\left(\log f_X(x)\right)^{''}=-Bc\exp(cx)<0 ,
\]
so the density is log-concave for every admissible parameter pair and Proposition~\ref{prop:logConcave_f_X} yields GMD dominance unconditionally. For the illustrative calibration $B=10^{-4}$ and $c=0.09$, numerical evaluation gives
\[
\SD[X]\approx14.00<\GMD[X]\approx15.26,
\qquad
\mathcal{R}[X]\approx0.917 .
\]
The interpretation is actuarially transparent: for a Gompertz lifetime, the average absolute difference between the ages at death of two independent lives exceeds the standard deviation of a single lifetime, so a Gini-loaded reserve is the more conservative of the two.
\end{itemize}
\end{example}

The uniform distribution on an interval, whose density is the log-concave indicator of a convex set, deserves separate mention: it satisfies $\SD[X]/\GMD[X]=\sqrt{3}/2$ and therefore attains the sharp lower bound of Proposition~\ref{prop:spacingcriterion}; see also \cite{LaHaye2019}. It is thus the most GMD-dominant of all absolutely continuous risks, and it sits at the opposite extreme of the dispersion scale from the heavy-tailed severities of Section~\ref{sec:SDgeGMD}. Table~\ref{tab:GMDdominance} collects the closed-form expressions obtained above, complementing Table~\ref{tab:SDdominance}.

\begin{table}[H]
\centering
\small
\setlength{\tabcolsep}{4pt}
\caption{Distributions with log-concave density, including the lifetime and duration models of Example~\ref{ex:SDleGMD}, for which $\SD[X]\le\GMD[X]$ by Proposition~\ref{prop:logConcave_f_X}. Scale parameters are normalized to unity.}
\label{tab:GMDdominance}
\begin{tabular}{@{}lllll@{}}
\toprule
Family & Log-concavity range & $\SD[X]$ & $\GMD[X]$ & $\mathcal{R}[X]$ \\
\midrule
$\mathrm{Uniform}(0,1)$
& always
& $\dfrac{1}{\sqrt{12}}$
& $\dfrac{1}{3}$
& $\dfrac{\sqrt{3}}{2}\approx0.866$
\\[6pt]
$\mathrm{Normal}(0,1)$
& always
& $1$
& $\dfrac{2}{\sqrt{\pi}}$
& $\dfrac{\sqrt{\pi}}{2}\approx0.886$
\\[6pt]
$\mathrm{Logistic}(0,1)$
& always
& $\dfrac{\pi}{\sqrt{3}}$
& $2$
& $\dfrac{\pi}{2\sqrt{3}}\approx0.907$
\\[6pt]
$\mathrm{Laplace}(0,1)$
& always
& $\sqrt{2}$
& $\dfrac{3}{2}$
& $\dfrac{2\sqrt{2}}{3}\approx0.943$
\\[6pt]
$\mathrm{Gamma}(\alpha,1)$
& $\alpha\ge1$
& $\sqrt{\alpha}$
& $\dfrac{2\,\Gamma\!\left(\alpha+\frac12\right)}{\sqrt{\pi}\,\Gamma(\alpha)}$
& $\dfrac{\sqrt{\pi\alpha}\,\Gamma(\alpha)}{2\,\Gamma\!\left(\alpha+\frac12\right)}$
\\[6pt]
$\mathrm{Weibull}(\alpha,1)$
& $\alpha\ge1$
& $\sqrt{\Gamma\!\left(1+\frac{2}{\alpha}\right)-\Gamma\!\left(1+\frac{1}{\alpha}\right)^{2}}$
& $2\left(1-2^{-\frac{1}{\alpha}}\right)\Gamma\!\left(1+\frac{1}{\alpha}\right)$
& --
\\[6pt]
$\mathrm{Beta}(\alpha,1)$
& $\alpha\ge1$
& $\sqrt{\dfrac{\alpha}{(\alpha+1)^{2}(\alpha+2)}}$
& $\dfrac{2\alpha}{(\alpha+1)(2\alpha+1)}$
& --
\\[6pt]
\bottomrule
\end{tabular}
\end{table}

The catalogue above should not be read as exhaustive, and it would be misleading to suggest that every lifetime model of actuarial practice satisfies the hypotheses of Theorem~\ref{thm:SDleGMD}. The following remark documents an important case in which they fail.

\begin{remark}
\label{rmk:makeham}
Consider the Makeham law, which augments the Gompertz hazard with an age-independent accident component,
\[
h_X(x)=A+B\exp(cx),\qquad A,B,c>0,\qquad x\ge0 ,
\]
and is the standard mortality model of life contingencies \citep{Dickson2019}. Its hazard rate is manifestly increasing, so $S_X$ is log-concave. The density, however, is not log-concave: writing $\log f_X(x)=\log\big(A+B\exp(cx)\big)-Ax-\tfrac{B}{c}(\exp(cx)-1)$, one obtains
\[
\left(\log f_X(x)\right)^{''}
=Bc\exp(cx)\left[\dfrac{Ac}{\big(A+B\exp(cx)\big)^{2}}-1\right],
\]
which is strictly positive precisely when $A+B\exp(cx)<\sqrt{Ac}$, that is, on an initial age range. For the standard calibration $A=0.00022$, $B=2.7\times10^{-6}$ and $c=0.124$, this holds for all $x<60.7$, so the Makeham density is log-convex over essentially the whole pre-retirement range and Proposition~\ref{prop:logConcave_f_X} does not apply. More strikingly, numerical evaluation shows that $F_X$ fails to be log-concave on approximately $[30.8,\,70.9]$, so that $r_X$ is increasing there and the hypotheses of Theorem~\ref{thm:SDleGMD} are themselves violated. Nevertheless, direct computation for the same calibration yields
\[
\SD[X]\approx12.03<\GMD[X]\approx12.27,
\qquad
\mathcal{R}[X]\approx0.981 ,
\]
so GMD dominance holds regardless.

Two conclusions follow. First, the sufficient conditions developed in this paper, although they cover the great majority of parametric families used in practice, do not exhaust the GMD-dominant class; identifying necessary and sufficient conditions remains open, and we discuss the difficulty in Section~\ref{sec:conclusions}. Second, the obstruction is confined to working ages: the Makeham density is log-concave beyond age $60.7$, so Proposition~\ref{prop:SDleGMDtail} of Section~\ref{sec:extensions} certifies GMD dominance for every exceedance $(X\mid X>u)$ with $u\ge60.7$, precisely the age range relevant to annuity and pension liabilities; direct computation confirms the ordering at every threshold.
\end{remark}

\begin{example}
\label{ex:counterexample_f_X_notlogconcave}
This example is constructed as a slight modification of the illustration in Example~\ref{ex:counterexampleGMDdominance}. Suppose that $X$ has a distribution with DDF
\[
S_X(x)=1-\frac{\mathrm{erfi}\left(\frac{x}{2}\right)}{\mathrm{erfi}\left(\frac{1}{2}\right)},\quad x\in[0,1],
\]
where $\mathrm{erfi}(\cdot)$ denotes the imaginary error function introduced in Example~\ref{ex:counterexampleGMDdominance}. Since $f_X'(x)=\tfrac{x}{2}f_X(x)$, the second derivatives of $\log F_X(x)$ and $\log S_X(x)$ are given by
\begin{align*}
\left(\log S_X(x)\right)^{''}&=-\dfrac{h_X(x)}{2S_X(x)}\left(2f_X(x)+xS_X(x)\right),
\\
\left(\log F_X(x)\right)^{''}&=-\dfrac{r_X(x)}{2F_X(x)}\left(2f_X(x)-xF_X(x)\right).
\end{align*}
From the first expression it follows immediately that $\left(\log S_X(x)\right)^{''}\le0$, so that $S_X(x)$ is log-concave on $[0,1]$. For the second, we claim that
\[
\mathrm{erfi}\left(w\right)\le \dfrac{\exp\left(w^2\right)}{w\sqrt{\pi}},\qquad w\in\left[0,\dfrac{1}{2}\right].
\]
Indeed, bounding the integrand by its largest value gives
\[
\mathrm{erfi}(w)=\dfrac{2}{\sqrt{\pi}}\int_0^{w}\exp(t^2)\,\d t\ \le\ \dfrac{2w}{\sqrt{\pi}}\exp\!\left(w^2\right),
\]
and $\dfrac{2w}{\sqrt{\pi}}\exp(w^2)\le\dfrac{\exp(w^2)}{w\sqrt{\pi}}$ holds if and only if $2w^{2}\le1$, that is, for $w\le1/\sqrt{2}\approx0.707$, which covers the range $[0,\tfrac12]$ with room to spare. Setting $w=x/2$ yields
\[
F_X(x)\le \dfrac{2}{x}f_X(x),
\]
so that $\left(\log F_X(x)\right)^{''}\le0$ and $F_X(x)$ is also log-concave. However, examining the density of $X$ reveals
\[
\log f_X(x)=\dfrac{x^2}{4}-\log\left(\sqrt{\pi}\,\mathrm{erfi}\left(\dfrac{1}{2}\right)\right),
\]
which is a convex function. Equivalently, $f_X(x)$ is log-convex, and therefore Proposition~\ref{prop:logConcave_f_X} does not apply, although Theorem~\ref{thm:SDleGMD} does. Direct computation of the dispersion measures confirms that
\[
\SD[X]\approx0.291<\GMD[X]\approx0.335,
\qquad
\mathcal{R}[X]\approx0.8665 ,
\]
as expected. It is worth noting how close this value lies to the extremal constant $\sqrt{3}/2\approx0.8660$ of Proposition~\ref{prop:spacingcriterion}: the law constructed here is, in the sense of the dispersion ratio, very nearly the most GMD-dominant risk that exists.
\end{example}

Beyond affine transformations, the sufficient condition of Proposition~\ref{prop:logConcave_f_X}, and with it GMD dominance, is preserved under independent convolution. We conclude this section by recalling the corresponding closure property of log-concave densities.

\begin{proposition}
\label{prop:f_Xconvolutions}
If $X_1$ and $X_2$ are independent random variables with log-concave densities, then the density of $X_1+X_2$ is also log-concave.
\end{proposition}
This closure property is classical and is due to \cite{Ibragimov1956}; it is also an immediate consequence of the Pr\'ekopa--Leindler inequality \citep{Prekopa1971}, and is surveyed in \cite{Bagnoli2005}. The proof is presented in Appendix~A.\ref{app:f_Xconvolutions}.
\newline

If both \(f_{X_1}(x)\) and \(f_{X_2}(x)\) are log-concave, Proposition~\ref{prop:f_Xconvolutions} ensures that the density of their sum is also log-concave, and Proposition~\ref{prop:logConcave_f_X} applied to each marginal and to the sum yields
\[
\SD[X_1] \le \GMD[X_1],
\qquad
\SD[X_2] \le \GMD[X_2],
\qquad
\SD[X_1 + X_2] \le \GMD[X_1 + X_2].
\]
By induction the conclusion extends to any finite sum of independent risks with log-concave densities, and hence to the aggregate claim of a portfolio of fixed size whose individual severities are log-concave.

\begin{remark}
\label{rmk:closureasymmetry}
Comparing this closure property with that of Section~\ref{sec:SDgeGMD} reveals a duality that we regard as the organizing principle of the two regimes. The sufficient condition for SD dominance is stable under \emph{mixing} but not under convolution: by Proposition~\ref{prop:h_xr_X_mixtures} the class of log-convex survival functions with a common lower endpoint is closed under mixtures, whereas sums of $\DFR$ risks need not be $\DFR$. The sufficient condition for GMD dominance is stable under \emph{convolution} but not under mixing: by Proposition~\ref{prop:f_Xconvolutions} the class of log-concave densities is closed under summation, whereas a mixture of log-concave densities is generally not log-concave, as Example~\ref{ex:SDleGMDtail} will illustrate with a scale mixture of normal laws. These are closure properties of the sufficient conditions, not of the orderings themselves: an equal mixture of $\Expo(1)$ and $2+\Expo(1)$ has $\mathcal{R}[X]\approx0.90$, although each component satisfies the hypothesis of Theorem~\ref{thm:SDgeGMD}.

The actuarial reading is direct. Aggregating homogeneous log-concave risks, as in the collective risk model with a fixed claim count, preserves GMD dominance; introducing unobserved heterogeneity in the scale of exponential-type severities, as in the frailty construction of Example~\ref{ex:frailty}, drives the comparison the other way. Whether the SD or the Gini loading of \eqref{eq:premia} is the more conservative on a given portfolio is thus governed less by the marginal severity itself than by the balance between pooling and heterogeneity in the portfolio's construction.
\end{remark}
%%%%%%%%%%%%%%%%%%%%%%%%%%%%%%%%%%%%%%%%%%%%%%%%%%%%
\section{Truncation and tail conditioning}
\label{sec:extensions}
The results of Sections~\ref{sec:SDgeGMD} and~\ref{sec:SDleGMD} rest on monotonicity assumptions imposed on the \emph{whole} support of the risk. In actuarial practice such global assumptions are frequently unavailable: Remark~\ref{rmk:makeham} exhibited a standard mortality law whose density and reverse hazard rate are well behaved only above a certain age, and Example~\ref{ex:SDgeGMDtail} below exhibits a severity whose hazard rate turns only beyond a threshold. Moreover, the quantities that matter for solvency are rarely computed on the unconditional loss; they are computed on the exceedance $(X\mid X>u)$ above a Value-at-Risk level. This section therefore develops the theory for truncated risks, and shows that each dominance regime is inherited by all sufficiently deep conditionings.

\begin{definition}
\label{def:truncatedX}
Let
\[
X_u^+ := (X \mid X > u)
\quad \text{and} \quad
X_u^- := (X \mid X \le u)
\]
denote the lower- and upper-truncated versions of $X$ at the threshold $u$, respectively.
\end{definition}

The dispersion measures of $X_u^+$ and $X_u^-$ are defined analogously to those of $X$, by considering deviations from independent copies of the truncated variables. Let $(X_u^+)'$ and $(X_u^-)'$ denote independent copies of $X_u^+$ and $X_u^-$, respectively.

\begin{definition}
\label{def:SDGMDtail}
The right- and left-tail versions of the SD and the GMD measures are defined by replacing $X$ with $X_u^+$ and $X_u^-$, respectively, in Equations \eqref{eq:SD} and \eqref{eq:GMD}. Explicitly,
\begin{align*}
\SD[X_u^+]
&=
\sqrt{\dfrac{1}{2}\,\E\left[\left(X_u^+-(X_u^+)'\right)^2\right]}
=
\sqrt{\dfrac{1}{2}\,\E[(X-X')^2 \mid X>u,\,X'>u]},
\\[4pt]
\GMD[X_u^+]
&=
\E\!\left[\left|X_u^+-(X_u^+)'\right|\right]
=
\E[|X-X'| \mid X>u,\,X'>u],
\\[6pt]
\SD[X_u^-]
&=
\sqrt{\dfrac{1}{2}\,\E\left[\left(X_u^--(X_u^-)'\right)^2\right]}
=
\sqrt{\dfrac{1}{2}\,\E[(X-X')^2 \mid X\le u,\,X'\le u]},
\\[4pt]
\GMD[X_u^-]
&=
\E\!\left[\left|X_u^--(X_u^-)'\right|\right]
=
\E[|X-X'| \mid X\le u,\,X'\le u].
\end{align*}
We write $\mathcal{R}[X_u^{\pm}]=\SD[X_u^{\pm}]/\GMD[X_u^{\pm}]$ for the corresponding dispersion ratios.
\end{definition}

Tail versions of the SD and the GMD have appeared in the literature as measures of risk dispersion; see, for example, \cite{Furman2017} and \cite{Chen2025a}. In these contexts the truncation level $u$ is typically chosen as the Value-at-Risk of $X$ at a prudence level $p\in(0,1)$, so that $\GMD[X_u^+]$ becomes the tail Gini functional entering the Gini shortfall and the associated capital allocation rules. The results below therefore determine which of the two tail dispersion measures is the more conservative at any given prudence level.

The starting point is an elementary but decisive observation: the hazard rate is \emph{invariant} under lower truncation, and the reverse hazard rate is invariant under upper truncation. Indeed, for $x\ge u$,
\begin{equation}
\label{eq:hazardinvariance}
S_{X_u^+}(x)=\dfrac{S_X(x)}{S_X(u)},
\qquad
f_{X_u^+}(x)=\dfrac{f_X(x)}{S_X(u)},
\qquad\text{hence}\qquad
h_{X_u^+}(x)=h_X(x),
\end{equation}
and symmetrically $r_{X_u^-}(x)=r_X(x)$ for $x\le u$. The monotonicity hypotheses of Theorem~\ref{thm:SDgeGMD} are thus inherited by the truncated variable whenever they hold on the corresponding half-line, which yields the following proposition.

\begin{proposition}
\label{prop:tailSDgeGMD}
Let $u^*\in\RR$.
\begin{itemize}
\item[(i)] If $h_X(x)$ is decreasing on $[u^*,\infty)$, then
\[
\SD[X_u^+] \ge \GMD[X_u^+]
\qquad\text{for every } u\ge u^* .
\]
\item[(ii)] If $r_X(x)$ is increasing on $(-\infty,u^*]$, then
\[
\SD[X_u^-] \ge \GMD[X_u^-]
\qquad\text{for every } u\le u^* .
\]
\end{itemize}
In either case the inequality is strict unless the corresponding truncated variable is exponential up to reflection and location shift.
\end{proposition}
The proof is given in Appendix~A.\ref{app:tailSDgeGMD}.
\newline

\begin{example}
\label{ex:SDgeGMDtail}
Let $X$ be a random variable with hazard rate
\[
h_X(x)=x\exp(-\theta x)+1,
\qquad x\ge0,\ \theta>0,
\]
so that, integrating, its DDF is
\[
S_X(x)
=
\exp\!\left(
-x-\frac{1}{\theta^2}\left(1-(\theta x+1)\exp(-\theta x)\right)
\right),
\quad x\ge 0 .
\]
Differentiating gives
\[
h_X'(x)=\exp(-\theta x)(1-\theta x),
\]
so $h_X(x)$ is increasing on $[0,1/\theta]$ and decreasing on $[1/\theta,\infty)$. The risk is therefore neither $\IFR$ nor $\DFR$, and no result of Sections~\ref{sec:SDgeGMD} or~\ref{sec:SDleGMD} applies to $X$ itself. Nevertheless, for any $u\ge u^*:=1/\theta$, Proposition~\ref{prop:tailSDgeGMD}\emph{(i)} guarantees SD dominance for $X_u^+$.

Setting $\theta=0.1$, so that $u^*=10$, and plotting the difference $\SD[X_u^+]-\GMD[X_u^+]$ as a function of $u$ yields Figure~\ref{fig:SDgeGMDtail}.
\begin{figure}[H]
\centering
\includegraphics[scale=0.75]{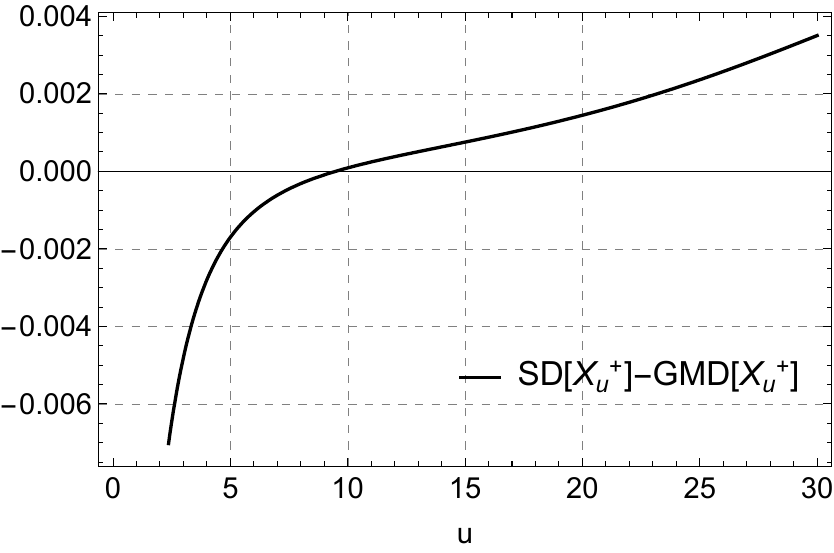}
\caption{Plot of $\SD[X_u^+]-\GMD[X_u^+]$ as a function of the threshold $u$ for the risk of Example~\ref{ex:SDgeGMDtail} with $\theta=0.1$. The gap is positive for every $u\ge u^*=10$, as Proposition~\ref{prop:tailSDgeGMD} predicts, and changes sign at $u\approx9.7$.}
\label{fig:SDgeGMDtail}
\end{figure}
The figure confirms the prediction: the difference is positive for every $u\ge10$, and increases thereafter. It also shows that the ordering is reversed at shallow thresholds, the difference being negative for small $u$ and changing sign at approximately $u\approx9.7$, slightly below the certified threshold $u^*=10$. The conclusion of Proposition~\ref{prop:tailSDgeGMD} is thus valid but conservative, a recurring feature of the sufficient conditions developed in this paper.
\end{example}

A symmetric construction is obtained from the reflected variable $-X$ of Example~\ref{ex:SDgeGMDtail}. By Remark~\ref{rmk:h_Xr_X_Mirror} its reverse hazard rate satisfies $r_{-X}(x)=h_X(-x)$, which is increasing for $x\le-1/\theta$ and decreasing on $[-1/\theta,0]$. Proposition~\ref{prop:tailSDgeGMD}\emph{(ii)} therefore gives $\SD[(-X)_u^-]\ge\GMD[(-X)_u^-]$ for every $u\le-1/\theta$, that is, for every $u\le-10$ when $\theta=0.1$.

We now turn to the complementary regime. Here the invariance \eqref{eq:hazardinvariance} is not enough on its own, because GMD dominance requires the monotonicity of \emph{both} hazard-type functions, and lower truncation leaves only the hazard rate unchanged. For $x\ge u$,
\[
F_{X_u^+}(x)=\dfrac{F_X(x)-F_X(u)}{S_X(u)},
\qquad\text{so that}\qquad
r_{X_u^+}(x)=\dfrac{f_X(x)}{F_X(x)-F_X(u)}=r_X(x)\,\dfrac{F_X(x)}{F_X(x)-F_X(u)} .
\]
The second factor is positive and non-increasing in $x$, because $F_X$ is non-decreasing, so a decreasing reverse hazard rate is inherited by every lower truncation; symmetrically, $h_{X_u^-}(x)=h_X(x)\,S_X(x)/\big(S_X(x)-S_X(u)\big)$ for $x<u$, so an increasing hazard rate is inherited by every upper truncation. Unlike those of Theorem~\ref{thm:SDgeGMD}, the hypotheses of Theorem~\ref{thm:SDleGMD} are therefore hereditary in both directions, in agreement with the closure of log-concave distribution functions under truncation \citep{Bagnoli2005}. Log-concavity of the density is hereditary as well, since a concave function restricted to a convex subset of its domain remains concave, and it supplies a convenient sufficient condition. This yields the following counterpart.

\begin{proposition}
\label{prop:SDleGMDtail}
Let $u^*\in\RR$.
\begin{itemize}
\item[(i)] If $X_{u^*}^+$ satisfies the hypotheses of Theorem~\ref{thm:SDleGMD}, that is, if $h_X$ is increasing on $[u^*,\infty)$ and the reverse hazard rate of $X_{u^*}^+$ is decreasing, then
\[
\SD[X_u^+] \le \GMD[X_u^+]
\qquad\text{for every } u\ge u^* .
\]
Both conditions hold in particular when $h_X$ is increasing and $r_X$ is decreasing on $[u^*,\infty)$, and when $f_X$ is log-concave on $[u^*,\infty)$.
\item[(ii)] If $X_{u^*}^-$ satisfies the hypotheses of Theorem~\ref{thm:SDleGMD}, that is, if $r_X$ is decreasing on $(-\infty,u^*]$ and the hazard rate of $X_{u^*}^-$ is increasing, then
\[
\SD[X_u^-] \le \GMD[X_u^-]
\qquad\text{for every } u\le u^* .
\]
Both conditions hold in particular when $r_X$ is decreasing and $h_X$ is increasing on $(-\infty,u^*]$, and when $f_X$ is log-concave on $(-\infty,u^*]$.
\end{itemize}
\end{proposition}
The proof is given in Appendix~A.\ref{app:SDleGMDtail}.
\newline

\begin{example}
\label{ex:SDleGMDtail}
Let $X$ be a random variable with density given by a mixture of two centered normal distributions with distinct variances,
\[
f_X(x)
= \frac{ q}{\sqrt{2 \pi } \sigma_1}\,\exp\!\left(-\frac{x^2}{2 \sigma_1^2}\right)
+ \frac{1-q }{\sqrt{2 \pi } \sigma_2}\,\exp\!\left(-\frac{x^2}{2 \sigma_2^2}\right),
\quad x\in\RR,
\]
where $\sigma_1,\sigma_2>0$, $\sigma_1\neq\sigma_2$, and $q\in(0,1)$. This is the canonical representation of a portfolio composed of two homogeneous sub-portfolios of differing volatility.

As shown in Example~\ref{ex:SDleGMD}\emph{(3)}, a single normal density is log-concave. However, as anticipated in Remark~\ref{rmk:closureasymmetry}, log-concavity is not preserved under mixtures, and hence $f_X(x)$ is not globally log-concave. Its tail behavior nevertheless remains well behaved: since the component with the larger variance dominates the tail,
\[
\lim_{x\to -\infty}\left(\log f_X(x)\right)''
=
\lim_{x\to \infty}\left(\log f_X(x)\right)''
=
-\dfrac{1}{\max(\sigma_1^2,\sigma_2^2)}<0,
\]
so that both tails of $f_X(x)$ are log-concave.

For concreteness, take $\sigma_1=\tfrac{1}{2}$, $\sigma_2=2$ and $q=\tfrac{3}{4}$. Numerical investigation shows that $\left(\log f_X(x)\right)''<0$ for all $x\ge u^*\approx2.05$ and, by symmetry, for all $x\le v^*\approx-2.05$. By Proposition~\ref{prop:SDleGMDtail} we therefore obtain GMD dominance for $X_u^+$ for every $u\ge u^*$ and for $X_v^-$ for every $v\le v^*$. Since $f_X(x)$ depends on $x$ only through $x^{2}$, the mixture is symmetric about the origin, so that $X_v^-\eqd-X_{-v}^+$ and the left-tail gap is the exact mirror of the right-tail gap; we therefore display only the latter, in Figure~\ref{fig:tailgap}.

\begin{figure}[H]
\centering
\includegraphics[scale=0.75]{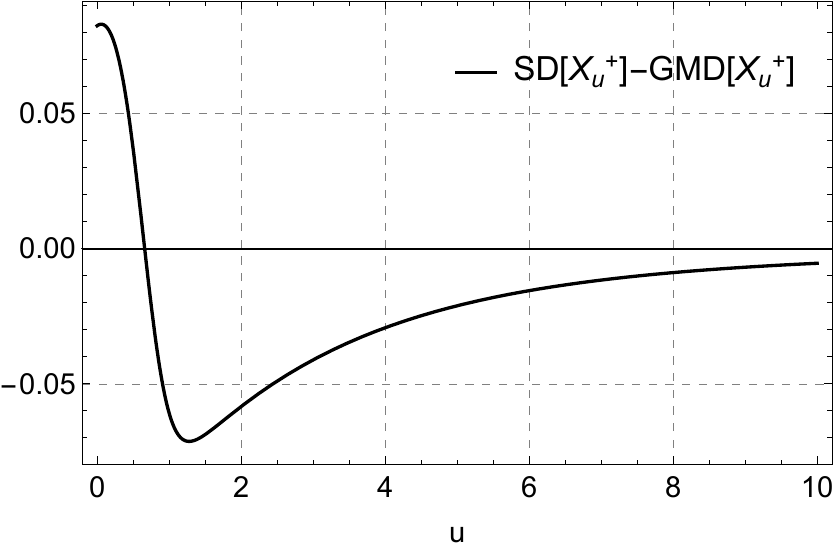}
\caption{Plot of $\SD[X_u^+]-\GMD[X_u^+]$ as a function of the threshold $u$ for the normal scale mixture of Example~\ref{ex:SDleGMDtail}. The gap is positive at shallow thresholds, changes sign at $u\approx0.65$, and remains negative for every deeper conditioning, decaying to zero as $u\to\infty$.}
\label{fig:tailgap}
\end{figure}

Beyond confirming the prediction, the figure reveals a feature of independent interest. At shallow thresholds the difference is \emph{positive}: SD dominance prevails near the centre of the distribution, the sign changes at approximately $u\approx0.65$, and GMD dominance then persists for every deeper threshold, the difference attaining its most negative value near $u\approx1.3$ before decaying to zero as $u\to\infty$. The decay reflects the fact that, as the conditioning becomes extreme, the tail distribution approaches an exponential limit, for which the two measures coincide by Proposition~\ref{prop:expbenchmark}.

The reversal is not an artifact of conditioning on a half-line. The unconditional mixture is itself SD-dominant. Indeed $\Var[X]=q\sigma_1^2+(1-q)\sigma_2^2=\tfrac{19}{16}$, while decomposing $X-X'$ over the four component pairs and using $\E|N(0,v)|=\sqrt{2v/\pi}$ gives
\[
\GMD[X]
=q^2\sqrt{\tfrac{4\sigma_1^2}{\pi}}
+2q(1-q)\sqrt{\tfrac{2(\sigma_1^2+\sigma_2^2)}{\pi}}
+(1-q)^2\sqrt{\tfrac{4\sigma_2^2}{\pi}}
\approx 1.0752 ,
\]
against $\SD[X]\approx1.0897$, so that $\mathcal{R}[X]\approx1.013>1$.
\end{example}

\begin{remark}
\label{rmk:tailreversal}
Example~\ref{ex:SDleGMDtail} deserves emphasis because it contradicts the natural presumption that the ordering observed on a portfolio transfers to its tail. Here a heterogeneous portfolio is SD-dominant in the aggregate and yet GMD-dominant in both tails, and the two statements have different sources: the SD dominance of the aggregate is a heterogeneity effect, produced by the spread of scales across sub-portfolios rather than by a closure property, since neither Gaussian component is $\DFR$, while each tail inherits GMD dominance from the log-concavity of the dominant Gaussian component through Proposition~\ref{prop:SDleGMDtail}. Since the premium loading is computed on the unconditional risk while economic capital is computed on the exceedance above a Value-at-Risk threshold, the practical conclusion is that the more conservative dispersion measure for pricing need not be the more conservative one for capital, and the two questions must be settled separately.
\end{remark}

Propositions~\ref{prop:tailSDgeGMD} and~\ref{prop:SDleGMDtail} may be read in the contrapositive, and doing so converts them into a statistical diagnostic. Define the \emph{tail dispersion gap}
\begin{equation}
\label{eq:Deltau}
\Delta(u):=\SD[X_u^+]-\GMD[X_u^+],
\qquad u\in\RR .
\end{equation}

\begin{corollary}
\label{cor:tailagingdiagnostic}
Let $u_0\in\RR$ be such that $S_X(u_0)>0$.
\begin{itemize}
\item[(i)] If $\Delta(u_0)<0$, then $h_X$ is not decreasing on $[u_0,\infty)$; that is, $X$ is not $\DFR$ beyond $u_0$.
\item[(ii)] If $\Delta(u_0)>0$, then $f_X$ is not log-concave on $[u_0,\infty)$.
\end{itemize}
\end{corollary}
Both statements are the contrapositives of Propositions~\ref{prop:tailSDgeGMD}\emph{(i)} and~\ref{prop:SDleGMDtail}\emph{(i)}, applied at $u=u_0$.
\newline

The practical value of Corollary~\ref{cor:tailagingdiagnostic} lies in the fact that $\Delta(u)$ is built from two ordinary moments of the exceedances over $u$. Both $\SD[X_u^+]$ and $\GMD[X_u^+]$ admit standard unbiased or consistent empirical counterparts computed from the observations exceeding the threshold, the sample Gini mean difference being a classical $U$-statistic, so that $\Delta(u)$ can be estimated directly from claims data without fitting any parametric model. A sign change in $\widehat{\Delta}(u)$ as the threshold is raised therefore signals a change of aging regime in the tail, and Figures~\ref{fig:SDgeGMDtail} and~\ref{fig:tailgap} are precisely pictures of that signal: in each case the diagnostic changes sign in the neighborhood of the point at which the hazard rate turns.
%%%%%%%%%%%%%%%%%%%%%%%%%%%%%%%%%%%%%%%%%%%%%%%%%%%%
\section{Discrete risks and frequency models}
\label{sec:discrete}
The theory developed hitherto has concerned absolutely continuous risks, for which the spacing $Y=|X-X'|$ carries no atom at the origin. Claim counts, however, are integer valued, and the collective risk model rests on frequency distributions supported on the non-negative integers. This section develops the discrete counterpart of the theory in full. As we shall see, the lattice structure does not merely complicate the arguments: it tilts the comparison systematically in favor of the SD, and it does so by an amount that we are able to quantify exactly.

Without loss of generality, let $X$ be a random variable whose support is contained in the integers $\mathbb{Z}$. Throughout this section $p_X(x)$ denotes the probability mass function (PMF) of $X$, a notation chosen to avoid any confusion with the density of the continuous setting, while the CDF $F_X(x)=\mathbb{P}(X\le x)$ and the DDF $S_X(x)=\mathbb{P}(X>x)$ are defined as usual.

The functions $r_X(x)$, $D_X(x,t)$ and $C_X(x,t)$ are defined exactly as in Definition~\ref{def:importantfunctions}, with $p_X$ in place of $f_X$. The hazard rate requires the customary discrete modification, since the conditioning event must include the atom at $x$:
\[
h_X(x)
= \mathbb{P}(X=x \mid X \ge x)
= \dfrac{p_X(x)}{S_X(x-1)},
\quad\text{whereas}\quad
r_X(x)=\mathbb{P}(X=x\mid X\le x)=\dfrac{p_X(x)}{F_X(x)} .
\]
This asymmetry between the two definitions, innocuous as it appears, is the source of every discrepancy between the discrete and continuous theories developed below.

For the quantities above, Proposition~\ref{prop:equivalences} remains valid verbatim, where log-convexity (log-concavity) of a distribution-related function $G_X(x)$, whether $p_X(x)$, $S_X(x)$ or $F_X(x)$, is understood in its discrete sense,
\[
G_X(x)^2 \le (\ge)\, G_X(x-1)\, G_X(x+1),
\]
for all integers $x$ in the support of $X$. Indeed, writing
\[
h_X(x)=1-\dfrac{S_X(x)}{S_X(x-1)}
\qquad\text{and}\qquad
r_X(x)=1-\dfrac{F_X(x-1)}{F_X(x)}
\]
shows at once that $h_X$ is increasing if and only if $S_X$ is log-concave, and that $r_X$ is increasing if and only if $F_X$ is log-convex, which are exactly the equivalences \emph{(A1)}$\iff$\emph{(A3)} and \emph{(B1)}$\iff$\emph{(B3)}. The implications and support restrictions of Proposition~\ref{prop:r_Xh_Ximplications} likewise persist, and admit an even shorter discrete justification, which we record in Appendix~A.\ref{app:discreteRHimplications}.

Finally, for a non-negative integer-valued random variable $X$ and integer $t\ge0$, the mean excess function $m_X(t)$ is defined as in Definition~\ref{def:meanexcessfunction}. Here the crucial difference with the continuous theory appears: since $S_X(0)=\mathbb{P}(X>0)$ may be strictly less than one,
\[
m_X(0)=\dfrac{\E[X]}{S_X(0)}\ \ge\ \E[X],
\]
with equality if and only if $\mathbb{P}(X=0)=0$. Applied to the spacing $Y=|X-X'|$, which necessarily has an atom at the origin of size $\mathbb{P}(X=X')>0$, the identity \eqref{eq:mY0} fails, and it is exactly this failure that the discrete theory must absorb.

We now present the discrete analogues of Propositions~\ref{prop:OrderNonNegative} and~\ref{prop:m_Y_Representation}.

\begin{proposition}
\label{prop:orderNonNegativeDiscrete}
Let $Y$ be a non-negative integer-valued random variable with $\E[Y^2]<\infty$. If
\[
m_Y(t) \ge (\le)\, \E[Y] + \dfrac{1}{2}
\quad \text{for all integers } t \ge 0,
\]
then
\[
\sqrt{\dfrac{1}{2}\, \E[Y^2]} \ge (\le)\, \E[Y].
\]
\end{proposition}

\begin{proposition}
\label{prop:m_Y_RepresentationDiscrete}
Suppose $Y=|X-X'|$ for independent and identically distributed integer-valued random variables $X$ and $X'$. Then the mean excess function $m_Y(t)$ admits the representation
\begin{equation}
\label{eq:m_YRepresentationDiscrete}
m_Y(t)
=\dfrac{\E^{F}\!\left[C_X(X-1,t)\,h_X(X)^{-1}\right]}{\E^{F}[C_X(X-1,t)]}
=\dfrac{\E^{S}\!\left[D_X(X,t)\,r_X(X)^{-1}\right]}{\E^{S}[D_X(X,t)]},
\end{equation}
where the expectations $\E^{F}[\cdot]$ and $\E^{S}[\cdot]$ are taken with respect to the probability measures
\[
Q^{F}(x)=\dfrac{F_X(x-1)}{\E[F_X(X-1)]}\, p_X(x),
\qquad
Q^{S}(x)=\dfrac{S_X(x)}{\E[S_X(X)]}\, p_X(x),
\]
respectively.
\end{proposition}
The proofs are given in Appendices~A.\ref{app:orderNonNegativeDiscrete} and~A.\ref{app:m_Y_RepresentationDiscrete}.
\newline

Two differences from the continuous counterparts, Propositions~\ref{prop:OrderNonNegative} and~\ref{prop:m_Y_Representation}, deserve note, and both stem from the asymmetry recorded above: the condition $m_Y(t)\ge(\le)\,m_Y(0)=\E[Y]$ is replaced by $m_Y(t)\ge(\le)\,\E[Y]+\tfrac12$, a lattice correction compensating for the atom the spacing carries at the origin, and the argument of $C_X$ in the first representation of \eqref{eq:m_YRepresentationDiscrete} is shifted by $-1$, because the discrete hazard rate conditions on $\{X\ge x\}$ while the reverse hazard rate conditions on $\{X\le x\}$. Both are decisive for the ordering, as we now show.

\begin{definition}
\label{def:concentrationMeasure}
The concentration value of a random variable $X$ is defined as
\[
\Lambda=\mathbb{P}(X=X')=\sum_x p_X(x)^2=\E[p_X(X)],
\]
where $X'$ denotes an independent copy of $X$.
\end{definition}

The quantity $\Lambda$ is the probability that two independent claim counts drawn from the same portfolio coincide, and its reciprocal $1/\Lambda$ is the inverse Simpson index, interpretable as the \emph{effective number of distinct outcomes} of the frequency distribution. Both readings will be used below.

\begin{theorem}
\label{thm:SDGMDDiscrete}
Let $X$ be integer valued with $\E[X^2]<\infty$.
\begin{itemize}
\item[(i)] If $h_X(x)$ is decreasing or $r_X(x)$ is increasing, then
\[
\SD[X]>\GMD[X],
\]
the inequality being strict without exception.
\item[(ii)] Suppose that $h_X(x)$ is increasing and $r_X(x)$ is decreasing. If, in addition,
\[
\GMD[X]\le \dfrac{1-\Lambda}{2\Lambda},
\]
then
\[
\SD[X]\le \GMD[X].
\]
\end{itemize}
\end{theorem}
The proof is given in Appendix~A.\ref{app:SDGMDDiscrete}.
\newline

Theorem~\ref{thm:SDGMDDiscrete} exposes a fundamental asymmetry: discrete risks favor the SD. On the heavy-tailed side the conclusion is not merely inherited from the continuous theory but strengthened, monotonicity of either hazard-type function delivering a \emph{strict} inequality where Theorem~\ref{thm:SDgeGMDstrict} admitted an exponential boundary case. On the light-tailed side, by contrast, monotonicity alone is insufficient and must be supplemented by the concentration constraint. The proof makes the mechanism explicit: the argument of part~\emph{(i)} yields the lower bound
\[
m_Y(0)\ \ge\ \dfrac{1+\Lambda}{2\Lambda}\ >\ \dfrac{1}{2\Lambda},
\]
whereas the required threshold for the reverse ordering is $m_Y(0)\le1/(2\Lambda)$, and the same Chebyshev argument run in the opposite direction delivers only the upper bound $m_Y(0)\le(1+\Lambda)/(2\Lambda)$, which overshoots the target by exactly $\tfrac12$. The gap between the two bounds is the lattice correction, and it must be closed by an explicit hypothesis.

The concentration constraint admits a transparent reformulation. Since $\GMD[X]=\E[Y]$ and $S_Y(0)=1-\Lambda$, the condition of part~\emph{(ii)} is equivalent to each of
\begin{equation}
\label{eq:concentrationforms}
m_Y(0)\le\dfrac{1}{2\Lambda},
\qquad
\Lambda\le\dfrac{1}{1+2\,\GMD[X]},
\qquad
\dfrac{1}{\Lambda}\ \ge\ 1+2\,\GMD[X] .
\end{equation}
The last of these is the most informative: GMD dominance requires the frequency distribution to be spread over at least $1+2\,\GMD[X]$ effective atoms. A count distribution that is too concentrated -- one in which ties between two independent policyholders are too likely -- cannot exhibit GMD dominance however light its tails, because the atom of the spacing at the origin then contributes too much mass to the first moment relative to the second.

\begin{remark}
\label{rmk:latticepremium}
The strictness in part~\emph{(i)} is best understood through the geometric distribution, the discrete memoryless law and hence the natural analogue of the exponential benchmark of Proposition~\ref{prop:expbenchmark}. From Example~\ref{ex:SDGMDDiscrete}\emph{(1)} below,
\[
\mathcal{R}[X]=\dfrac{\SD[X]}{\GMD[X]}=\dfrac{2-p}{2\sqrt{1-p}},
\qquad p\in(0,1),
\]
and since $(2-p)^2-4(1-p)=p^2>0$ we obtain $\mathcal{R}[X]>1$ for every $p$, with $\mathcal{R}[X]\downarrow1$ as $p\downarrow0$. The knife-edge of the continuous theory is therefore broken by discreteness itself, and the excess is governed by $p$, which is precisely the ratio of the lattice spacing to the mean of the law. As the spacing becomes negligible relative to the scale of the risk, the geometric law converges to the exponential and the equality of Proposition~\ref{prop:expbenchmark} is recovered in the limit. The term $\tfrac12$ in Proposition~\ref{prop:orderNonNegativeDiscrete} is the exact accounting of this effect.

As a consistency check, the geometric law satisfies $\Lambda=p/(2-p)$, whence $(1-\Lambda)/(2\Lambda)=(1-p)/p$, while $\GMD[X]=2(1-p)/\big(p(2-p)\big)$. The concentration condition therefore fails by the factor $2/(2-p)>1$, exactly as it must for a law that is SD-dominant.
\end{remark}

To illustrate the scope of Theorem~\ref{thm:SDGMDDiscrete}, we now examine the frequency distributions of collective risk theory. We note first that the reflection argument of Remark~\ref{rmk:h_Xr_X_Mirror}, and the sufficiency of log-concavity of $p_X(x)$ for an increasing $h_X(x)$ and a decreasing $r_X(x)$ established in Proposition~\ref{prop:logConcave_f_X}, both remain valid in the discrete setting.

\begin{example}
\label{ex:SDGMDDiscrete}
We present several common frequency distributions satisfying the conditions of Theorem~\ref{thm:SDGMDDiscrete}. The first two illustrate case~\emph{(i)}, in which SD dominance holds, while the last two correspond to case~\emph{(ii)}.
~
\begin{itemize}
\item[(1)] Suppose that $X$ follows a geometric distribution with parameter $p\in(0,1)$, so that $p_X(x)=p(1-p)^x$ for $x\ge0$. Since $S_X(x-1)=\mathbb{P}(X\ge x)=(1-p)^x$,
\[
h_X(x)=\dfrac{p(1-p)^x}{(1-p)^x}=p,\qquad x\ge0,
\]
so the hazard rate is decreasing, indeed constant. Direct computation of the two measures yields
\[
\SD[X]=\dfrac{\sqrt{1-p}}{p}
\;>\;
\GMD[X]=\frac{2(1-p)}{p(2-p)},
\]
in agreement with Theorem~\ref{thm:SDGMDDiscrete}\emph{(i)} and with Remark~\ref{rmk:latticepremium}.

\item[(2)] Let $X$ follow a discrete Pareto-type distribution, also known as the Zipf distribution, with PMF
\[
p_X(x)=\dfrac{x^{-\alpha-1}}{\zeta(\alpha+1)},\qquad x\ge1,\ \alpha>2,
\]
where $\zeta(\cdot)$ denotes the Riemann zeta function; the restriction $\alpha>2$ ensures a finite second moment. The corresponding hazard rate is decreasing for all such $\alpha$: writing $1/h_X(x)=S_X(x-1)/p_X(x)=\sum_{j\ge0}\prod_{i=0}^{j-1}p_X(x+i+1)/p_X(x+i)$, each ratio $p_X(y+1)/p_X(y)=\big(y/(y+1)\big)^{\alpha+1}$ is increasing in $y$, so $1/h_X$ is increasing. For concreteness, consider $\alpha=3$.

Here $\SD[X]=\sqrt{\zeta(2)/\zeta(4)-\big(\zeta(3)/\zeta(4)\big)^2}$, and numerical evaluation gives
\[
\SD[X]\approx0.535>\GMD[X]\approx0.209,
\qquad
\mathcal{R}[X]\approx2.56 ,
\]
again in agreement with Theorem~\ref{thm:SDGMDDiscrete}. The magnitude of the ratio is notable: for a heavy-tailed count distribution the SD can exceed the GMD by a factor well above two, so the choice between the two loadings is far from immaterial.

\item[(3)] Assume that $X$ follows a Poisson distribution with mean $\theta>0$. To ensure that $h_X(x)$ is increasing and $r_X(x)$ is decreasing, we verify that the Poisson PMF is log-concave. For every integer $x\ge1$,
\[
\dfrac{p_X(x)^2}{p_X(x-1)\,p_X(x+1)}
=
\dfrac{\dfrac{\theta^{2x}\exp(-2\theta)}{(x!)^2}}
{\dfrac{\theta^{x-1}\exp(-\theta)}{(x-1)!}\,
 \dfrac{\theta^{x+1}\exp(-\theta)}{(x+1)!}}
=
\dfrac{(x-1)!\,(x+1)!}{(x!)^2}
=
\dfrac{x+1}{x}
>1,
\]
while for $x=0$ the inequality $p_X(0)^2>0=p_X(-1)p_X(1)$ holds trivially. Hence the PMF is log-concave, which establishes the first requirement for GMD dominance. Figure~\ref{fig:poissonpanels} displays the differences $\GMD[X]-(1-\Lambda)/(2\Lambda)$ and $\SD[X]-\GMD[X]$ as functions of $\theta$, where $\Lambda=\exp(-2\theta)I_0(2\theta)$ and $I_0$ denotes the modified Bessel function of the first kind.

\begin{figure}[H]
\centering
\begin{subfigure}{0.48\textwidth}
\centering
\includegraphics[width=\textwidth]{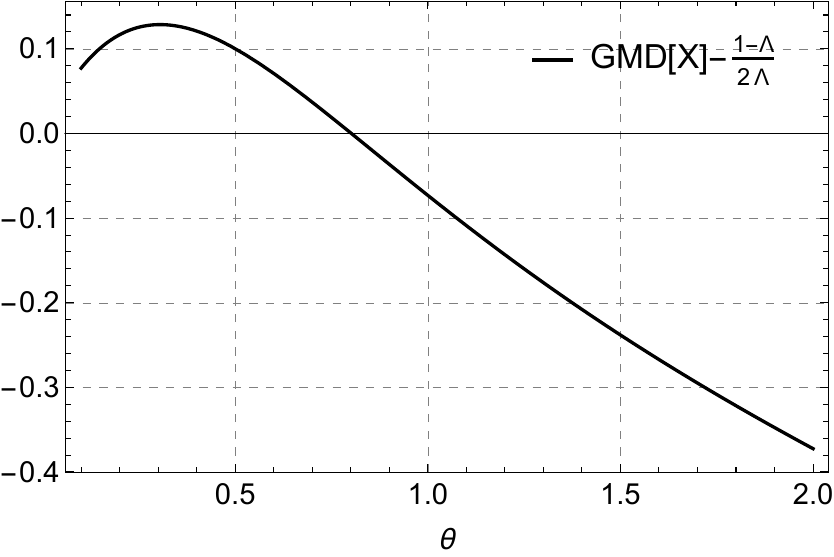}
\caption{$\GMD[X]-\frac{1-\Lambda}{2\Lambda}$}
\label{fig:SDleGMDDiscrete12}
\end{subfigure}
\hfill
\begin{subfigure}{0.48\textwidth}
\centering
\includegraphics[width=\textwidth]{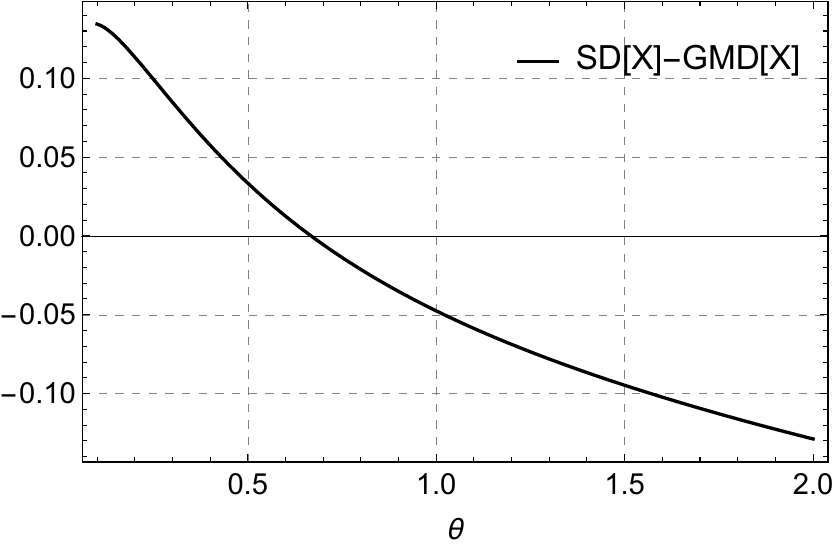}
\caption{$\SD[X]-\GMD[X]$}
\label{fig:SDleGMDDiscrete11}
\end{subfigure}
\caption{Poisson frequency model of Example~\ref{ex:SDGMDDiscrete}\emph{(3)}. Panel (a) plots $\GMD[X]-(1-\Lambda)/(2\Lambda)$, which becomes non-positive at $\theta\approx0.802$, so that the concentration bound of Theorem~\ref{thm:SDGMDDiscrete}\emph{(ii)} binds there; panel (b) plots $\SD[X]-\GMD[X]$, which changes sign already at $\theta\approx0.668$. The bound is therefore sufficient but conservative.}
\label{fig:poissonpanels}
\end{figure}

The concentration bound is satisfied for $\theta\ge0.802$, and GMD dominance is duly observed over that range. The ordering in fact reverses somewhat earlier, at $\theta\approx0.668$, so that the concentration condition, like the monotonicity conditions of the continuous theory, is sufficient but not necessary. The practical import is that for portfolios with an expected claim frequency below roughly two thirds of a claim per policy per period -- a regime common in motor and household lines -- the SD loading is the more conservative of the two, whereas above that level the Gini loading is.

\item[(4)] Let $X$ be a negative binomial random variable with parameters $r>0$ and $p\in(0,1)$, with PMF $p_X(x)=\binom{x+r-1}{x}(1-p)^xp^r$. A direct calculation gives
\begin{align*}
\dfrac{p_X(x)^2}{p_X(x-1)\,p_X(x+1)}
&=
\dfrac{\displaystyle\binom{x+r-1}{x}^2}
{\displaystyle\binom{x+r-2}{x-1}\binom{x+r}{x+1}}
=
\dfrac{(x+r-1)(x+1)}{x(x+r)} ,
\end{align*}
and since $(x+r-1)(x+1)-x(x+r)=r-1$, the PMF is log-concave if and only if $r\ge1$. Fixing $r=2$ and computing the relevant quantities as functions of $p$ yields
\[
\SD[X]=\dfrac{\sqrt{2(1-p)}}{p},\qquad
\GMD[X]=\frac{4(1-p)\bigl(3-(3-p)p\bigr)}{p(2-p)^3},
\]
and
\[
\dfrac{1-\Lambda}{2\Lambda}
=
\frac{2}{p}-\frac{1}{2-(2-p)p}-1,
\qquad\text{where}\qquad
\Lambda=\dfrac{p\left(1+(1-p)^2\right)}{(2-p)^3}.
\]
A direct comparison shows that $\GMD[X]\le(1-\Lambda)/(2\Lambda)$ for $p\in(0,0.569]$, while $\SD[X]\le\GMD[X]$ for $p\in(0,0.654]$. Together with the log-concavity of $p_X(x)$ at $r=2$, this confirms the implication of Theorem~\ref{thm:SDGMDDiscrete}\emph{(ii)} and again quantifies the conservatism of the concentration bound.
\end{itemize}
\end{example}

Beyond verifying the monotonicity of $h_X$ and $r_X$ directly, further models are generated through the closure properties, which persist unchanged in the discrete setting and retain the duality established in Remark~\ref{rmk:closureasymmetry}. As in Proposition~\ref{prop:h_xr_X_mixtures}, the class of distributions with decreasing $h_X(x)$ and a common lower endpoint is closed under mixing; since the geometric distribution exhibits SD dominance and the negative binomial with $r\le1$ is a mixture of geometric laws \citep{Korolev2019}, the negative binomial with $r\le1$ exhibits SD dominance as well. This is the frequency counterpart of Example~\ref{ex:frailty}: a negative binomial claim count arises from a Poisson count with gamma-distributed exposure heterogeneity, and it is that heterogeneity which drives the comparison toward the SD. Dually, the class of log-concave PMFs is closed under convolution, as in Proposition~\ref{prop:f_Xconvolutions}, so that any convolution of Poisson, binomial or negative binomial (with $r\ge1$) counts remains log-concave; there GMD dominance reduces to verifying the concentration bound \eqref{eq:concentrationforms}, which becomes easier to satisfy as the portfolio grows, since aggregation raises the effective number of atoms $1/\Lambda$.
%%%%%%%%%%%%%%%%%%%%%%%%%%%%%%%%%%%%%%%%%%%%%%%%%%%%
\section{Conclusions}
\label{sec:conclusions}
This paper has asked which of the two canonical measures of variability assigns the larger penalty to a given risk, and has answered in terms of tail structure. The analysis rested on a single reduction: writing $Y=|X-X'|$ for the spacing generated by two independent copies of the risk, the comparison of $\SD[X]$ with $\GMD[X]$ is \emph{equivalent} to the comparison of $\CV[Y]$ with unity, and hence to the position of the spacing relative to an exponential law of the same mean. The exponential distribution, whose spacing is again exponential, is therefore the universal knife-edge. Reading the GMD as twice the maxiance, that is, as the second central dual moment of \cite{Eeckhoudt2022}, the question becomes an explicit comparison between a primal and a dual assessment of second-order variability.

The structural engine was a closed-form representation of the mean excess function of the spacing in terms of the hazard and reverse hazard rates of $X$, taken under the laws of the sample maximum and minimum. A heavy tail on either side yields SD dominance, two-sided light tails yield GMD dominance, and one-sided lightness does not suffice. When both the hazard and the reverse hazard rate are monotone, equality characterizes the exponential law up to affine transformation, so that the dispersion ratio $\mathcal{R}[X]$ of an absolutely continuous risk occupies $[\sqrt3/2,\infty)$, extremal at the uniform law and critical at the exponential; the Gamma family traverses this entire scale, crossing unity exactly at $\alpha=1$. The sufficient conditions for the two regimes are closed under different operations -- those for SD dominance under mixing over a common support, those for GMD dominance under convolution, both under tail truncation -- and this duality carries the actuarial content: heterogeneity in the scale of claim severities drives the comparison toward the SD, pooling of homogeneous log-concave risks preserves GMD dominance, and in the gamma-frailty model the excess of $\mathcal{R}[X]^2$ over unity equals, to first order, the squared coefficient of variation of the unobserved frailty. Heavy-tailed severities are accordingly SD-dominant and lifetime models GMD-dominant, with the Makeham law an instructive exception that satisfies neither set of hypotheses yet obeys the ordering. Conditioning can reverse the ordering outright -- a scale mixture of normals is SD-dominant in aggregate and GMD-dominant in both tails -- so the measure that is more conservative for pricing need not be the one that is more conservative for capital. Read in the contrapositive, the truncation results turn the sign of $\Delta(u)=\SD[X_u^+]-\GMD[X_u^+]$ into a nonparametric diagnostic for tail aging, estimable from exceedance data without a parametric fit. In the discrete setting the lattice tilts the comparison toward the SD: monotonicity alone gives a strict inequality, and the reverse ordering additionally requires $1/\Lambda\ge1+2\GMD[X]$, a frequency distribution spread over a sufficient effective number of atoms.

Several questions remain open. The most pressing is the characterization problem: the Makeham law and the Beta family with shape in $[\tfrac14,1)$ show that the GMD-dominant class strictly exceeds the class covered by our conditions, and since narrowing a double-truncation window drives every risk toward the uniform law, no characterization in terms of truncations alone can succeed. The discrete concentration bound is likewise conservative, the diagnostic $\Delta(u)$ invites a formal test of tail aging based on $U$-statistics, and relaxing the independence of $X$ and $X'$, natural for two policyholders but not for one risk observed at two dates, is a natural sequel.

\section*{Acknowledgements}
We are grateful to the Handling Editor and three anonymous referees for their careful and constructive reports on an earlier version of this paper. The reduction of the comparison to the coefficient of variation of the spacing $|X-X'|$, which organizes the present version, was suggested by one of the referees, and the connection with dual moments and the maxiance was pointed out by another.
%% CHECK: add funding acknowledgements here if applicable.
%\newpage
%%%%%%%%%%%%%%%%%%%%%%%%%%%%%%%%%%%%%%%%%%%%%%%%%%%%%%
%% Bibliography embedded from References.bib (style: apalike, natbib author-year).
%% Entries corrected by hand: Yaari (sole author), Shaked--Shanthikumar (2007),
%% Rockafellar initials, Furman et al. (vol./pages), Chen et al. (arXiv),
%% Yitzhaki (2003, Metron), Ceriani--Verme (2012), Denuit et al. (title),
%% capitalization of proper nouns in titles.

%%%%%%%%%%%%%%%%%%%%%%%%%%%%%%%%%%%%%%%%%%%%%%%%%%%%%%%%%%%%%%%%%%%%%%%%%%%%%%%%%%%%%%%%%%%%%%%%%%%%%
\newpage
\begin{appendices}
\renewcommand{\thesection}{\Alph{section}}
\section{Proofs}
%\addcontentsline{toc}{section}{Proofs}
\renewcommand{\thesubsection}{\roman{subsection}}
\subsection{Proof of Proposition \ref{prop:equivalences}}
\label{app:equivalences}
\begin{proof}
We prove the equivalence of conditions \emph{(A1)--(A3)} and then that of \emph{(B1)--(B3)}. Throughout we use the fact that, $X$ being absolutely continuous with a density that is strictly positive on the interior of its support, the functions $\log S_X$ and $\log F_X$ are locally absolutely continuous there, with
\[
\big(\log S_X(x)\big)^{'}=-h_X(x)
\qquad\text{and}\qquad
\big(\log F_X(x)\big)^{'}=r_X(x)
\]
almost everywhere. We also use the elementary characterization that a function $\varphi$ is concave (convex) if and only if, for each $t\ge0$, the increment $x\mapsto\varphi(x+t)-\varphi(x)$ is decreasing (increasing) in $x$.

\medskip
\noindent\emph{(A1)}$\implies$\emph{(A2)}: Assume $h_X(x)$ is increasing (decreasing). Then, for each $t\ge0$, the function
\[
D_X(x,t)=\dfrac{S_X(x+t)}{S_X(x)}=\exp\left(-\int_{x}^{x+t}h_X(w)\,\d w\right)=\exp\left(-\int_{0}^{t}h_X(u+x)\,\d u\right)
\]
is decreasing (increasing) in $x$, since the exponent is increasing (decreasing) in $x$.

\medskip
\noindent\emph{(A2)}$\implies$\emph{(A3)}: Suppose that for each $t\ge0$, $D_X(x,t)$ is decreasing (increasing) in $x$. Taking $x_1<x_2$ then gives $D_X(x_1,t)\ge(\le)\,D_X(x_2,t)$, which implies
\[
\dfrac{\log S_X(x_1+t)-\log S_X(x_1)}{t}\ge(\le)\, \dfrac{\log S_X(x_2+t)-\log S_X(x_2)}{t},
\]
that is, the secant slopes of $\log S_X$ are decreasing (increasing). This is equivalent to $S_X(x)$ being log-concave (log-convex).

\medskip
\noindent\emph{(A3)}$\implies$\emph{(A1)}: Suppose that $S_X(x)$ is a log-concave (log-convex) DDF. Then $\log S_X(x)$ is concave (convex), so that its derivative is decreasing (increasing) almost everywhere. Since
\[
\big(\log S_X(x)\big)^{'}=-h_X(x),
\]
it follows that $h_X(x)$ is increasing (decreasing). This closes the first cycle of implications.

\medskip
\noindent\emph{(B1)}$\implies$\emph{(B2)}: Assume $r_X(x)$ is increasing (decreasing). For each $t\ge0$,
\[
C_X(x,t)=\dfrac{F_X(x-t)}{F_X(x)}=\exp\left(-\int_{x-t}^{x}r_X(w)\,\d w\right)=\exp\left(-\int_{0}^{t}r_X(x-u)\,\d u\right),
\]
whose exponent is increasing (decreasing) in $x$. Hence $C_X(x,t)$ is decreasing (increasing) in $x$.

\medskip
\noindent\emph{(B2)}$\implies$\emph{(B3)}: Suppose that for each $t\ge0$ the function $C_X(x,t)$ is decreasing (increasing) in $x$. Equivalently, $x\mapsto\log F_X(x)-\log F_X(x-t)$ is increasing (decreasing) in $x$, so that the secant slopes of $\log F_X$ are increasing (decreasing). This is precisely log-convexity (log-concavity) of $F_X(x)$. We emphasize that the direction of the conclusion is reversed relative to part \emph{(A)}, because $C_X$ compares the CDF at $x$ with its value at the \emph{smaller} argument $x-t$.

\medskip
\noindent\emph{(B3)}$\implies$\emph{(B1)}: If $F_X(x)$ is log-convex (log-concave), then $\log F_X(x)$ is convex (concave) and its derivative is increasing (decreasing) almost everywhere. Since $\big(\log F_X(x)\big)^{'}=r_X(x)$, the reverse hazard rate $r_X(x)$ is increasing (decreasing). This completes the proof.
\end{proof}

\subsection{Proof of Proposition \ref{prop:GMDdual}}
\label{app:GMDdual}
\begin{proof}
Throughout, $X'$ is an independent copy of $X$, with $X_{(1)}=\min(X,X')$ and $X_{(2)}=\max(X,X')$.

\medskip
\noindent\emph{(i)} For any real numbers $a$ and $b$ one has $|a-b|=\max(a,b)-\min(a,b)$, whence $Y=|X-X'|=X_{(2)}-X_{(1)}$ pathwise. Taking expectations, which are finite because $\E|X|<\infty$, yields $\GMD[X]=\E[X_{(2)}]-\E[X_{(1)}]$.

\medskip
\noindent\emph{(ii)} From $\max(a,b)+\min(a,b)=a+b$ we obtain $X_{(1)}+X_{(2)}=X+X'$ pathwise, and therefore $\E[X_{(1)}]+\E[X_{(2)}]=2\E[X]$. Substituting $\E[X_{(1)}]=2\E[X]-\E[X_{(2)}]$ into \emph{(i)} gives
\[
\GMD[X]=2\E\big[X_{(2)}\big]-2\E[X]=2\,\Mx[X],
\]
and substituting $\E[X_{(2)}]=2\E[X]-\E[X_{(1)}]$ instead gives the symmetric expression $2(\E[X]-\E[X_{(1)}])$.

\medskip
\noindent\emph{(iii)} Since $X$ and $X'$ are independent with common CDF $F_X$, we have
\[
\mathbb{P}\big(X_{(2)}>x\big)=1-F_X(x)^2
\qquad\text{and}\qquad
\mathbb{P}\big(X_{(1)}>x\big)=S_X(x)^2 .
\]
Because $X_{(2)}\ge X_{(1)}$ almost surely and both have finite means, part \emph{(i)} and Fubini's theorem give
\begin{align*}
\GMD[X]
&=\E\big[X_{(2)}-X_{(1)}\big]
\\
&=\int_{-\infty}^{\infty}\Big(\mathbb{P}\big(X_{(2)}>x\big)-\mathbb{P}\big(X_{(1)}>x\big)\Big)\d x
\\
&=\int_{-\infty}^{\infty}\Big(1-F_X(x)^2-\big(1-F_X(x)\big)^2\Big)\d x
\\
&=\int_{-\infty}^{\infty}\Big(2F_X(x)-2F_X(x)^2\Big)\d x
\\
&=2\int_{-\infty}^{\infty}F_X(x)\,S_X(x)\,\d x ,
\end{align*}
the integrand being non-negative and integrable.

\medskip
\noindent\emph{(iv)} Since $X$ is absolutely continuous, $\mathbb{P}(X=X')=0$ and hence, by exchangeability,
\[
\E\big[X_{(2)}\big]=\E\big[X\,\mathbbm{1}\{X'\le X\}\big]+\E\big[X'\,\mathbbm{1}\{X<X'\}\big]=2\,\E\big[X\,\mathbbm{1}\{X'\le X\}\big].
\]
Conditioning on $X$ and using independence gives $\E[X\,\mathbbm{1}\{X'\le X\}]=\E[X F_X(X)]$, so that $\E[X_{(2)}]=2\,\E[XF_X(X)]$. Since $F_X(X)$ is uniform on $(0,1)$ we have $\E[F_X(X)]=\tfrac12$, and therefore
\[
\Cov\big(X,F_X(X)\big)=\E\big[XF_X(X)\big]-\tfrac12\E[X]=\tfrac12\E\big[X_{(2)}\big]-\tfrac12\E[X]=\tfrac12\,\Mx[X]=\tfrac14\,\GMD[X],
\]
where the last two equalities use part \emph{(ii)}. Multiplying by $4$ yields the assertion.

\medskip
\noindent\emph{(v)} For $X\ge0$ the variable $X_{(2)}$ is non-negative, so
\begin{align*}
\E\big[X_{(2)}\big]&=\int_0^{\infty}\mathbb{P}\big(X_{(2)}>x\big)\,\d x=\int_0^{\infty}\Big(1-F_X(x)^2\Big)\d x
\\
&=\int_0^{\infty}\Big(1-\big(1-S_X(x)\big)^2\Big)\d x=\int_0^{\infty}g\big(S_X(x)\big)\,\d x ,
\end{align*}
with $g(p)=1-(1-p)^2=2p-p^2$. The function $g$ satisfies $g(0)=0$, $g(1)=1$, $g'(p)=2-2p\ge0$ and $g''(p)=-2<0$ on $[0,1]$, hence it is an increasing concave distortion. This completes the proof.
\end{proof}

\subsection{Proof of Proposition \ref{prop:spacingcriterion}}
\label{app:spacingcriterion}
\begin{proof}
Since $X$ is non-degenerate we have $\mathbb{P}(X\neq X')>0$ and hence $\mu:=\E[Y]>0$, while $\E[Y^2]\le2\E[X^2]+2\E[X]^2<\infty$ by the elementary inequality $(a-b)^2\le 2a^2+2b^2$ together with independence. Both sides of the asserted equivalences are therefore finite and the chain
\begin{align*}
\SD[X]\ \ge\ \GMD[X]
&\iff \sqrt{\tfrac12\,\E[Y^2]}\ \ge\ \E[Y]
\\
&\iff \tfrac12\,\E[Y^2]\ \ge\ \big(\E[Y]\big)^2
\\
&\iff \E[Y^2]-\big(\E[Y]\big)^2\ \ge\ \big(\E[Y]\big)^2
\\
&\iff \Var[Y]\ \ge\ \big(\E[Y]\big)^2
\\
&\iff \CV[Y]=\dfrac{\sqrt{\Var[Y]}}{\E[Y]}\ \ge\ 1
\end{align*}
follows, all steps being reversible because the quantities involved are non-negative and $\E[Y]>0$. The same chain with the inequalities reversed, or replaced by equalities, gives the remaining assertions.

For the final bound, recall from Proposition~\ref{prop:GMDdual}\emph{(iv)} that $\GMD[X]=4\,\Cov\big(X,F_X(X)\big)$. Since $X$ is absolutely continuous, $U:=F_X(X)$ is uniformly distributed on $(0,1)$, so that $\Var[U]=\tfrac{1}{12}$, and the Cauchy--Schwarz inequality gives
\[
\GMD[X]=4\,\Cov(X,U)\ \le\ 4\,\SD[X]\sqrt{\Var[U]}=\dfrac{2}{\sqrt3}\,\SD[X],
\]
which upon rearrangement is $\mathcal{R}[X]\ge\sqrt3/2$. Equality holds if and only if $X-\E[X]=c\,(U-\tfrac12)$ almost surely for some constant $c$, and then $c=12\,\Cov(X,U)=3\,\GMD[X]>0$ because $X$ is non-degenerate; that is, if and only if $X$ is an increasing affine function of a uniform variable, hence uniformly distributed on an interval. Conversely, the uniform law on $[0,1]$ has $\SD[X]=1/\sqrt{12}$ and $\GMD[X]=1/3$, so that $\mathcal{R}[X]=\sqrt3/2$, and the same holds on every interval by affine invariance.
\end{proof}

\subsection{Proof of Proposition \ref{prop:expbenchmark}}
\label{app:expbenchmark}
\begin{proof}
Let $X\sim\Expo(\theta)$, so that $f_X(x)=\theta\exp(-\theta x)$ and $S_X(x)=\exp(-\theta x)$ for $x\ge0$. Using the representation of the survival function of the spacing established in Appendix~A.\ref{app:m_Y_Representation}, namely
\[
S_Y(y)=2\int_{0}^{\infty}S_X(x+y)\,\d F_X(x),\qquad y\ge0,
\]
we obtain
\begin{align*}
S_Y(y)&=2\int_{0}^{\infty}\exp\big(-\theta(x+y)\big)\,\theta\exp(-\theta x)\,\d x
=2\theta\exp(-\theta y)\int_{0}^{\infty}\exp(-2\theta x)\,\d x
\\
&=\dfrac{2\theta\exp(-\theta y)}{2\theta}
=\exp(-\theta y).
\end{align*}
Hence $Y\sim\Expo(\theta)$, so that $\E[Y]=\theta^{-1}$ and $\Var[Y]=\theta^{-2}$, giving $\CV[Y]=1$. Proposition~\ref{prop:spacingcriterion} then yields $\SD[X]=\GMD[X]$, and the common value is $\GMD[X]=\E[Y]=\theta^{-1}$.

For the affine statement, let $Z=aX+b$ with $a\neq0$. By the invariance recorded in Section~\ref{sec:intro},
\[
\SD[Z]=|a|\,\SD[X]=|a|\,\GMD[X]=\GMD[Z],
\]
so $\mathcal{R}[Z]=\mathcal{R}[X]=1$. This completes the proof.
\end{proof}

\subsection{Proof of Proposition \ref{prop:OrderNonNegative}}
\label{app:OrderNonNegative}
\begin{proof}
We begin by rewriting the quantity $\dfrac{1}{2}\,\E[Y^2]$ as follows:
\begin{align*}
\dfrac{1}{2}\,\E[Y^2]
&=\dfrac{1}{2}\,\int_0^\infty y^2\, \d F_Y(y)
\\
&=\int_0^\infty\left(\int_0^y w\,\d w\right)\d F_Y(y)
\\
&=\int_0^\infty\left(\int_w^\infty \d F_Y(y)\right)w\,\d w
\\
&=\int_0^\infty wS_Y(w)\,\d w
\\
&=\int_0^\infty\left(\int_0^w\d z\right)S_Y(w)\,\d w
\\
&=\int_0^\infty\left(\int_z^\infty S_Y(w)\,\d w\right)\d z
\\
&=\int_0^\infty m_Y(z)\,S_Y(z)\,\d z,
\end{align*}
where every interchange of the order of integration is justified by Tonelli's theorem, the integrands being non-negative. In the final step we used the standard representation of the mean excess function,
\[
m_Y(z)=\dfrac{\int_z^\infty S_Y(w)\,\d w}{S_Y(z)}.
\]

Since $Y$ has no atom at the origin we have $S_Y(0)=1$ and therefore
\[
m_Y(0)=\dfrac{\int_0^{\infty}S_Y(w)\,\d w}{S_Y(0)}=\E[Y].
\]
This identity is the point at which the absence of an atom at $0$ is used; without it one would only have $m_Y(0)=\E[Y]/S_Y(0)\ge\E[Y]$, and the argument below would produce the weaker conclusion $\tfrac12\E[Y^2]\ge m_Y(0)\E[Y]$. The corresponding correction in the presence of atoms is quantified in Proposition~\ref{prop:orderNonNegativeDiscrete}.

Assume now that $m_Y(z)\ge(\le)\,m_Y(0)$ for all $z\ge0$. It follows that
\begin{align*}
\dfrac{1}{2}\,\E[Y^2]
&=\int_0^\infty m_Y(z)\,S_Y(z)\,\d z
\\
&\ge(\le)\,m_Y(0)\int_0^\infty S_Y(z)\,\d z
\\
&=m_Y(0)\,\E[Y]
\\
&=\E[Y]^2 .
\end{align*}
Equivalently,
\[
\dfrac{1}{2}\,\E[Y^2]\ge(\le)\,\E[Y]^2 .
\]
Taking square roots on both sides, which is legitimate since both quantities are non-negative, yields the desired result and completes the proof.
\end{proof}

\subsection{Proof of Proposition \ref{prop:cxExponential}}
\label{app:cxExponential}
\begin{proof}
Write $g_Y(t)=\E\big[(Y-t)_+\big]=\int_t^{\infty}S_Y(w)\,\d w$ for $t\ge0$ for the stop-loss transform of $Y$, and let $g_{E_{\mu}}$ denote that of $E_{\mu}$. The function $g_Y$ is non-negative, convex, decreasing and locally absolutely continuous, with $g_Y(0)=\mu$ and $g_Y^{'}(t)=-S_Y(t)$ for almost every $t$. A direct computation gives $g_{E_{\mu}}(t)=\mu\exp(-t/\mu)$.

\medskip
\noindent\emph{The $\NBUE$ case.} By Definition~\ref{def:NBUENWUE} and Definition~\ref{def:meanexcessfunction}, the $\NBUE$ property reads $g_Y(t)\le\mu\,S_Y(t)$ for every $t$ with $S_Y(t)>0$, that is,
\[
g_Y^{'}(t)=-S_Y(t)\le-\dfrac{g_Y(t)}{\mu}
\qquad\text{for almost every $t$ with }g_Y(t)>0 .
\]
On the interval where $g_Y>0$ we therefore have
\[
\dfrac{\d}{\d t}\Big(\log g_Y(t)+\dfrac{t}{\mu}\Big)\le0 ,
\]
so that $\log g_Y(t)+t/\mu\le\log g_Y(0)=\log\mu$, whence
\begin{equation}
\label{eq:stoplossbound}
g_Y(t)\le\mu\exp\!\left(-\dfrac{t}{\mu}\right)=g_{E_{\mu}}(t),\qquad t\ge0 .
\end{equation}
If $g_Y$ vanishes from some point onwards, \eqref{eq:stoplossbound} holds trivially there, so the bound is valid for all $t\ge0$. Since $Y$ and $E_{\mu}$ are non-negative with the same mean $\mu$, stop-loss dominance \eqref{eq:stoplossbound} together with equality of the means is exactly the statement $Y\le_{\cx}E_{\mu}$ \citep[Section~3.A]{Shaked2007}.

\medskip
\noindent\emph{The $\NWUE$ case.} The $\NWUE$ property reads $g_Y(t)\ge\mu\,S_Y(t)$, hence $g_Y^{'}(t)\ge-g_Y(t)/\mu$ almost everywhere, and the same argument with the inequalities reversed yields $g_Y(t)\ge\mu\exp(-t/\mu)$ for all $t\ge0$, that is, $Y\ge_{\cx}E_{\mu}$.

\medskip
\noindent\emph{Consequences for the coefficient of variation.} Applying Tonelli's theorem as in Appendix~A.\ref{app:OrderNonNegative} gives the identity
\[
\int_0^{\infty}g_Y(t)\,\d t=\int_0^{\infty}\!\!\int_t^{\infty}S_Y(w)\,\d w\,\d t=\int_0^{\infty}w\,S_Y(w)\,\d w=\dfrac{1}{2}\,\E[Y^2].
\]
Integrating \eqref{eq:stoplossbound} over $t\in[0,\infty)$ therefore yields, in the $\NBUE$ case,
\[
\dfrac{1}{2}\,\E[Y^2]=\int_0^{\infty}g_Y(t)\,\d t\le\int_0^{\infty}\mu\exp\!\left(-\dfrac{t}{\mu}\right)\d t=\mu^2 ,
\]
so that $\E[Y^2]\le2\mu^2$, $\Var[Y]\le\mu^2$ and $\CV[Y]\le1$. The $\NWUE$ case is identical with the inequalities reversed.

\medskip
\noindent\emph{The equality case.} Suppose $Y$ is $\NBUE$ with $\CV[Y]=1$. Then $\Var[Y]=\mu^2$, hence $\tfrac12\E[Y^2]=\mu^2$, and the display above becomes
\[
\int_0^{\infty}\Big(\mu\exp\!\left(-\dfrac{t}{\mu}\right)-g_Y(t)\Big)\d t=0 .
\]
The integrand is non-negative by \eqref{eq:stoplossbound} and continuous, since $g_Y$ is convex and finite on $[0,\infty)$. It therefore vanishes identically, giving $g_Y(t)=\mu\exp(-t/\mu)$ for every $t\ge0$. Differentiating yields $S_Y(t)=-g_Y^{'}(t)=\exp(-t/\mu)$ for almost every, and hence by right-continuity for every, $t\ge0$. Consequently $Y\eqd E_{\mu}$. The $\NWUE$ case follows by the same argument with the roles of the two stop-loss transforms interchanged. This completes the proof.
\end{proof}

\subsection{Proof of Proposition \ref{prop:m_Y_Representation}}
\label{app:m_Y_Representation}
\begin{proof}
We derive both representations by expressing the DDF $S_Y(y)$ in terms of $F_X(x)$ and $S_X(x)$ and substituting it back into the standard representation of $m_Y(t)$.

Since $Y=|X-X'|$ with $X$ and $X'$ independent and identically distributed, the events $\{X-X'>y\}$ and $\{X'-X>y\}$ are disjoint for $y\ge0$ and have equal probability by exchangeability; as $X$ is absolutely continuous, $\mathbb{P}(X=X')=0$ and hence
\[
S_Y(y)=\mathbb{P}(|X-X'|>y)=2\,\mathbb{P}(X-X'>y),\qquad y\ge0 .
\]
Conditioning on one variable at a time and using independence gives the two equivalent forms
\begin{equation}
\label{eq:SYtwoforms}
S_Y(y)
=2\int_{-\infty}^{\infty}S_X(x+y)\,\d F_X(x)
=2\int_{-\infty}^{\infty}F_X(x-y)\,\d F_X(x).
\end{equation}
Since $F_X(X)$ is uniformly distributed on $(0,1)$, we have $\E[F_X(X)]=\E[S_X(X)]=\tfrac12$, so that the tilted measures of the statement take the explicit forms
\begin{equation}
\label{eq:tiltedmeasures}
\d Q^F(x)=2F_X(x)f_X(x)\,\d x
\qquad\text{and}\qquad
\d Q^{S}(x)=2S_X(x)f_X(x)\,\d x .
\end{equation}

\medskip
\noindent\emph{First representation.} Substituting the first expression in \eqref{eq:SYtwoforms} into the integral defining $m_Y(t)$ yields
\begin{align*}
m_Y(t)S_Y(t)
&=\int_t^\infty S_Y(y)\,\d y
\\
&=2\int_t^\infty\left(\int_{-\infty}^{\infty}S_X(x+y)\,\d F_X(x)\right)\d y
\\
&=2\int_{-\infty}^{\infty}\left(\int_{x+t}^\infty S_X(w)\,\d w\right)\d F_X(x)
\\
&=2\int_{-\infty}^{\infty}\left(\int_{-\infty}^{w-t}\d F_X(x)\right)S_X(w)\,\d w
\\
&=2\int_{-\infty}^{\infty}F_X(w-t)S_X(w)\,\d w
\\
&=2\int_{-\infty}^{\infty}
\dfrac{F_X(w-t)}{F_X(w)}\,
\dfrac{S_X(w)}{f_X(w)}\,
F_X(w)f_X(w)\,\d w
\\
&=2\int_{-\infty}^{\infty}
C_X(w,t)\,h_X(w)^{-1}\,
F_X(w)f_X(w)\,\d w
\\
&=\E^{F}\big[C_X(X,t)\,h_X(X)^{-1}\big],
\end{align*}
where the third equality substitutes $w=x+y$ and applies Tonelli's theorem, the fourth exchanges the order of integration over the region $\{(x,w):\,x<w-t\}$, the sixth multiplies and divides by $F_X(w)f_X(w)$ -- legitimate since $f_X$ and $F_X$ are strictly positive on the interior of the support of $X$, outside of which the integrand vanishes -- and the last uses \eqref{eq:tiltedmeasures}.

Using the second expression in \eqref{eq:SYtwoforms} for the left-hand factor gives
\[
S_Y(t)
=2\int_{-\infty}^{\infty}F_X(x-t)\,\d F_X(x)
=2\int_{-\infty}^{\infty}\dfrac{F_X(x-t)}{F_X(x)}\,F_X(x)f_X(x)\,\d x
=\E^{F}\big[C_X(X,t)\big].
\]
Dividing the two displays yields
\[
m_Y(t)=\dfrac{\E^{F}\big[C_X(X,t)\,h_X(X)^{-1}\big]}{\E^{F}\big[C_X(X,t)\big]}.
\]

\medskip
\noindent\emph{Second representation.} We now interchange the roles of the two expressions in \eqref{eq:SYtwoforms}. Substituting the second expression into the integral defining $m_Y(t)$ gives
\begin{align*}
m_Y(t)S_Y(t)
&=\int_t^\infty S_Y(y)\,\d y
\\
&=2\int_t^\infty\left(\int_{-\infty}^{\infty}F_X(x-y)\,\d F_X(x)\right)\d y
\\
&=2\int_{-\infty}^{\infty}\left(\int_{-\infty}^{x-t}F_X(w)\,\d w\right)\d F_X(x)
\\
&=2\int_{-\infty}^{\infty}F_X(w)\left(\int_{w+t}^{\infty}\d F_X(x)\right)\d w
\\
&=2\int_{-\infty}^{\infty}F_X(w)\,S_X(w+t)\,\d w
\\
&=2\int_{-\infty}^{\infty}
\dfrac{S_X(w+t)}{S_X(w)}\,
\dfrac{F_X(w)}{f_X(w)}\,
S_X(w)f_X(w)\,\d w
\\
&=2\int_{-\infty}^{\infty}
D_X(w,t)\,r_X(w)^{-1}\,
S_X(w)f_X(w)\,\d w
\\
&=\E^{S}\big[D_X(X,t)\,r_X(X)^{-1}\big],
\end{align*}
the third equality substituting $w=x-y$ and the fourth exchanging the order of integration over the region $\{(x,w):\,w<x-t\}$. Using the first expression in \eqref{eq:SYtwoforms} for the left-hand factor gives
\[
S_Y(t)
=2\int_{-\infty}^{\infty}S_X(x+t)\,\d F_X(x)
=2\int_{-\infty}^{\infty}\dfrac{S_X(x+t)}{S_X(x)}\,S_X(x)f_X(x)\,\d x
=\E^{S}\big[D_X(X,t)\big],
\]
and dividing yields
\[
m_Y(t)=\dfrac{\E^{S}\big[D_X(X,t)\,r_X(X)^{-1}\big]}{\E^{S}\big[D_X(X,t)\big]}.
\]
This establishes both representations in \eqref{eq:m_YRepresentation} and completes the proof.
\end{proof}
\subsection{An auxiliary correlation inequality}
\label{app:chebyshev}
The following classical correlation inequality, in the form due to Chebyshev and Hoeffding, is used repeatedly in Sections~\ref{sec:SDgeGMD} and~\ref{sec:SDleGMD}; see \cite{Hardy1988} for the discrete and integral versions. We state it with the equality case, which is required for Theorem~\ref{thm:SDgeGMDstrict}.

\begin{lemma}
\label{lem:chebyshev}
Let $Z$ be a random variable defined on a probability space, and let $g,k:\RR\to\RR$ be Borel measurable with $g$ bounded and $\E|k(Z)|<\infty$.
\begin{itemize}
\item[(i)] If $g$ and $k$ are similarly ordered, that is, both non-decreasing or both non-increasing, then
\[
\E\big[g(Z)k(Z)\big]\ \ge\ \E\big[g(Z)\big]\,\E\big[k(Z)\big].
\]
\item[(ii)] If $g$ and $k$ are oppositely ordered, the inequality is reversed.
\item[(iii)] In either case, equality holds if and only if
\[
\big(g(Z)-g(Z')\big)\big(k(Z)-k(Z')\big)=0
\qquad\text{almost surely,}
\]
where $Z'$ is an independent copy of $Z$. In particular, equality holds whenever one of $g(Z)$, $k(Z)$ is almost surely constant.
\end{itemize}
\end{lemma}

\begin{proof}
Let $Z'$ be an independent copy of $Z$. Since $g$ is bounded and $k(Z)$ is integrable, every expectation below is finite, and expanding the product gives
\begin{align*}
\E\Big[\big(g(Z)-g(Z')\big)\big(k(Z)-k(Z')\big)\Big]
&=\E\big[g(Z)k(Z)\big]-\E\big[g(Z)k(Z')\big]
\\
&\qquad-\E\big[g(Z')k(Z)\big]+\E\big[g(Z')k(Z')\big]
\\
&=2\Big(\E\big[g(Z)k(Z)\big]-\E\big[g(Z)\big]\E\big[k(Z)\big]\Big),
\end{align*}
where the middle two terms were evaluated using independence and the fact that $Z$ and $Z'$ are identically distributed.

If $g$ and $k$ are similarly ordered, then for every pair $(z,z')$ the two factors $g(z)-g(z')$ and $k(z)-k(z')$ carry the same sign, so their product is non-negative; the displayed expectation is therefore non-negative, which is assertion \emph{(i)}. If they are oppositely ordered the product is non-positive, giving \emph{(ii)}. In both cases the expectation of a random variable of constant sign vanishes if and only if that random variable is almost surely zero, which is assertion \emph{(iii)}. The final statement is immediate, since constancy of either factor makes the product identically zero.
\end{proof}

\subsection{Proof of Proposition \ref{prop:r_Xh_Ximplications}}
\label{app:r_Xh_Ximplications}
\begin{proof}
Throughout, let $(a,b)$ denote the interior of the support of $X$, with $-\infty\le a<b\le+\infty$.

\medskip
\noindent\emph{Monotonicity.} We begin by expressing each hazard-type function in terms of the other:
\[
r_X(x) = h_X(x)\,\frac{S_X(x)}{F_X(x)} \quad \text{and} \quad
h_X(x) = r_X(x)\,\frac{F_X(x)}{S_X(x)}.
\]
Observe that $x\mapsto S_X(x)/F_X(x)$ is always decreasing and non-negative, while $x\mapsto F_X(x)/S_X(x)$ is always increasing and non-negative. Consequently:
\begin{itemize}
    \item[(i)] If $h_X(x)$ is decreasing, then $r_X(x) = h_X(x)\,S_X(x)/F_X(x)$ is a product of two non-negative decreasing functions, hence decreasing.
    \item[(ii)] Analogously, if $r_X(x)$ is increasing, then $h_X(x) = r_X(x)\,F_X(x)/S_X(x)$ is a product of two non-negative increasing functions, hence increasing.
\end{itemize}

\medskip
We now turn to the support statements. We prove claim \emph{(i)}; claim \emph{(ii)} follows by applying \emph{(i)} to the reflected variable $-X$ and invoking Remark~\ref{rmk:h_Xr_X_Mirror}.

\medskip
\noindent\emph{$X$ is bounded below.} By Proposition~\ref{prop:equivalences}, a decreasing $h_X(x)$ is equivalent to log-convexity of the DDF $S_X(x)$. Suppose, for the sake of contradiction, that $S_X(x)$ is log-convex and that $a=-\infty$. Since $X$ is non-degenerate we may fix $x_1<x_2$ in the support with $S_X(x_1)>S_X(x_2)>0$. Writing $\varphi=\log S_X$, which is convex by assumption, the monotonicity of secant slopes of a convex function gives, for any $t\ge0$,
\[
\frac{\varphi(x_2)-\varphi(x_1)}{x_2-x_1}
\ \ge\
\frac{\varphi(x_2)-\varphi(x_1-t)}{x_2-x_1+t}.
\]
Rearranging yields
\[
\varphi(x_1-t)
\ \ge\
\varphi(x_2)
+
\left(1+\frac{t}{x_2-x_1}\right)
\big(\varphi(x_1)-\varphi(x_2)\big)
=
\varphi(x_1)+\frac{t}{x_2-x_1}\big(\varphi(x_1)-\varphi(x_2)\big).
\]
Since $\varphi(x_1)-\varphi(x_2)>0$, letting $t\to\infty$ causes the right-hand side to diverge to $+\infty$. However, the left-hand side satisfies $\varphi(x_1-t)=\log S_X(x_1-t)\le0$ for every $t$, and indeed $\varphi(x_1-t)\to0$ as $t\to\infty$ because $S_X(x)\to1$ as $x\to-\infty$. This contradiction shows that $a>-\infty$, that is, $X$ is bounded below.

\medskip
\noindent\emph{$X$ is unbounded above.} Suppose that $h_X(x)$ is decreasing and, toward a contradiction, that $b<\infty$. Recall the standard representation
\[
S_X(x)=\exp\left(-\int_{a}^x h_X(t)\,\d t\right),
\qquad x\in(a,b),
\]
and note that $S_X(x)\to0$ as $x\uparrow b$, by definition of the upper endpoint of the support, so that
\[
\lim_{x\uparrow b}\int_{a}^x h_X(t)\,\d t=\infty.
\]
We show that this is impossible. Fix any interior point $x_0\in(a,b)$. On the one hand,
\[
\int_a^{x_0}h_X(t)\,\d t=-\log S_X(x_0)<\infty ,
\]
since $S_X(x_0)>0$. On the other hand, because $h_X$ is decreasing on $(a,b)$ and $h_X(x_0)<\infty$,
\[
\int_{x_0}^{b}h_X(t)\,\d t\ \le\ (b-x_0)\,h_X(x_0)\ <\ \infty ,
\]
the finiteness of the bound being a consequence of $b<\infty$. Adding the two displays gives $\int_a^b h_X(t)\,\d t<\infty$, whence $S_X(b^-)=\exp\big(-\int_a^bh_X\big)>0$, contradicting $S_X(b^-)=0$. We emphasize that a decreasing hazard rate may well be unbounded, but only at the lower endpoint $a$, which is precisely why the integral must be split at an interior point. Hence $b=\infty$ and $X$ is unbounded above.

The proof of claim \emph{(i)} is complete.
\end{proof}

\subsection{Proof of Theorem \ref{thm:SDgeGMD}}
\label{app:SDgeGMD}
\begin{proof}
We prove the assertion under the assumption that $h_X(x)$ is decreasing; the case in which $r_X(x)$ is increasing follows by applying the argument to the reflected variable $-X$ and invoking Remark~\ref{rmk:h_Xr_X_Mirror}, both dispersion measures being invariant under reflection.

\medskip
Since $h_X(x)$ is decreasing, Proposition~\ref{prop:r_Xh_Ximplications} implies that $r_X(x)$ is also decreasing.
By Proposition~\ref{prop:equivalences}, part \emph{(B1)}$\iff$\emph{(B2)}, it then follows that, for each $t\ge 0$, the function $x\mapsto C_X(x,t)$ is increasing.

\medskip
Recall from Proposition~\ref{prop:m_Y_Representation} that the mean excess function $m_Y(t)$ admits the representation
\[
m_Y(t)
=\frac{\E^{F}\!\left[C_X(X,t)\,h_X(X)^{-1}\right]}
{\E^{F}[C_X(X,t)]}.
\]
We apply Lemma~\ref{lem:chebyshev} under the probability measure $Q^F$, with
\[
g(x)=C_X(x,t)
\qquad\text{and}\qquad
k(x)=h_X(x)^{-1}.
\]
The function $g$ is increasing, as just established, and bounded, since $C_X(x,t)=F_X(x-t)/F_X(x)\in[0,1]$ for $t\ge0$. The function $k$ is increasing because $h_X$ is decreasing and positive, and it is $Q^F$-integrable because, by Corollary~\ref{cor:GMDhazard}, $\E^{F}[h_X(X)^{-1}]=\GMD[X]\le\sqrt{2}\,\SD[X]<\infty$. The two functions are therefore similarly ordered and satisfy the hypotheses of Lemma~\ref{lem:chebyshev}\emph{(i)}, which yields
\[
m_Y(t)
=\frac{\E^{F}\!\left[C_X(X,t)\,h_X(X)^{-1}\right]}
{\E^{F}[C_X(X,t)]}
\ \ge\
\frac{\E^{F}\!\left[C_X(X,t)\right]\E^{F}\!\left[h_X(X)^{-1}\right]}
{\E^{F}[C_X(X,t)]}
=
\E^{F}\!\left[h_X(X)^{-1}\right].
\]
Since $C_X(x,0)=F_X(x)/F_X(x)=1$, evaluating the representation at $t=0$ gives $\E^{F}[h_X(X)^{-1}]=m_Y(0)$, and therefore
\[
m_Y(t)\ \ge\ m_Y(0), \qquad\text{for all } t\ge 0.
\]
Equivalently, the spacing $Y$ is $\NWUE$ in the sense of Definition~\ref{def:NBUENWUE}.

\medskip
Finally, since $X$ is absolutely continuous the spacing $Y$ has no atom at the origin, so Proposition~\ref{prop:OrderNonNegative} applies and the condition $m_Y(t)\ge m_Y(0)$ for all $t\ge 0$ entails
\[
\SD[X]\ge \GMD[X],
\]
which completes the proof.
\end{proof}

\subsection{Proof of Corollary \ref{cor:GMDhazard}}
\label{app:GMDhazard}
\begin{proof}
Evaluating the two representations of Proposition~\ref{prop:m_Y_Representation} at $t=0$ and using
\[
C_X(x,0)=\dfrac{F_X(x)}{F_X(x)}=1
\qquad\text{and}\qquad
D_X(x,0)=\dfrac{S_X(x)}{S_X(x)}=1
\]
gives immediately
\[
m_Y(0)=\E^{F}\!\left[h_X(X)^{-1}\right]=\E^{S}\!\left[r_X(X)^{-1}\right].
\]
By \eqref{eq:mY0}, $m_Y(0)=\E[Y]=\GMD[X]$, which establishes the first pair of identities. The second pair follows from Remark~\ref{rmk:orderstatistics}, which identifies $Q^F$ and $Q^S$ as the laws of $X_{(2)}$ and $X_{(1)}$ respectively.

As a consistency check, note that
\[
\E^{F}\!\left[h_X(X)^{-1}\right]
=2\int_{-\infty}^{\infty}\dfrac{S_X(x)}{f_X(x)}\,F_X(x)f_X(x)\,\d x
=2\int_{-\infty}^{\infty}F_X(x)S_X(x)\,\d x ,
\]
which is exactly the tail representation of the GMD established in Proposition~\ref{prop:GMDdual}\emph{(iii)}.
\end{proof}

\subsection{Proof of Theorem \ref{thm:SDgeGMDstrict}}
\label{app:SDgeGMDstrict}
\begin{proof}
We treat the case of a decreasing $h_X(x)$; the case of an increasing $r_X(x)$ follows by reflection, as in Remark~\ref{rmk:h_Xr_X_Mirror}.

\medskip
\noindent\emph{Sufficiency.} If $h_X(x)\equiv\theta$ is constant then, by Proposition~\ref{prop:r_Xh_Ximplications}, the support of $X$ is $(a,\infty)$ with $a$ finite, and $S_X(x)=\exp(-\theta(x-a))$, so that $X-a\sim\Expo(\theta)$. Proposition~\ref{prop:expbenchmark}, together with the affine invariance of both measures, gives $\SD[X]=\GMD[X]=\theta^{-1}$.

\medskip
\noindent\emph{Necessity.} Suppose $\SD[X]=\GMD[X]$. Since $h_X$ is decreasing, Proposition~\ref{prop:r_Xh_Ximplications} shows that $r_X$ is decreasing as well, so both hazard-type functions are monotone and Lemma~\ref{lem:equalitydichotomy} applies: either $h_X$ is constant on the support, in which case $X$ is a shifted exponential and we are done, or $r_X$ is constant on the support, in which case the same lemma gives a finite upper endpoint $b<\infty$. The second alternative is impossible, because Proposition~\ref{prop:r_Xh_Ximplications} asserts that a decreasing hazard rate forces the support to be unbounded above. Hence $h_X$ is constant and $X$ is an exponential random variable up to a location shift.

Contrapositively, if $X$ is not of this form then the inequality of Theorem~\ref{thm:SDgeGMD} is strict, which is the assertion.
\end{proof}

\subsection{Proof of Proposition \ref{prop:h_xr_X_mixtures}}
\label{app:h_xr_X_mixtures}
\begin{proof}
We prove the case of decreasing hazard rates; the argument for increasing reverse hazard rates is verbatim the same with the DDF $S$ replaced by the CDF $F$ throughout, using the equivalence \emph{(B1)}$\iff$\emph{(B3)} of Proposition~\ref{prop:equivalences} in place of \emph{(A1)}$\iff$\emph{(A3)}.

By Proposition~\ref{prop:equivalences}, a decreasing hazard rate is equivalent to the corresponding DDF being log-convex. Let $a$ denote the common lower endpoint. Since each $h_{X_{\theta}}(x)$ is decreasing, it follows that each $S_{X_{\theta}}(x)$ is log-convex on $[a,\infty)$. Therefore, it suffices to show that the mixture DDF
\[
S_X(x)=\int S_{X_{\theta}}(x)\,\d F_{\Theta}(\theta)
\]
is also log-convex on $[a,\infty)$, which contains the support of $X$. That is, for any $x_1,x_2\ge a$ and $\lambda\in(0,1)$, we must verify that
\[
S_X(\lambda x_1+(1-\lambda)x_2)\le S_X(x_1)^{\lambda}S_X(x_2)^{1-\lambda}.
\]
Using the representation of $S_X(x)$ and the log-convexity of each $S_{X_{\theta}}(x)$, we obtain
\begin{align*}
S_X(\lambda x_1+(1-\lambda)x_2)
&=\int S_{X_{\theta}}(\lambda x_1+(1-\lambda)x_2)\,\d F_{\Theta}(\theta)
\\
&\le \int S_{X_{\theta}}(x_1)^{\lambda}S_{X_{\theta}}(x_2)^{1-\lambda}\,\d F_{\Theta}(\theta).
\end{align*}

Applying H\"older's inequality (see, for example, \cite{Hardy1988}) to the functions
$g(\theta)=S_{X_{\theta}}(x_1)^{\lambda}$ and
$k(\theta)=S_{X_{\theta}}(x_2)^{1-\lambda}$,
with conjugate exponents $p=1/\lambda$ and $q=1/(1-\lambda)$, yields
\begin{align*}
S_X(\lambda x_1+(1-\lambda)x_2)
&\le \int S_{X_{\theta}}(x_1)^{\lambda}S_{X_{\theta}}(x_2)^{1-\lambda}\,\d F_{\Theta}(\theta)
\\
&\le \left(\int S_{X_{\theta}}(x_1)\,\d F_{\Theta}(\theta)\right)^{\lambda}
     \left(\int S_{X_{\theta}}(x_2)\,\d F_{\Theta}(\theta)\right)^{1-\lambda}
\\
&= S_X(x_1)^{\lambda}S_X(x_2)^{1-\lambda}.
\end{align*}
Thus $S_X(x)$ is log-convex, and by Proposition~\ref{prop:equivalences} the hazard rate $h_X(x)$ of the mixture is decreasing. The boundary cases $\lambda\in\{0,1\}$ are trivial, and the result follows. The common endpoint is essential: a component supported on $[a_\theta,\infty)$ with $a_\theta>a$ has $S_{X_\theta}\equiv1$ to the left of $a_\theta$, hence a log-concave kink at $a_\theta$, and the H\"older step fails there. No further assumption on the mixing law $F_\Theta$ is needed; see \cite{Proschan1963} and \cite{Barlow1975}.
\end{proof}

\subsection{Proof of Proposition \ref{prop:gammaratio}}
\label{app:gammaratio}
\begin{proof}
By the affine invariance of both dispersion measures it suffices to treat the unit-scale case $\theta=1$; the general case follows upon multiplication by $\theta$, which leaves $\mathcal{R}[X]$ unchanged.

\medskip
\noindent\emph{Closed form for the GMD.} Let $X,X'$ be independent $\mathrm{Gamma}(\alpha,1)$ variables, and set $T=X+X'$ and $B=X/(X+X')$. It is classical that $T\sim\mathrm{Gamma}(2\alpha,1)$ and $B\sim\mathrm{Beta}(\alpha,\alpha)$ are independent. Since $|X-X'|=T\,|2B-1|$,
\[
\GMD[X]=\E[T]\,\E|2B-1|=2\alpha\,\E|W|,\qquad W:=2B-1,
\]
where $W$ has density proportional to $(1-w^2)^{\alpha-1}$ on $(-1,1)$. The substitution $t=w^2$ gives $\int_{-1}^{1}(1-w^2)^{\alpha-1}\,\d w=\mathrm{B}\big(\tfrac12,\alpha\big)=\sqrt{\pi}\,\Gamma(\alpha)/\Gamma\big(\alpha+\tfrac12\big)$, while $\int_{-1}^{1}|w|(1-w^2)^{\alpha-1}\,\d w=2\int_0^1 w(1-w^2)^{\alpha-1}\,\d w=1/\alpha$. Hence
\[
\E|W|=\dfrac{\Gamma\!\left(\alpha+\tfrac12\right)}{\alpha\,\sqrt{\pi}\,\Gamma(\alpha)}
\qquad\text{and}\qquad
\GMD[X]=\dfrac{2\,\Gamma\!\left(\alpha+\tfrac12\right)}{\sqrt{\pi}\,\Gamma(\alpha)} .
\]
Since $\Var[X]=\alpha$, we obtain $\SD[X]=\sqrt{\alpha}$ and hence
\[
\mathcal{R}[X]=\dfrac{\sqrt{\alpha}\,\sqrt{\pi}\,\Gamma(\alpha)}{2\,\Gamma\!\left(\alpha+\tfrac12\right)}
=\dfrac{\sqrt{\pi\alpha}\,\Gamma(\alpha)}{2\,\Gamma\!\left(\alpha+\tfrac12\right)} .
\]
As a verification, $\alpha=\tfrac12$ gives $\GMD[X]=2/\pi$, $\alpha=1$ gives $\GMD[X]=1=\SD[X]$, and $\alpha=2$ gives $\GMD[X]=\tfrac32$, each of which may be confirmed by direct integration of $2\int_0^\infty F_XS_X\,\d x$.

\medskip
\noindent\emph{Strict monotonicity.} Taking logarithms,
\[
\log\mathcal{R}[X]=\tfrac12\log\pi+\tfrac12\log\alpha+\log\Gamma(\alpha)-\log 2-\log\Gamma\!\left(\alpha+\tfrac12\right),
\]
and differentiating with respect to $\alpha$ gives
\[
\dfrac{\d}{\d\alpha}\log\mathcal{R}[X]
=\dfrac{1}{2\alpha}+\psi(\alpha)-\psi\!\left(\alpha+\tfrac12\right),
\]
where $\psi=\Gamma'/\Gamma$ denotes the digamma function. Since
\[
\psi''(x)=-2\sum_{n\ge0}\dfrac{1}{(x+n)^{3}}<0
\qquad\text{for } x>0,
\]
the function $\psi$ is strictly concave on $(0,\infty)$, and the strict midpoint inequality gives
\[
\psi\!\left(\alpha+\tfrac12\right)>\tfrac12\Big(\psi(\alpha)+\psi(\alpha+1)\Big).
\]
Combining this with the recurrence $\psi(\alpha+1)-\psi(\alpha)=1/\alpha$ yields
\[
\psi\!\left(\alpha+\tfrac12\right)-\psi(\alpha)>\tfrac12\big(\psi(\alpha+1)-\psi(\alpha)\big)=\dfrac{1}{2\alpha},
\]
and therefore $\frac{\d}{\d\alpha}\log\mathcal{R}[X]<0$ for every $\alpha>0$. Hence $\mathcal{R}[X]$ is strictly decreasing.

\medskip
\noindent\emph{Limits.} As $\alpha\downarrow0$ we have $\Gamma(\alpha)\sim1/\alpha$ and $\Gamma(\alpha+\tfrac12)\to\sqrt{\pi}$, so that
\[
\mathcal{R}[X]\sim\dfrac{\sqrt{\pi\alpha}}{2\sqrt{\pi}\,\alpha}=\dfrac{1}{2\sqrt{\alpha}}\longrightarrow\infty .
\]
At $\alpha=1$, $\Gamma(1)=1$ and $\Gamma(\tfrac32)=\sqrt{\pi}/2$, giving $\mathcal{R}[X]=1$. As $\alpha\to\infty$, Wendel's asymptotic $\Gamma(\alpha+\tfrac12)/\Gamma(\alpha)\sim\sqrt{\alpha}$ gives
\[
\mathcal{R}[X]\sim\dfrac{\sqrt{\pi\alpha}}{2\sqrt{\alpha}}=\dfrac{\sqrt{\pi}}{2},
\]
which is precisely the dispersion ratio of the standard normal law, consistent with the central limit behavior of the standardized Gamma distribution and with Example~\ref{ex:SDleGMD}\emph{(3)}.

\medskip
\noindent\emph{Conclusion.} Strict monotonicity together with $\mathcal{R}[X]\big|_{\alpha=1}=1$ gives $\mathcal{R}[X]>1$ for $\alpha<1$ and $\mathcal{R}[X]<1$ for $\alpha>1$, which is the stated equivalence. We note that the SD-dominance half of the conclusion also follows from Theorem~\ref{thm:SDgeGMD}, since the Gamma hazard rate is decreasing for $\alpha\le1$, and the GMD-dominance half from Proposition~\ref{prop:logConcave_f_X}, since the Gamma density is log-concave for $\alpha\ge1$; the present proposition strengthens both into an exact characterization within the family.
\end{proof}
\subsection{Proof of Theorem \ref{thm:SDleGMD}}
\label{app:SDleGMD}
\begin{proof}
From Proposition~\ref{prop:m_Y_Representation}, the mean excess function $m_Y(t)$ admits the representation
\[
m_Y(t)
=\frac{\E^{F}\!\left[C_X(X,t)\,h_X(X)^{-1}\right]}
{\E^{F}[C_X(X,t)]}.
\]
Since $h_X(x)$ is increasing and positive, its reciprocal $h_X(x)^{-1}$ is decreasing. Moreover, by Proposition~\ref{prop:equivalences}, part \emph{(B1)}$\iff$\emph{(B2)}, the assumption that $r_X(x)$ is decreasing is equivalent to $x\mapsto C_X(x,t)$ being increasing for every $t\ge0$. The two functions
\[
g(x)=C_X(x,t)
\qquad\text{and}\qquad
k(x)=h_X(x)^{-1}
\]
are therefore oppositely ordered. As in Appendix~A.\ref{app:SDgeGMD}, $g$ is bounded, since $C_X(x,t)\in[0,1]$, and $k$ is $Q^F$-integrable, since $\E^{F}[h_X(X)^{-1}]=\GMD[X]<\infty$ by Corollary~\ref{cor:GMDhazard}. Applying Lemma~\ref{lem:chebyshev}\emph{(ii)} under the probability measure $Q^F$ therefore yields
\[
m_Y(t)
=\frac{\E^{F}\!\left[C_X(X,t)\,h_X(X)^{-1}\right]}
{\E^{F}[C_X(X,t)]}
\ \le\
\frac{\E^{F}\!\left[C_X(X,t)\right]\E^{F}\!\left[h_X(X)^{-1}\right]}
{\E^{F}[C_X(X,t)]}
=
\E^{F}\!\left[h_X(X)^{-1}\right]
= m_Y(0),
\]
the final equality following from $C_X(x,0)=1$. Hence
\[
m_Y(t)\le m_Y(0), \qquad \text{for all } t\ge 0,
\]
that is, the spacing $Y$ is $\NBUE$ in the sense of Definition~\ref{def:NBUENWUE}.

\medskip
Since $X$ is absolutely continuous, $Y$ has no atom at the origin, so Proposition~\ref{prop:OrderNonNegative} applies and the inequality $m_Y(t)\le m_Y(0)$ for all $t\ge0$ implies
\[
\SD[X]\le \GMD[X],
\]
which completes the proof.

\medskip
\noindent\emph{Alternative route.} The same conclusion follows from the second representation in \eqref{eq:m_YRepresentation}. By Proposition~\ref{prop:equivalences}, part \emph{(A1)}$\iff$\emph{(A2)}, an increasing $h_X$ is equivalent to $x\mapsto D_X(x,t)$ being decreasing, while a decreasing $r_X$ makes $r_X^{-1}$ increasing; applying Lemma~\ref{lem:chebyshev}\emph{(ii)} under $Q^{S}$ gives $m_Y(t)\le\E^{S}[r_X(X)^{-1}]=m_Y(0)$. Both hypotheses of the theorem are thus used symmetrically, once through each representation.
\end{proof}

\subsection{An equality dichotomy}
\label{app:equalitydichotomy}
The following lemma isolates the structure of the equality case common to Theorems~\ref{thm:SDgeGMD} and~\ref{thm:SDleGMD}. It supplies the proofs of Theorems~\ref{thm:SDgeGMDstrict} and~\ref{thm:SDleGMDstrict}, and, applied to the heavy-tailed regime, it also furnishes a shorter derivation of the former than the direct argument of Appendix~A.\ref{app:SDgeGMDstrict}.

\begin{lemma}
\label{lem:equalitydichotomy}
Let $X$ be absolutely continuous and non-degenerate with $\E[X^2]<\infty$, whose support has interior $(a,b)$ with $-\infty\le a<b\le+\infty$. Assume that $h_X$ and $r_X$ are both monotone, so that the hypotheses of either Theorem~\ref{thm:SDgeGMD} or Theorem~\ref{thm:SDleGMD} are satisfied. If
\[
\SD[X]=\GMD[X],
\]
then either $h_X$ is constant on $(a,b)$, in which case $a>-\infty$ and $X-a$ is exponentially distributed, or $r_X$ is constant on $(a,b)$, in which case $b<\infty$ and $b-X$ is exponentially distributed.
\end{lemma}

\begin{proof}
\emph{Step 1: the mean excess function of the spacing is constant.} By the proofs of Theorems~\ref{thm:SDgeGMD} and~\ref{thm:SDleGMD}, the quantity $m_Y(t)-m_Y(0)$ has a constant sign on $[0,\infty)$; denote it by $\varepsilon\in\{+1,-1\}$, so that $\varepsilon\big(m_Y(t)-m_Y(0)\big)\ge0$ for every $t$. From the proof of Proposition~\ref{prop:OrderNonNegative} together with identity \eqref{eq:mY0},
\[
\dfrac{1}{2}\E[Y^2]=\int_0^{\infty}m_Y(t)S_Y(t)\,\d t
\qquad\text{and}\qquad
\E[Y]^2=m_Y(0)\int_0^{\infty}S_Y(t)\,\d t ,
\]
so that the assumed equality $\tfrac12\E[Y^2]=\E[Y]^2$ gives
\[
\int_0^{\infty}\varepsilon\big(m_Y(t)-m_Y(0)\big)S_Y(t)\,\d t=0 .
\]
The integrand is non-negative. The spacing $Y=|X-X'|$ has support $[0,\,b-a)$, so $S_Y(t)>0$ for every $t\in[0,b-a)$; and $m_Y$ is continuous there, since $Y$ is absolutely continuous with $S_Y$ continuous and strictly positive. Consequently
\begin{equation}
\label{eq:mYconstant}
m_Y(t)=m_Y(0)
\qquad\text{for every } t\in[0,\,b-a).
\end{equation}

\medskip
\noindent\emph{Step 2: one of the two factors is constant.} Fix $t\in[0,b-a)$. By \eqref{eq:mYconstant} and the proofs of Theorems~\ref{thm:SDgeGMD} and~\ref{thm:SDleGMD}, equality holds in Lemma~\ref{lem:chebyshev} applied under $Q^F$ to $g_t=C_X(\cdot,t)$ and $k=h_X^{-1}$. Letting $Z,Z'$ be independent with law $Q^F$, part \emph{(iii)} of that lemma gives
\begin{equation}
\label{eq:productzero}
\big(g_t(Z)-g_t(Z')\big)\big(k(Z)-k(Z')\big)=0
\qquad\text{almost surely.}
\end{equation}

Suppose that $k$ is not $Q^F$-almost surely constant. Then there exists $\kappa\in\RR$ with
\[
Q^F\big(\{k<\kappa\}\big)>0
\qquad\text{and}\qquad
Q^F\big(\{k>\kappa\}\big)>0 ;
\]
write $N=\{k<\kappa\}$ and $P=\{k>\kappa\}$. For $(z,z')\in N\times P$ the second factor in \eqref{eq:productzero} is non-zero, so $g_t(z)=g_t(z')$ for $(Q^F\otimes Q^F)$-almost every such pair. By Fubini's theorem there exists $z_0'\in P$ such that $g_t(z)=g_t(z_0')=:\gamma$ for $Q^F$-almost every $z\in N$, and likewise there exists $z_0\in N$ such that $g_t(z')=g_t(z_0)=\gamma$ for $Q^F$-almost every $z'\in P$. Finally, for $w$ in the remaining level set $\{k=\kappa\}$ we have $k(w)\ne k(z)$ for $z\in N$, so the same argument gives $g_t(w)=\gamma$ almost surely there as well. Hence $g_t=\gamma$ $Q^F$-almost surely. Since $Q^F$ is equivalent to the law of $X$ and therefore has support $[a,b]$, and since $x\mapsto C_X(x,t)$ is continuous on $(a,b)$, we conclude
\begin{equation}
\label{eq:Cconstant}
C_X(x,t)=\gamma(t)
\qquad\text{for every } x\in(a,b) \text{ and every } t\in[0,\,b-a).
\end{equation}

\medskip
\noindent\emph{Step 3: identification.} Suppose first that $k=h_X^{-1}$ is $Q^F$-almost surely constant. Then $h_X$ is constant on $(a,b)$, say $h_X\equiv\theta>0$, whence $S_X(x)=S_X(x_0)\exp\big(-\theta(x-x_0)\big)$ for $x\in(a,b)$. Letting $x\downarrow a$ and using $S_X(a^+)=1$ forces $a>-\infty$, and then $S_X(x)=\exp(-\theta(x-a))$, that is, $X-a\sim\Expo(\theta)$; in particular $b=\infty$.

Suppose instead that $k$ is not constant, so that \eqref{eq:Cconstant} holds. We claim $a=-\infty$. If not, pick $a<x_1<x_2<b$ and choose
\[
t\in\big[\,x_1-a,\ x_2-a\,\big)\cap\big[0,\,b-a\big),
\]
which is non-empty because $0\le x_1-a<x_2-a$ and $x_1-a<b-a$. For such $t$ we have $x_1-t\le a$, so $F_X(x_1-t)=0$ and $C_X(x_1,t)=0$, whereas $x_2-t>a$, so $F_X(x_2-t)>0$ and $C_X(x_2,t)>0$. This contradicts \eqref{eq:Cconstant}, and therefore $a=-\infty$.

With $a=-\infty$, \eqref{eq:Cconstant} reads $F_X(x-t)=\gamma(t)F_X(x)$ for all $x<b$ and all $t\ge0$. Writing $\psi=\log F_X$, which is differentiable almost everywhere with $\psi'=r_X$, we obtain $\psi(x-t)-\psi(x)=\log\gamma(t)$, and differentiating in $x$ gives $r_X(x-t)=r_X(x)$ for almost every $x<b$ and every $t\ge0$. Fixing such an $x$ and letting $t$ range over $[0,\infty)$ shows that $r_X$ is constant on $(-\infty,b)$, say $r_X\equiv c>0$, whence $F_X(x)=K\exp(cx)$ there. Since $F_X(b^-)=1$, the endpoint $b$ must be finite and $K=\exp(-cb)$, so that $F_X(x)=\exp\big(c(x-b)\big)$ for $x\le b$, that is, $b-X\sim\Expo(c)$.

The two alternatives are exactly those stated, and the proof is complete.
\end{proof}

\begin{remark}
The two branches of Lemma~\ref{lem:equalitydichotomy} are mutually exclusive, since the first forces $b=\infty$ and the second $b<\infty$. Under the hypotheses of Theorem~\ref{thm:SDgeGMD} with $h_X$ decreasing, Proposition~\ref{prop:r_Xh_Ximplications} gives $b=\infty$, so the second branch is excluded and only the exponential law survives; symmetrically, under $r_X$ increasing, Proposition~\ref{prop:r_Xh_Ximplications} gives $a=-\infty$ and only the reflected exponential survives. Under the hypotheses of Theorem~\ref{thm:SDleGMD} no such constraint is available and both branches occur.
\end{remark}

\subsection{Proof of Theorem \ref{thm:SDleGMDstrict}}
\label{app:SDleGMDstrict}
\begin{proof}
\emph{Sufficiency.} If $h_X\equiv\theta$ is constant then, as in Step~3 of Lemma~\ref{lem:equalitydichotomy}, $X-a\sim\Expo(\theta)$ for a finite $a$, and Proposition~\ref{prop:expbenchmark} together with affine invariance gives $\SD[X]=\GMD[X]=\theta^{-1}$. Both hypotheses of the theorem are indeed satisfied in this case, since a constant hazard rate is increasing in the non-strict sense and the exponential reverse hazard rate $r_X(x)=\theta e^{-\theta(x-a)}/\big(1-e^{-\theta(x-a)}\big)$ is decreasing. If instead $r_X\equiv c$ is constant then $b-X\sim\Expo(c)$ for a finite $b$; here $r_X$ is decreasing in the non-strict sense and $h_X(x)=c e^{c(x-b)}/\big(1-e^{c(x-b)}\big)$ is increasing, so the hypotheses again hold, and equality follows from the reflection invariance of both measures together with Proposition~\ref{prop:expbenchmark}.

\medskip
\noindent\emph{Necessity.} Under the hypotheses of the theorem both $h_X$ and $r_X$ are monotone, so Lemma~\ref{lem:equalitydichotomy} applies: if $\SD[X]=\GMD[X]$ then $h_X$ or $r_X$ is constant, and $X$ is correspondingly a shifted exponential or a reflected shifted exponential. Contrapositively, if $X$ is of neither form then the inequality of Theorem~\ref{thm:SDleGMD} is strict.

\medskip
\noindent\emph{Sharpness.} The Laplace law of Example~\ref{ex:SDleGMD}\emph{(6)} shows that the dichotomy cannot be weakened to constancy on a half-line. There $h_X$ is constant on $[0,\infty)$ and $r_X$ is constant on $(-\infty,0]$, yet neither is constant on the whole of $\RR$, and indeed $\SD[X]=b\sqrt{2}<\tfrac32 b=\GMD[X]$ strictly. The global constancy established in Step~2 of Lemma~\ref{lem:equalitydichotomy}, rather than constancy on a sub-interval, is therefore essential to the argument.
\end{proof}

\subsection{Proof of Proposition \ref{prop:logConcave_f_X}}
\label{app:logConcave_f_X}
\begin{proof}
Suppose that the density function \(f_X(x)\) is log-concave. Define the functions
\[
g(x,t) = f_X(t)\,\mathbf{1}_{\{t \le x\}}
\quad \text{and} \quad
k(x,t) = f_X(t)\,\mathbf{1}_{\{t \ge x\}},
\qquad (x,t)\in\RR^2,
\]
where \(\mathbf{1}\) denotes the indicator function. We first show that both \(g\) and \(k\) are jointly log-concave on $\RR^2$.

The map $(x,t)\mapsto f_X(t)$ is log-concave on $\RR^2$, being the composition of the log-concave function $f_X$ with a linear projection. The set $\{(x,t):t\le x\}$ is a closed half-plane, hence convex, so its indicator is log-concave on $\RR^2$. Since a product of log-concave functions is log-concave, $g$ is jointly log-concave; the same argument applied to the half-plane $\{(x,t):t\ge x\}$ gives the joint log-concavity of $k$.

Explicitly, for \(x_1,x_2,t_1,t_2 \in \mathbb{R}\) and \(\lambda \in (0,1)\),
\begin{align*}
g\big(\lambda(x_1,t_1) + (1-\lambda)(x_2,t_2)\big)
&= f_X\big(\lambda t_1 + (1-\lambda)t_2\big)\,
\mathbf{1}_{\{\lambda t_1 + (1-\lambda)t_2 \le \lambda x_1 + (1-\lambda)x_2\}}
\\
&\ge f_X(t_1)^{\lambda} f_X(t_2)^{1-\lambda}\,
\mathbf{1}_{\{\lambda t_1 + (1-\lambda)t_2 \le \lambda x_1 + (1-\lambda)x_2\}}
\\
&\ge f_X(t_1)^{\lambda} f_X(t_2)^{1-\lambda}\,
\mathbf{1}_{\{t_1 \le x_1\}}\,
\mathbf{1}_{\{t_2 \le x_2\}}
\\
&= \left(f_X(t_1)\,\mathbf{1}_{\{t_1 \le x_1\}}\right)^{\lambda}
   \left(f_X(t_2)\,\mathbf{1}_{\{t_2 \le x_2\}}\right)^{1-\lambda}
\\
&= g(x_1,t_1)^{\lambda}\, g(x_2,t_2)^{1-\lambda},
\end{align*}
where the first inequality is the log-concavity of $f_X$ and the second uses the implication $\{t_1\le x_1\}\cap\{t_2\le x_2\}\subseteq\{\lambda t_1+(1-\lambda)t_2\le\lambda x_1+(1-\lambda)x_2\}$, valid because $\lambda\in(0,1)$.

Now fix \(x_1,x_2 \in \mathbb{R}\) and \(\lambda \in (0,1)\), and apply the Pr\'ekopa--Leindler inequality \citep{Prekopa1971} to the three non-negative measurable functions
\[
t\mapsto g\big(\lambda x_1+(1-\lambda)x_2,\,t\big),
\qquad
t\mapsto g(x_1,t),
\qquad
t\mapsto g(x_2,t),
\]
whose required pointwise hypothesis is precisely the joint log-concavity just established. This yields
\begin{align*}
F_X\big(\lambda x_1 + (1-\lambda)x_2\big)
&= \int g\big(\lambda x_1 + (1-\lambda)x_2, t\big)\,\mathrm{d}t
\\
&\ge \left(\int g(x_1,t)\,\mathrm{d}t\right)^{\lambda}
     \left(\int g(x_2,t)\,\mathrm{d}t\right)^{1-\lambda}
\\
&= F_X(x_1)^{\lambda}\, F_X(x_2)^{1-\lambda},
\end{align*}
and hence \(F_X(x)\) is log-concave. The identical argument applied to \(k\), for which $\int k(x,t)\,\d t=S_X(x)$, shows that \(S_X(x)\) is log-concave. The boundary cases $\lambda\in\{0,1\}$ are trivial.

Finally, by Proposition~\ref{prop:equivalences}, the log-concavity of both \(F_X(x)\) and \(S_X(x)\) is equivalent to the hazard rate \(h_X(x)\) being increasing and the reverse hazard rate \(r_X(x)\) being decreasing. By Theorem~\ref{thm:SDleGMD}, these properties imply the ordering
\[
\SD[X] \le \GMD[X].
\]
This completes the proof. We note that the statement of this proposition is the classical result of \cite{Ibragimov1956} and \cite{Prekopa1971}, catalogued as Theorems~1 and~2 of \cite{Bagnoli2005}; see also \cite{An1998} for the corresponding log-convex theory, and Example~\ref{ex:counterexample_f_X_notlogconcave} for a demonstration that the implication is strict.
\end{proof}

\subsection{Proof of Proposition \ref{prop:f_Xconvolutions}}
\label{app:f_Xconvolutions}
\begin{proof}
The proof follows the same strategy as that of Proposition~\ref{prop:logConcave_f_X}.
Define the auxiliary function
\[
g(s,x)=f_{X_1}(x)\,f_{X_2}(s-x), \qquad (s,x)\in\mathbb{R}^2.
\]
We first show that $g(s,x)$ is jointly log-concave. The map $(s,x)\mapsto f_{X_1}(x)$ is log-concave on $\RR^2$, and so is $(s,x)\mapsto f_{X_2}(s-x)$, being the composition of the log-concave function $f_{X_2}$ with the linear map $(s,x)\mapsto s-x$; their product is therefore log-concave. Explicitly, for $s_1,s_2,x_1,x_2\in\mathbb{R}$ and $\lambda\in(0,1)$,
\begin{align*}
g\!\left(\lambda(s_1,x_1)+(1-\lambda)(s_2,x_2)\right)
&= f_{X_1}\!\left(\lambda x_1+(1-\lambda)x_2\right)
   f_{X_2}\!\left(\lambda(s_1-x_1)+(1-\lambda)(s_2-x_2)\right)
\\
&\ge f_{X_1}(x_1)^{\lambda} f_{X_1}(x_2)^{1-\lambda}
     f_{X_2}(s_1-x_1)^{\lambda} f_{X_2}(s_2-x_2)^{1-\lambda}
\\
&=\bigl(g(s_1,x_1)\bigr)^{\lambda}\,
   \bigl(g(s_2,x_2)\bigr)^{1-\lambda},
\end{align*}
both inequalities being instances of the log-concavity of the respective densities.

Next, fix $s_1,s_2\in\mathbb{R}$ and $\lambda\in(0,1)$. Since $X_1$ and $X_2$ are independent, $f_{X_1+X_2}(s)=\int g(s,x)\,\mathrm{d}x$, and the Pr\'ekopa--Leindler inequality applied to the three functions $x\mapsto g(\lambda s_1+(1-\lambda)s_2,x)$, $x\mapsto g(s_1,x)$ and $x\mapsto g(s_2,x)$ yields
\begin{align*}
f_{X_1+X_2}\!\left(\lambda s_1+(1-\lambda)s_2\right)
&=\int g\!\left(\lambda s_1+(1-\lambda)s_2,x\right)\,\mathrm{d}x
\\
&\ge \left(\int g(s_1,x)\,\mathrm{d}x\right)^{\lambda}
     \left(\int g(s_2,x)\,\mathrm{d}x\right)^{1-\lambda}
\\
&= f_{X_1+X_2}(s_1)^{\lambda}\,
   f_{X_1+X_2}(s_2)^{1-\lambda}.
\end{align*}
This shows that $f_{X_1+X_2}(s)$ is log-concave, completing the proof. The result is due to \cite{Ibragimov1956}; see also \cite{Prekopa1971} and \cite{Bagnoli2005}.
\end{proof}
\subsection{Proof of Proposition \ref{prop:tailSDgeGMD}}
\label{app:tailSDgeGMD}
\begin{proof}
We prove part~\emph{(i)}; part~\emph{(ii)} follows by applying \emph{(i)} to the reflected variable $-X$ and invoking Remark~\ref{rmk:h_Xr_X_Mirror}, both dispersion measures being invariant under reflection.

Fix $u\ge u^*$ and let $x\ge u$. Since $\{X>x\}\subseteq\{X>u\}$, the survival function of $X_u^+$ is
\[
S_{X_u^+}(x)=\mathbb{P}(X>x\mid X>u)=\dfrac{S_X(x)}{S_X(u)},
\]
and differentiating gives $f_{X_u^+}(x)=f_X(x)/S_X(u)$. Consequently
\[
h_{X_u^+}(x)
=-\big(\log S_{X_u^+}(x)\big)'
=-\big(\log S_X(x)-\log S_X(u)\big)'
=\dfrac{f_X(x)}{S_X(x)}
=h_X(x),
\qquad x\ge u ,
\]
the normalizing constant $S_X(u)$ being independent of $x$. The hazard rate is therefore invariant under lower truncation, and since $h_X$ is decreasing on $[u^*,\infty)\supseteq[u,\infty)$, the function $h_{X_u^+}$ is decreasing on the support $(u,\infty)$ of $X_u^+$.

Applying Theorem~\ref{thm:SDgeGMD} to the random variable $X_u^+$ yields $\SD[X_u^+]\ge\GMD[X_u^+]$. The strictness assertion follows from Theorem~\ref{thm:SDgeGMDstrict}, applied likewise to $X_u^+$: the inequality is strict unless $h_{X_u^+}$ is constant, that is, unless $X_u^+$ is exponential up to a location shift.
\end{proof}

\subsection{Proof of Proposition \ref{prop:SDleGMDtail}}
\label{app:SDleGMDtail}
\begin{proof}
We prove \emph{(i)}; part \emph{(ii)} follows by applying \emph{(i)} to $-X$ and invoking Remark~\ref{rmk:h_Xr_X_Mirror}.

\medskip
\noindent\emph{Heredity.} Write $V=X_{u^*}^+$, with distribution function $G$, density $g$ and reverse hazard rate $\rho=g/G$ on $(u^*,\infty)$, and fix $u\ge u^*$. Since $\{X>u\}\subseteq\{X>u^*\}$, the variable $X_u^+$ coincides in law with $V_u^+$. By \eqref{eq:hazardinvariance} its hazard rate is $h_{X_u^+}(x)=h_X(x)$ for $x>u$, which is increasing by hypothesis. Its distribution function is $F_{X_u^+}(x)=\big(G(x)-G(u)\big)/\big(1-G(u)\big)$ for $x\ge u$, whence
\[
r_{X_u^+}(x)=\frac{g(x)}{G(x)-G(u)}=\rho(x)\,\frac{G(x)}{G(x)-G(u)},\qquad x>u .
\]
The first factor is positive and non-increasing by hypothesis. The second equals $1+G(u)/\big(G(x)-G(u)\big)$, which is positive and non-increasing because $G$ is non-decreasing. A product of two positive non-increasing functions is non-increasing, so $r_{X_u^+}$ is decreasing. Hence $X_u^+$ satisfies the hypotheses of Theorem~\ref{thm:SDleGMD}, which gives $\SD[X_u^+]\le\GMD[X_u^+]$.

\medskip
\noindent\emph{Sufficient conditions.} If $h_X$ is increasing and $r_X$ is decreasing on $[u^*,\infty)$, then $h_V=h_X$ is increasing on $(u^*,\infty)$ and the same factorization, with $u^*$ in place of $u$ and $F_X$ in place of $G$, gives $\rho(x)=r_X(x)\,F_X(x)/\big(F_X(x)-F_X(u^*)\big)$, a product of two positive non-increasing functions. If instead $f_X$ is log-concave on $[u^*,\infty)$, then $g=f_X/S_X(u^*)$ is log-concave on the support of $V$, a concave function restricted to a convex set remaining concave, and Proposition~\ref{prop:logConcave_f_X} applied to $V$ shows that $h_V=h_X$ is increasing and $\rho$ is decreasing on $(u^*,\infty)$. In either case $V$ satisfies the hypotheses of Theorem~\ref{thm:SDleGMD}.
\end{proof}

\subsection{Discrete form of Proposition \ref{prop:r_Xh_Ximplications}}
\label{app:discreteRHimplications}
\begin{proof}
Let $X$ be integer valued with support contained in $\mathbb{Z}$.

\medskip
\noindent\emph{Monotonicity.} From the definitions,
\[
r_X(x)=h_X(x)\,\dfrac{S_X(x-1)}{F_X(x)}
\qquad\text{and}\qquad
h_X(x)=r_X(x)\,\dfrac{F_X(x)}{S_X(x-1)} .
\]
The ratio $S_X(x-1)/F_X(x)$ is non-negative and decreasing in $x$, being the quotient of a decreasing and an increasing non-negative function, and its reciprocal is non-negative and increasing. Hence a decreasing $h_X$ yields a decreasing $r_X$, and an increasing $r_X$ yields an increasing $h_X$, exactly as in the continuous case.

\medskip
\noindent\emph{Support.} Suppose $h_X$ is decreasing. If the support had a finite maximum $b$, then $h_X(b)=\mathbb{P}(X=b\mid X\ge b)=1$, whereas $h_X(x)<1$ for every $x<b$ in the support, since $\mathbb{P}(X>x)>0$ there; this contradicts the assumed monotonicity, so the support is unbounded above. If the support were unbounded below, then $S_X(x-1)\to1$ and $p_X(x)\to0$ as $x\to-\infty$, so that $h_X(x)\to0$; a decreasing function converging to $0$ at $-\infty$ is identically zero, which is impossible for a PMF. Hence the support is bounded below. The statements for an increasing $r_X$ follow by reflection.
\end{proof}

\subsection{Proof of Proposition \ref{prop:orderNonNegativeDiscrete}}
\label{app:orderNonNegativeDiscrete}
\begin{proof}
We begin by expanding the second moment of $Y$, using the elementary identity $y^2=\sum_{w=1}^{y}(2w-1)$ and interchanging the order of summation, which is legitimate throughout since all terms are non-negative:
\begin{align*}
\dfrac{1}{2}\,\E[Y^2]
&=\dfrac{1}{2}\sum_{y=1}^{\infty} y^2 p_Y(y) \\
&=\dfrac{1}{2}\sum_{y=1}^{\infty}\sum_{w=1}^{y} (2w-1)\, p_Y(y) \\
&=\dfrac{1}{2}\sum_{w=1}^{\infty} (2w-1)\, \mathbb{P}(Y\ge w) \\
&=\dfrac{1}{2}\sum_{w=1}^{\infty} (2w-1)\, S_Y(w-1) \\
&=\sum_{w=1}^{\infty} w\, S_Y(w-1)
   -\dfrac{1}{2}\sum_{w=1}^{\infty} S_Y(w-1) \\
&=\sum_{w=1}^{\infty}\sum_{z=1}^{w} S_Y(w-1)
   -\dfrac{1}{2}\sum_{w=1}^{\infty} S_Y(w-1) \\
&=\sum_{z=1}^{\infty}\sum_{w=z}^{\infty} S_Y(w-1)
   -\dfrac{1}{2}\sum_{w=1}^{\infty} S_Y(w-1) \\
&=\sum_{z=1}^{\infty} m_Y(z-1)\, S_Y(z-1)
   -\dfrac{1}{2}\sum_{w=1}^{\infty} S_Y(w-1).
\end{align*}
In the final step we used the discrete representation of the mean excess function,
\[
m_Y(t)\, S_Y(t)=\E\big[(Y-t)^{+}\big]=\sum_{k=t}^{\infty}S_Y(k)=\sum_{w=t+1}^{\infty} S_Y(w-1),
\]
valid for every integer $t\ge0$.

Observe next that
\[
\sum_{w=1}^{\infty}S_Y(w-1)=\sum_{k=0}^{\infty}S_Y(k)=\E[Y] .
\]
Assume now that $m_Y(t)\ge(\le)\,\E[Y]+\tfrac12$ for every integer $t\ge0$. Substituting into the display above gives
\begin{align*}
\dfrac{1}{2}\,\E[Y^2]
&=\sum_{z=1}^{\infty} m_Y(z-1)\, S_Y(z-1)
   -\dfrac{1}{2}\sum_{w=1}^{\infty} S_Y(w-1) \\
&\ge (\le)\ \left(\E[Y] + \dfrac{1}{2}\right)
   \sum_{z=1}^{\infty} S_Y(z-1)
   -\dfrac{1}{2}\sum_{w=1}^{\infty} S_Y(w-1) \\
&=\E[Y]\sum_{w=1}^{\infty} S_Y(w-1) \\
&=\E[Y]^2 .
\end{align*}
Taking square roots of the resulting inequality $\tfrac12\E[Y^2]\ge(\le)\,\E[Y]^2$, both sides being non-negative, yields the assertion.

We note where the lattice correction enters. Had the condition been stated with $\E[Y]$ in place of $\E[Y]+\tfrac12$, the second term $-\tfrac12\sum_w S_Y(w-1)=-\tfrac12\E[Y]$ would have survived and produced the weaker conclusion $\tfrac12\E[Y^2]\ge\E[Y]^2-\tfrac12\E[Y]$. The constant $\tfrac12$ is thus exactly the amount required to absorb the discrete correction, and it is the same constant that appears in Theorem~\ref{thm:SDGMDDiscrete}.
\end{proof}

\subsection{Proof of Proposition \ref{prop:m_Y_RepresentationDiscrete}}
\label{app:m_Y_RepresentationDiscrete}
\begin{proof}
The argument parallels that of Proposition~\ref{prop:m_Y_Representation}. For an integer $y\ge0$ the events $\{X-X'>y\}$ and $\{X'-X>y\}$ are disjoint and equiprobable, so
\[
S_Y(y)=\mathbb{P}(|X-X'|>y)=2\,\mathbb{P}(X-X'>y),
\]
and conditioning on one variable at a time gives the two equivalent forms
\begin{equation}
\label{eq:SYtwoformsDiscrete}
S_Y(y)
=2\sum_{x=-\infty}^{\infty} S_X(x+y)\, p_X(x)
=2\sum_{x=-\infty}^{\infty} F_X(x-y-1)\, p_X(x),
\end{equation}
the second form using $\mathbb{P}(X'<x-y)=F_X(x-y-1)$ for integer-valued $X'$.

\medskip
\noindent\emph{First representation.} Using the first expression in \eqref{eq:SYtwoformsDiscrete},
\begin{align*}
m_Y(t)\, S_Y(t)
&=\sum_{y=t+1}^{\infty} S_Y(y-1) \\
&=2\sum_{y=t+1}^{\infty}\sum_{x=-\infty}^{\infty}
   S_X(x+y-1)\, p_X(x) \\
&=2\sum_{x=-\infty}^{\infty}\sum_{w=x+t+1}^{\infty}
   S_X(w-1)\, p_X(x) \\
&=2\sum_{w=-\infty}^{\infty}\sum_{x=-\infty}^{w-t-1}
   S_X(w-1)\, p_X(x) \\
&=2\sum_{w=-\infty}^{\infty}
   F_X(w-t-1)\, S_X(w-1) \\
&=2\sum_{w=-\infty}^{\infty}
   C_X(w-1,t)\, h_X(w)^{-1}\, F_X(w-1)\, p_X(w) \\
&=2\, \E[F_X(X-1)]\,
   \E^{F}\!\left[C_X(X-1,t)\, h_X(X)^{-1}\right],
\end{align*}
where the third equality substitutes $w=x+y$, the fourth interchanges the order of summation over $\{(x,w):x\le w-t-1\}$, and the sixth uses the two identities
\[
F_X(w-t-1)=C_X(w-1,t)\,F_X(w-1)
\qquad\text{and}\qquad
S_X(w-1)=h_X(w)^{-1}p_X(w),
\]
the latter being the definition of the discrete hazard rate. Using the second expression in \eqref{eq:SYtwoformsDiscrete} for the left-hand factor gives, by the same substitution,
\begin{align*}
S_Y(t)
&=2\sum_{x=-\infty}^{\infty}F_X(x-t-1)p_X(x)
=2\sum_{x=-\infty}^{\infty}C_X(x-1,t)F_X(x-1)p_X(x)
\\
&=2\, \E[F_X(X-1)]\, \E^{F}[C_X(X-1,t)].
\end{align*}
Dividing the two displays yields the first representation in \eqref{eq:m_YRepresentationDiscrete}.

\medskip
\noindent\emph{Second representation.} Interchanging the roles of the two expressions in \eqref{eq:SYtwoformsDiscrete} and noting that $S_Y(y-1)=2\sum_x F_X(x-y)p_X(x)$, we obtain
\begin{align*}
m_Y(t)\, S_Y(t)
&=\sum_{y=t+1}^{\infty} S_Y(y-1) \\
&=2\sum_{y=t+1}^{\infty}\sum_{x=-\infty}^{\infty}F_X(x-y)\,p_X(x) \\
&=2\sum_{x=-\infty}^{\infty}\sum_{w=-\infty}^{x-t-1}F_X(w)\,p_X(x) \\
&=2\sum_{w=-\infty}^{\infty}F_X(w)\sum_{x=w+t+1}^{\infty}p_X(x) \\
&=2\sum_{w=-\infty}^{\infty}F_X(w)\,S_X(w+t) \\
&=2\sum_{w=-\infty}^{\infty}D_X(w,t)\,r_X(w)^{-1}\,S_X(w)\,p_X(w) \\
&=2\,\E[S_X(X)]\,\E^{S}\!\left[D_X(X,t)\,r_X(X)^{-1}\right],
\end{align*}
using $w=x-y$ in the third equality, $S_X(w+t)=D_X(w,t)S_X(w)$ and $F_X(w)=r_X(w)^{-1}p_X(w)$ in the sixth. Correspondingly,
\[
S_Y(t)=2\sum_{x=-\infty}^{\infty}S_X(x+t)p_X(x)
=2\sum_{x=-\infty}^{\infty}D_X(x,t)S_X(x)p_X(x)
=2\,\E[S_X(X)]\,\E^{S}[D_X(X,t)],
\]
and dividing gives the second representation.

\medskip
The asymmetry between the two representations, the first carrying the shift $-1$ in the argument of $C_X$ and the second carrying none, is now transparent: it originates in the identity $S_X(w-1)=h_X(w)^{-1}p_X(w)$, in which the discrete hazard rate conditions on $\{X\ge w\}$, against $F_X(w)=r_X(w)^{-1}p_X(w)$, in which the reverse hazard rate conditions on $\{X\le w\}$. The two constructions allocate the atom at $w$ differently, and the shift is the bookkeeping that results.
\end{proof}

\subsection{Proof of Theorem \ref{thm:SDGMDDiscrete}}
\label{app:SDGMDDiscrete}
\begin{proof}
Throughout we use the discrete analogue of Lemma~\ref{lem:chebyshev}, whose proof is identical with integrals replaced by sums.

\medskip
\noindent\emph{Part \emph{(i)}.} Recall from Proposition~\ref{prop:m_Y_RepresentationDiscrete} that
\[
m_Y(t)
=\dfrac{\E^{F}[C_X(X-1,t)\, h_X(X)^{-1}]}
       {\E^{F}[C_X(X-1,t)]} .
\]
Assume $h_X(x)$ is decreasing. By Appendix~A.\ref{app:discreteRHimplications} the reverse hazard rate $r_X(x)$ is then decreasing as well, so by Proposition~\ref{prop:equivalences} the function $x\mapsto C_X(x-1,t)$ is increasing, as is $h_X(x)^{-1}$. Lemma~\ref{lem:chebyshev}\emph{(i)} applied under $Q^{F}$ therefore yields
\begin{equation}
\label{eq:discreteNWUE}
m_Y(t)\ \ge\ \E^{F}[h_X(X)^{-1}]\ =\ m_Y(0),\qquad \text{for all integers } t\ge 0,
\end{equation}
the identification of $m_Y(0)$ following from $C_X(x-1,0)=1$.

It remains to compare $m_Y(0)$ with $\E[Y]+\tfrac12$. Since $S_Y(0)=\mathbb{P}(X\ne X')=1-\Lambda$ and $\E[Y]=m_Y(0)S_Y(0)$, we have
\begin{equation}
\label{eq:thresholdequiv}
m_Y(0)>\E[Y]+\dfrac{1}{2}
\quad \iff \quad
m_Y(0)\big(1-(1-\Lambda)\big)>\dfrac12
\quad \iff \quad
m_Y(0)>\dfrac{1}{2\Lambda}.
\end{equation}
To establish the right-hand inequality we expand $m_Y(0)$ and apply Lemma~\ref{lem:chebyshev} a second time, now under the auxiliary probability measure
\[
\widetilde{Q}(x)=\dfrac{F_X(x)\,p_X(x)}{\E[F_X(X)]},
\]
with expectation written $\widetilde{\E}$. Using $F_X(x-1)=C_X(x,1)F_X(x)$ we obtain
\begin{align*}
m_Y(0)
&=\E^{F}[h_X(X)^{-1}] \\
&=\dfrac{1}{\E[F_X(X-1)]}
   \sum_{x=-\infty}^{\infty}
   h_X(x)^{-1}\, F_X(x-1)\, p_X(x) \\
&=\dfrac{\E[F_X(X)]}{\E[F_X(X-1)]}\,
   \widetilde{\E}\!\left[h_X(X)^{-1}\, C_X(X,1)\right] \\
&\ge
\dfrac{\E[F_X(X)]}{\E[F_X(X-1)]}\,
\widetilde{\E}\!\left[h_X(X)^{-1}\right]
\widetilde{\E}\!\left[C_X(X,1)\right] \\
&=\widetilde{\E}\!\left[h_X(X)^{-1}\right]
=\dfrac{1}{\E[F_X(X)]}
   \sum_{x=-\infty}^{\infty}
   h_X(x)^{-1}\, F_X(x)\, p_X(x),
\end{align*}
where the inequality is Lemma~\ref{lem:chebyshev}\emph{(i)}, both factors being increasing, and the penultimate equality uses
\[
\widetilde{\E}\!\left[C_X(X,1)\right]
=\dfrac{1}{\E[F_X(X)]}\sum_x \dfrac{F_X(x-1)}{F_X(x)}F_X(x)p_X(x)
=\dfrac{\E[F_X(X-1)]}{\E[F_X(X)]} .
\]
Decomposing $F_X(x)=p_X(x)+F_X(x-1)$ and using $h_X(x)^{-1}p_X(x)=S_X(x-1)$ gives
\begin{align*}
m_Y(0)
&\ge\dfrac{1}{\E[F_X(X)]}
   \left(
   \sum_{x} S_X(x-1)\, p_X(x)
   +\sum_{x} h_X(x)^{-1}\, F_X(x-1)\, p_X(x)
   \right) \\
&=\dfrac{1}{\E[F_X(X)]}
   \Big(
   \E[S_X(X-1)]
   +\E[F_X(X-1)]\, m_Y(0)
   \Big) \\
&=1+\dfrac{\E[F_X(X-1)]}{\E[F_X(X)]}\, m_Y(0),
\end{align*}
the last equality following from
\[
\E[S_X(X-1)]=\mathbb{P}(X'\ge X)
=\dfrac{1}{2}\bigl(1+\mathbb{P}(X=X')\bigr)
=\dfrac{1+\Lambda}{2}
=\mathbb{P}(X'\le X)=\E[F_X(X)] .
\]
Rearranging, and using $\E[F_X(X)]-\E[F_X(X-1)]=\E[p_X(X)]=\Lambda$, we obtain
\[
m_Y(0)\ \ge\
\dfrac{\E[F_X(X)]}{\E[F_X(X)]-\E[F_X(X-1)]}
=\dfrac{(1+\Lambda)/2}{\Lambda}
=\dfrac{1+\Lambda}{2\Lambda}
\ >\ \dfrac{1}{2\Lambda},
\]
the final inequality being strict because $\Lambda>0$ for any integer-valued random variable. Combining with \eqref{eq:discreteNWUE} and \eqref{eq:thresholdequiv},
\[
m_Y(t)\ \ge\ m_Y(0)\ >\ \E[Y]+\dfrac{1}{2}
\qquad\text{for all integers } t\ge0,
\]
and Proposition~\ref{prop:orderNonNegativeDiscrete} yields $\SD[X]>\GMD[X]$. The case of an increasing $r_X(x)$ follows by the reflection argument of Remark~\ref{rmk:h_Xr_X_Mirror}, or directly, on noting that both $h_X(x)^{-1}$ and $C_X(x-1,t)$ are then decreasing so that Lemma~\ref{lem:chebyshev}\emph{(i)} again applies.

\medskip
\noindent\emph{Part \emph{(ii)}.} Suppose now that $h_X(x)$ is increasing and $r_X(x)$ is decreasing. Then $h_X(x)^{-1}$ is decreasing while $C_X(x-1,t)$ is increasing, so Lemma~\ref{lem:chebyshev}\emph{(ii)} reverses the first inequality and gives
\[
m_Y(t)\le m_Y(0),\qquad \text{for all integers } t\ge 0 .
\]
By \eqref{eq:thresholdequiv} with the inequalities reversed, invoking Proposition~\ref{prop:orderNonNegativeDiscrete} requires $m_Y(0)\le1/(2\Lambda)$. Running the second Chebyshev argument of part~\emph{(i)} in the opposite direction, however, yields only
\[
m_Y(0)\le \dfrac{1+\Lambda}{2\Lambda}
=\dfrac{1}{2\Lambda}+\dfrac{1}{2},
\]
which overshoots the required threshold by exactly $\tfrac12$. The monotonicity of $h_X(x)$ and $r_X(x)$ alone is therefore insufficient, and this deficit is precisely the lattice correction identified in Appendix~A.\ref{app:orderNonNegativeDiscrete}.

Assume now in addition that
\[
\GMD[X]=\E[Y]\le \dfrac{1-\Lambda}{2\Lambda} .
\]
Since $S_Y(0)=1-\Lambda$, this is equivalent to
\[
m_Y(0)=\dfrac{\E[Y]}{S_Y(0)}\le \dfrac{1}{2\Lambda},
\]
and hence, by \eqref{eq:thresholdequiv}, to $m_Y(0)\le\E[Y]+\tfrac12$. Combining,
\[
m_Y(t)\le m_Y(0)\le \E[Y]+\dfrac{1}{2}
\qquad\text{for all integers } t\ge0,
\]
and Proposition~\ref{prop:orderNonNegativeDiscrete} implies $\SD[X]\le \GMD[X]$. Finally, elementary rearrangement shows the hypothesis to be equivalent to each of the forms recorded in \eqref{eq:concentrationforms}:
\[
\E[Y]\le\dfrac{1-\Lambda}{2\Lambda}
\iff 2\Lambda\,\E[Y]\le1-\Lambda
\iff \Lambda\big(1+2\E[Y]\big)\le1
\iff \dfrac{1}{\Lambda}\ge1+2\,\GMD[X] .
\]
This completes the proof.
\end{proof}
\end{appendices}
\end{document}